\pdfoutput=1
\documentclass{article}%
\usepackage{amsmath}
\usepackage{amsfonts}
\usepackage{amssymb}
\usepackage{graphicx}
\usepackage{lscape}
\usepackage{hyperref}%
\providecommand{\U}[1]{\protect\rule{.1in}{.1in}}
\newtheorem{theorem}{Theorem}
\newtheorem{acknowledgement}[theorem]{Acknowledgement}

\newenvironment{proof}[1][Proof]{\noindent\textbf{#1.} }{\ \rule{0.5em}{0.5em}}
\begin{document}

\title{\textbf{Comparing two treatments in terms of the likelihood ratio order}}
\author{Martin, N.$^{1}$, Mata, R.$^{2}$ and Pardo, L.$^{2}$\\$^{1}${\small Department of Statistics, Carlos III University of Madrid, 28903
Getafe (Madrid), Spain}\\$^{2}${\small Department of Statistics and O.R., Complutense University of
Madrid, 28040 Madrid, Spain} }
\date{\today}
\maketitle

\begin{abstract}
In this paper new families of test statistics are introduced and studied for
the problem of comparing two treatments in terms of the likelihood ratio
order. The considered families are based on phi-divergence measures and arise
as natural extensions of the classical likelihood ratio test and Pearson test
statistics. It is proven that their asymptotic distribution is a common
chi-bar random variable. An illustrative example is presented and the
performance of these statistics is analysed through a simulation study.
Through a simulation study it is shown that, for most of the proposed
scenarios adjusted to be small or moderate, some members of this new family of
test-statistic display clearly better performance with respect to the power in
comparison to the classical likelihood ratio and the Pearson's chi-square test
while the exact size remains closed to the nominal size. In view of the exact
powers and significance levels, the study also shows that the Wilcoxon
test-statistic is not as good as the two classical test-statistics.

\end{abstract}

\bigskip


\section{Introduction}

In order to motivate the problem dealt in this paper, we have considered the
results of an experiment carried out by Doll and Pygott (1952) to assess the
factors influencing the rate of healing of gastric ulcers. Two treatments
groups were compared. Patients in group 2 were treated in bed in hospital for
four weeks. For the first two weeks they were given a moderate strict orthodox
diet and for the last two weeks a more liberal one. They were then reexamined
radiographically, discharged, recommended to continue on a convalescent diet
and advised return to work as soon as they felt fit enough. Patients in group
1 were discharged immediately. They were treated from the outset in the way
that group 2 patients were treated after their month's stay in hospital. In
Table \ref{tttt1}, we present the results showed by Doll and Pygott (1952,
Table IV) for three months after starting the treatments. This article
proposes new families of test-statistics when we are interested in studying
the possibility that the ulcer treatment (Treatment $2$) is better than the
control (Treatment $1$).%

\begin{table}[htbp]  \tabcolsep2.8pt  \centering
$%
\begin{tabular}
[c]{lcccc}\hline
& Larger & $<\frac{1}{3}$ Healed & $\geq\frac{2}{3}$ Healed & Healed\\\hline
Treatment $1$ & 11 & 8 & 8 & 5\\
Treatment $2$ & 6 & 4 & 10 & 12\\\hline
\end{tabular}
\ \ \ \ \ \ \ \ \ \ \ $\caption{Change in size of ulcer crater.\label{tttt1}}%
\end{table}%

Let $Y$ denote the ordinal response variable and $X$ denote an ordinal
explanatory variable with two categories. The variable $Y$ takes the values
$1$, $2$, $3$ and $4$, which represent different levels of healing, from less
to much capacity to heal the ulcer. The variable $X$ takes the values $1$ and
$2$ according as the treatment group, $1$ is control and $2$ is the treatment
group by itself. We shall initially focus on making statistical inference on
the theoretical probabilities displayed in Table \ref{ttt2}.%

\begin{table}[htbp]  \tabcolsep2.8pt  \centering
$%
\begin{tabular}
[c]{lcccc}\hline
& Larger & $<\frac{1}{3}$ Healed & $\geq\frac{2}{3}$ Healed & Healed\\\hline
Treatment $1$ & $\Pr(Y=1|X=1)$ & $\Pr(Y=2|X=1)$ & $\Pr(Y=3|X=1)$ &
$\Pr(Y=4|X=1)$\\
Treatment $2$ & $\Pr(Y=1|X=2)$ & $\Pr(Y=2|X=2)$ & $\Pr(Y=3|X=2)$ &
$\Pr(Y=4|X=2)$\\\hline
\end{tabular}
\ \ \ \ \ \ \ \ \ \ \ \ $%
\caption{Theoretical conditional probabilities.\label{ttt2}}%
\end{table}%

There are several ways of formulating the statement \textquotedblleft the
treatment is better than the control\textquotedblright. Initially, we shall
consider that Treatment $2$ is at least as good as Treatment $1$ if the ratio
$\frac{\Pr(Y=j|X=2)}{\Pr(Y=j|X=1)}$ increases as the response category, $j$,
increases, i.e.%
\begin{equation}
\tfrac{\Pr(Y=j|X=2)}{\Pr(Y=j|X=1)}\leq\tfrac{\Pr(Y=j+1|X=2)}{\Pr
(Y=j+1|X=1)}\qquad\text{for every }j\text{,} \label{eq1}%
\end{equation}
and Treatment 2 is better than the Treatment 1 if (\ref{eq1}) holds with at
least one strict inequality.

If we assume that Treatment 2 is at least as good as Treatment 1, i.e.,
(\ref{eq1}) holds, is there any evidence to support the claim that treatment
$2$ is better? In such a case null and alternative hypotheses may be%
\begin{subequations}
\begin{align}
&  H_{0}:\;\tfrac{\Pr(Y=j|X=2)}{\Pr(Y=j|X=1)}=\tfrac{\Pr(Y=j+1|X=2)}%
{\Pr(Y=j+1|X=1)}\quad\text{for every }j\text{,}\label{eq2}\\
&  H_{1}:\;\tfrac{\Pr(Y=j|X=2)}{\Pr(Y=j|X=1)}\leq\tfrac{\Pr(Y=j+1|X=2)}%
{\Pr(Y=j+1|X=1)}\quad\text{for every }j\quad\text{and}\quad\tfrac
{\Pr(Y=j|X=2)}{\Pr(Y=j|X=1)}<\tfrac{\Pr(Y=j+1|X=2)}{\Pr(Y=j+1|X=1)}%
\quad\text{for at least one }j\text{.} \label{eq3}%
\end{align}
The null hypothesis means that both treatments are equally effective, while
the alternative hypothesis means that Treatment 2 is more effective than
Treatment 1. Note that if we multiply on the left and right hand side of
(\ref{eq2}) and (\ref{eq3}) by $\left(  \tfrac{\Pr(Y=j|X=2)}{\Pr
(Y=j|X=1)}\right)  ^{-1}$ we obtain
\end{subequations}
\begin{subequations}
\begin{align}
&  H_{0}:\;\vartheta_{j}=1\quad\text{for every }j\in\{1,...,J-1\}\text{,}%
\label{eq2b}\\
&  H_{1}:\;\vartheta_{j}\geq1\quad\text{for every }j\in\{1,...,J-1\}\quad
\text{and}\quad\vartheta_{j}>1\quad\text{for at least one }j\in
\{1,...,J-1\}\text{,} \label{eq3b}%
\end{align}
where $J$ is the number of ordered categories for response variable $Y$,%
\end{subequations}
\begin{equation}
\vartheta_{j}=\dfrac{\pi_{1j}\pi_{2,j+1}}{\pi_{2j}\pi_{1,j+1}},\quad\forall
j\in\{1,...,J-1\}, \label{2}%
\end{equation}
are \textquotedblleft local odds ratios\textquotedblright\ associated with
response category $j$, and%
\begin{equation}
\pi_{ij}=\Pr(Y=j|X=i). \label{eq5}%
\end{equation}
In case of considering the opposite inequalities given in (\ref{eq3}) or
(\ref{eq3b}), the easiest way to carry out the test is to exchange the
observation of the two rows in the contingency table (in the example,
Treatment $2$ in the first row and Treatment $1$ in the second row). In this
way, the mathematical background is not changed but the interpretation of the
aim is changed. In the example however, there is no sense in considering that
the control ($1$) is better than the treatment ($2$), if the experiment is
carried out with humans and it is assumed that the treatment will not harm
these patients.

The non-parametric statistical inference associated with the likelihood ratio
ordering for two multinomial samples was introduced for the first time in
Dykstra et al. (1995) using the likelihood ratio test-statistic. In the
literature related to different types of orderings, in general there is not
very clear what is the most appropriate ordering to compare two treatments
according to a categorized ordinal variable. In the case of having two
independent multinomial samples, the likelihood ratio ordering is the most
restricted ordering type; for example, if the likelihood ratio ordering holds,
then the simple stochastic ordering also holds. Dardanoni and Forcina (1998)
proposed a new method for making statistical inference associated with
different types of orderings. For unifying and comparing different types of
orderings, they reparametrize the initial model. Different ordering types can
be considered to be nested models and the likelihood ratio ordering is the
most parsimonious one. The advantage of nested models is that the most
restricted models tend to be more powerful for the alternatives that belong to
the most restricted alternatives. In this setting, our proposal in this paper
is to introduce new test-statistics that provide substantially better power
for testing (\ref{eq2}) against (\ref{eq3}).

The structure of the paper is as follows. In Section \ref{Sec:LM}, we have
considered the likelihood ratio order associated with a non-parametric model,
as in Dardanoni and Forcina (1998), but the specification of the model through
a saturated loglinear model is substantially different. Section \ref{Sec:PD}
presents the phi-divergence test-statistics as extension of the likelihood
ratio and chi-square test-statistics.\ The applied methodology in Section
\ref{sec:Main results} for proving the asymptotic distribution of the
phi-divergence test-statistics, based on loglinear modeling, has been
developed by following a completely new and meaningful method even for the
likelihood ratio test. A numerical example is given in Section
\ref{sec:Numerical example}. The aim of Section \ref{sec:Simulation Study} is
to study through simulation the behaviour of the phi-divergence
test-statistics for small and moderate simple sizes. Finally, we present an
Appendix in which we establish the part of the proofs of the results not shown
in Section \ref{sec:Main results}.

\section{Loglinear modeling\label{Sec:LM}}

We display the whole distribution of $\pi_{ij}$, given in (\ref{eq5}), in a
rectangular table having $2$ rows for the categories of $X$ and $J$ columns
for the categories of $Y$ (for the initial example, Table \ref{ttt2}) and we
denote the $2\times J$ matrix $\boldsymbol{\Pi}=(\boldsymbol{\pi}%
_{1},\boldsymbol{\pi}_{2})^{T}$, with two rows of probability vectors,
$\boldsymbol{\pi}_{i}=(\pi_{i1},...,\pi_{iJ})^{T}$, $i=1,2$. We consider two
independent random samples $\boldsymbol{N}_{i}=(N_{i1},...,N_{iJ})^{T}%
\sim\mathcal{M}(n_{i},\boldsymbol{\pi}_{i})$, $i=1,2$, where sizes $n_{i}$ are
prefixed and $\boldsymbol{\pi}_{i}>\boldsymbol{0}_{J}$, that is the
probability distribution of r.v. $\boldsymbol{N}=(\boldsymbol{N}_{1}%
^{T},\boldsymbol{N}_{2}^{T})^{T}$ is product-multinomial. Let%
\begin{equation}
p_{ij}=\Pr(X=i,Y=j), \label{eq5b}%
\end{equation}
be the joint probability distribution. Since $\Pr(X=i,Y=j)=\Pr(Y=j|X=i)\Pr
(X=i)$, i.e. $p_{ij}=\pi_{ij}\frac{n_{i}}{n}$, $i=1,2$, where $n=n_{1}+n_{2}$,
we can express (\ref{2}) also in terms of the joint probabilities%
\begin{equation}
\vartheta_{j}=\dfrac{p_{1j}p_{2,j+1}}{p_{2j}p_{1,j+1}},\quad\forall
j\in\{1,...,J-1\}. \label{2b}%
\end{equation}
Let $\boldsymbol{P}=(\mathbf{p}_{1},\mathbf{p}_{2})^{T}$, with $\mathbf{p}%
_{i}=(p_{i1},...,p_{iJ})^{T}$, $i=1,2$, be the $2\times J$ probability matrix
and
\begin{equation}
\boldsymbol{p}=\mathrm{vec}(\boldsymbol{P}^{T})=(\mathbf{p}_{1}^{T}%
,\mathbf{p}_{2}^{T})^{T} \label{0}%
\end{equation}
a probability vector obtained by stacking the columns of $\boldsymbol{P}^{T}%
$\ (i.e., the rows of matrix $\boldsymbol{P}$). Note that the components of
$\boldsymbol{P}$ are ordered in lexicographical order in $\boldsymbol{p}$. The
likelihood function of $\boldsymbol{N}$ is $\mathcal{L}(\boldsymbol{N}%
;\boldsymbol{p})=k%
{\textstyle\prod\nolimits_{j=1}^{J}}
p_{1j}^{N_{1j}}p_{2j}^{N_{2j}}$, where $k$ is a constant which does not depend
on $\boldsymbol{p}$ and the kernel of the loglikelihood function%
\begin{equation}
\ell(\boldsymbol{N};\boldsymbol{p})=%
{\displaystyle\sum\limits_{j=1}^{J}}
(N_{1j}\log p_{1j}+N_{2j}\log p_{2j}). \label{0b}%
\end{equation}

In matrix notation, we are interested in testing
\begin{equation}
H_{0}:\boldsymbol{\vartheta}=\boldsymbol{1}_{J-1}\text{ versus }%
H_{1}:\boldsymbol{\vartheta}\gneqq\boldsymbol{1}_{J-1}\text{,} \label{4}%
\end{equation}
where $\boldsymbol{1}_{a}$\ is the $a$-vector of $1$-s, $\boldsymbol{\vartheta
}=(\vartheta_{1},...,\vartheta_{J-1})^{T}$. Note that (\ref{4}) involves $J-1$
non-linear constraints on $\boldsymbol{p}$, defined by (\ref{0}). In this
article the hypothesis testing problem is formulated making a
reparametrization of $\boldsymbol{p}$ using the saturated loglinear model, so
that some linear restrictions are considered with respect to the new
parameters. This fact is important and interesting.

Focussed on $\boldsymbol{p}$, the saturated loglinear model with canonical
parametrization is defined by
\begin{equation}
\log p_{ij}=u+u_{1(i)}+\theta_{2(j)}+\theta_{12(ij)}, \label{3}%
\end{equation}
with the identifiabilty restrictions%
\begin{equation}
u_{1(2)}=0,\quad\theta_{2(J)}=0,\quad\theta_{12(1J)}=0,\quad\theta
_{12(2j)}=0,\quad j=1,...,J. \label{ident}%
\end{equation}
It is important to clarify that we have used the identifiability constraints
(\ref{ident}) in order to make easier the calculations and this model
formulation for making statistical inference with inequality restrictions with
local odds-ratios has been given in this paper for the first time. Similar
conditions have been used for instance in Lang (1996, examples of Section 7)
and Silvapulle and Sen (2005, exercise 6.25 in page 345). Let
$\boldsymbol{\theta}_{12}=(\theta_{12(11)},...,\theta_{12(1,J-1)})^{T}$,
$\boldsymbol{\theta}_{2}=(\theta_{2(1)},...,\theta_{2(J-1)})^{T}$ denote
subvectors of the unknown parameters $\boldsymbol{\theta}=(\boldsymbol{\theta
}_{2}^{T},\boldsymbol{\theta}_{12}^{T})^{T}$. The components of
$\boldsymbol{u}=(u,u_{1(1)})^{T}$ are redundant parameters since the term $u$
can be expressed in function of $\boldsymbol{\theta}$ using the fact that $%
{\textstyle\sum\nolimits_{j=1}^{J}}
p_{2j}=\frac{n_{2}}{n}$, i.e.%
\begin{equation}
u=u(\boldsymbol{\theta})=\log n_{2}-\log n-\log\left(  1+%
{\displaystyle\sum\limits_{j=1}^{J-1}}
\exp\{\theta_{2(j)}\}\right)  , \label{u}%
\end{equation}
and $u_{1(1)}$ taking into account that $%
{\textstyle\sum\nolimits_{j=1}^{J}}
p_{1j}=\frac{n_{1}}{n}$, i.e.%
\begin{equation}
u_{1(1)}=u_{1(1)}(\boldsymbol{\theta})=\log\frac{n_{1}}{n_{2}}+\log\frac{1+%
{\textstyle\sum\nolimits_{j=1}^{J-1}}
\exp\{\theta_{2(j)}\}}{1+%
{\displaystyle\sum\limits_{j=1}^{J-1}}
\exp\{\theta_{2(j)}+\theta_{12(1j)}\}}. \label{u1}%
\end{equation}
In matrix notation (\ref{3}) is given by%
\begin{equation}
\log\boldsymbol{p}(\boldsymbol{\theta})=\boldsymbol{W}_{0}\boldsymbol{u}%
+\boldsymbol{W\theta}, \label{loglin}%
\end{equation}
where $\boldsymbol{p}(\boldsymbol{\theta})$ is $\boldsymbol{p}$\ such that the
components are defined by (\ref{3}),
\[
\boldsymbol{W}_{0}=%
\begin{pmatrix}
1 & 1\\
1 & 0
\end{pmatrix}
\otimes\boldsymbol{1}_{J}%
\]
is a $2J\times2$ matrix with $\boldsymbol{1}_{a}$\ being the $a$-vector of
ones, $\boldsymbol{0}_{a}$\ the $a$-vector of zeros, $\otimes$ the Kronecker
product; $\boldsymbol{W}$ the full rank design matrix of size $2J\times
2(J-1)$, such that%
\begin{equation}
\boldsymbol{W}=%
\begin{pmatrix}
1 & 1\\
1 & 0
\end{pmatrix}
\otimes%
\begin{pmatrix}
\boldsymbol{I}_{J-1}\\
\boldsymbol{0}_{J-1}^{T}%
\end{pmatrix}
, \label{W}%
\end{equation}
with $\boldsymbol{I}_{a}$\ being the identity matrix of order $a$,
$\boldsymbol{0}_{a\times b}$\ the matrix of size $a\times b$ with zeros. The
condition (\ref{eq1}) can be expressed by the linear constraint
\begin{equation}
\theta_{12(1j)}-\theta_{12(2j)}-\theta_{12(1,j+1)}+\theta_{12(2,j+1)}%
\geq0,\text{ }\forall j\in\{1,...,J-1\}, \label{6}%
\end{equation}
since%
\[
\log\vartheta_{j}=\log p_{1j}-\log p_{2j}-\log p_{1,j+1}+\log p_{2,j+1}%
=\theta_{12(1j)}-\theta_{12(2j)}-\theta_{12(1,j+1)}+\theta_{12(2,j+1)}.
\]
Condition (\ref{6}) in matrix notation is given by $\boldsymbol{R\theta}%
\geq\boldsymbol{0}_{J-1}$, with $\boldsymbol{R}=\boldsymbol{e}_{2}^{T}%
\otimes\boldsymbol{G}_{J-1}=(\boldsymbol{0}_{(J-1)\times(J-1)},\boldsymbol{G}%
_{J-1})$, $\boldsymbol{e}_{a}$ is the $a$-th unit vector and $\boldsymbol{G}%
_{h}$ is a $h\times h$\ matrix with $1$-s in the main diagonal and $-1$-s in
the upper superdiagonal. Observe that the restrictions can be expressed also
as $\boldsymbol{G}_{J-1}\boldsymbol{\theta}_{12}\geq\boldsymbol{0}_{J-1}$, and
$\theta_{1(1)}$\ are $\boldsymbol{\theta}_{2}$\ are nuisance parameters
because they do not take part actively in the restrictions.

The kernel of the likelihood function with the new parametrization is obtained
replacing $\boldsymbol{p}$\ by $\boldsymbol{p}(\boldsymbol{\theta})$ in
(\ref{0b}), i.e.
\[
\ell(\boldsymbol{N};\boldsymbol{\theta})=\boldsymbol{N}^{T}\log\boldsymbol{p}%
(\boldsymbol{\theta})=\boldsymbol{N}^{T}(\boldsymbol{W}_{0}\boldsymbol{u}%
+\boldsymbol{W\theta})=nu(\boldsymbol{\theta})+n_{1}u_{1(1)}%
(\boldsymbol{\theta})+\boldsymbol{N}^{T}\boldsymbol{W\theta}.
\]
Hypotheses (\ref{4}) can be now formulated as%
\begin{equation}
H_{0}:\boldsymbol{R\theta}=\boldsymbol{0}_{J-1}\text{ versus }H_{1}%
:\boldsymbol{R\theta}\geq\boldsymbol{0}_{J-1}\text{ and }\boldsymbol{R\theta
}\neq\boldsymbol{0}_{J-1}\text{.} \label{4b}%
\end{equation}
Under $H_{0}$, the parameter space is $\Theta_{0}=\left\{  \boldsymbol{\theta
}\in%
\mathbb{R}
^{2(J-1)}:\boldsymbol{R\theta}=\boldsymbol{0}_{J-1}\right\}  $ and the maximum
likelihood estimator (MLE) of $\boldsymbol{\theta}$ in $\Theta_{0}$ is
$\widehat{\boldsymbol{\theta}}=\arg\max_{\boldsymbol{\theta\in}\Theta_{0}}%
\ell(\boldsymbol{N};\boldsymbol{\theta})$. The overall parameter space is
$\Theta=\left\{  \boldsymbol{\theta}\in%
\mathbb{R}
^{^{2(J-1)}}:\boldsymbol{R\theta}\geq\boldsymbol{0}_{J-1}\right\}  $ and the
MLE of $\boldsymbol{\theta}$ in $\Theta$ is $\widetilde{\boldsymbol{\theta}%
}=\arg\max_{\boldsymbol{\theta\in}\Theta}\ell(\boldsymbol{N}%
;\boldsymbol{\theta})$. It is worthwhile to mention that the probability
vectors for both parametric spaces, $\boldsymbol{p}%
(\widehat{\boldsymbol{\theta}})$ and $\boldsymbol{p}%
(\widetilde{\boldsymbol{\theta}})$\ can be obtained by following the
invariance property of the MLEs first estimating $\boldsymbol{\theta}$\ and
later plugging it into $\boldsymbol{p}(\boldsymbol{\theta})$, however
$\boldsymbol{p}(\widehat{\boldsymbol{\theta}})$ has an explicit expression,%
\begin{equation}
p_{ij}(\widehat{\boldsymbol{\theta}})=\frac{n_{i}(N_{1j}+N_{2j})}{n^{2}},
\label{ind}%
\end{equation}
where $n_{i}=%
{\textstyle\sum_{j=1}^{J}}
N_{ij}$\ (see Christensen (1997), Section 2.3, for more details).

\section{Phi-divergence test-statistics\label{Sec:PD}}

The likelihood ratio statistic for testing (\ref{4}), equivalent to one given
by Dykstra et al. (1995) but adapted for loglinear modeling, is%
\begin{equation}
G^{2}=2(\ell(\boldsymbol{N};\widetilde{\boldsymbol{\theta}})-\ell
(\boldsymbol{N};\widehat{\boldsymbol{\theta}}))=2n%
{\displaystyle\sum\limits_{i=1}^{2}}
{\displaystyle\sum\limits_{j=1}^{J}}
\overline{p}_{ij}\log\frac{p_{ij}(\widetilde{\boldsymbol{\theta}})}%
{p_{ij}(\widehat{\boldsymbol{\theta}})}, \label{LRT}%
\end{equation}
where $\overline{p}_{ij}=N_{ij}/n$, $i=1,2$, $j=1,...,J$. Taking into account
the identifiability constraints (\ref{ident}) and\ $\widehat{u}%
=u(\widehat{\boldsymbol{\theta}})$, $\widetilde{u}%
=u(\widetilde{\boldsymbol{\theta}})$, $\widehat{u}_{1(1)}=u_{1(1)}%
(\widehat{\boldsymbol{\theta}})$, $\widetilde{u}_{1(1)}=u_{1(1)}%
(\widetilde{\boldsymbol{\theta}})$\ (see formulas (\ref{u})-(\ref{u1})),
(\ref{LRT}) can also be expressed as%
\[
G^{2}=2n(\widetilde{u}-\widehat{u})+2n_{1}(\widetilde{u}_{1(1)}-\widehat{u}%
_{1(1)})+2\boldsymbol{N}^{T}\boldsymbol{W}%
(\boldsymbol{\widetilde{\boldsymbol{\theta}}-}\widehat{\boldsymbol{\theta}}).
\]
The chi-square statistic for testing (\ref{4}) is%
\begin{equation}
X^{2}=n%
{\displaystyle\sum\limits_{i=1}^{2}}
{\displaystyle\sum\limits_{j=1}^{J}}
\frac{(p_{ij}(\widehat{\boldsymbol{\theta}})-p_{ij}%
(\widetilde{\boldsymbol{\theta}}))^{2}}{p_{ij}(\widehat{\boldsymbol{\theta}}%
)}. \label{CS}%
\end{equation}

The Kullback-Leibler divergence measure between two $2J$-dimensional
probability vectors $\boldsymbol{p}$ and $\boldsymbol{q}$ is defined as%
\[
d_{Kull}(\boldsymbol{p},\boldsymbol{q})=%
{\displaystyle\sum\limits_{i=1}^{2}}
\sum_{j=1}^{J}p_{ij}\log\frac{p_{ij}}{q_{ij}}%
\]
and the Pearson divergence measure
\[
d_{Pearson}(\boldsymbol{p},\boldsymbol{q})=\frac{1}{2}%
{\displaystyle\sum\limits_{i=1}^{2}}
\sum_{j=1}^{J}\frac{\left(  p_{ij}-q_{ij}\right)  ^{2}}{q_{ij}}.
\]
It is not difficult to check that
\begin{equation}
G^{2}=2n(d_{Kull}(\overline{\boldsymbol{p}},\boldsymbol{p}%
(\widehat{\boldsymbol{\theta}}))-d_{Kull}(\overline{\boldsymbol{p}%
},\boldsymbol{p}(\widetilde{\boldsymbol{\theta}}))) \label{N1}%
\end{equation}
and
\begin{equation}
X^{2}=2nd_{Pearson}(\boldsymbol{p}(\widetilde{\boldsymbol{\theta}%
}),\boldsymbol{p}(\widehat{\boldsymbol{\theta}})), \label{N2}%
\end{equation}
being $\overline{\boldsymbol{p}}=\boldsymbol{N}/n=(\overline{p}_{11}%
,...,\overline{p}_{1J},\overline{p}_{21},....,\overline{p}_{2J})^{T}$ the
vector of relative frequencies.

More general than the Kullback-Leibler divergence and Pearson divergence
measures are $\phi$-divergence measures, defined as
\[
d_{\phi}(\boldsymbol{p},\boldsymbol{q})=%
{\displaystyle\sum\limits_{i=1}^{2}}
\sum_{j=1}^{J}q_{ij}\phi\left(  \frac{p_{ij}}{q_{ij}}\right)  ,
\]
where $\phi:%
\mathbb{R}
_{+}\longrightarrow%
\mathbb{R}
$ is a convex function such that%
\[
\phi(1)=\phi^{\prime}(1)=0\text{, }\phi^{\prime\prime}(1)>0\text{, }%
0\phi(\tfrac{0}{0})=0\text{, }0\phi(\tfrac{p}{0})=p\lim_{u\rightarrow\infty
}\tfrac{\phi(u)}{u}\text{, for }p\neq0.
\]
From a statistical point of view, the first asymptotic statistical results
based on divergence measures in multinomial populations were obtained in
Zografos et al. (1990). For more details about $\phi$-divergence measures see
Pardo (2006) and Cressie and Pardo (2002).

Apart from the likelihood ratio statistic (\ref{LRT}) and the chi-square
(\ref{CS}) statistic, we shall consider two new families of test-statistics
based on $\phi$-divergence measures. The first new family is obtained by
replacing in (\ref{N1}) the Kullback divergence measure by a $\phi$-divergence
measure,%
\begin{equation}
T_{\phi}(\overline{\boldsymbol{p}},\boldsymbol{p}%
(\widetilde{\boldsymbol{\theta}}),\boldsymbol{p}(\widehat{\boldsymbol{\theta}%
}))=\frac{2n}{\phi^{\prime\prime}(1)}(d_{\phi}(\overline{\boldsymbol{p}%
},\boldsymbol{p}(\widehat{\boldsymbol{\theta}}))-d_{\phi}(\overline
{\boldsymbol{p}},\boldsymbol{p}(\widetilde{\boldsymbol{\theta}}))). \label{5a}%
\end{equation}
The second new family is obtained by replacing in (\ref{N2}) the Pearson
divergence measure by a $\phi$-divergence measure,%
\begin{equation}
S_{\phi}(\boldsymbol{p}(\widetilde{\boldsymbol{\theta}}),\boldsymbol{p}%
(\widehat{\boldsymbol{\theta}}))=\frac{2n}{\phi^{\prime\prime}(1)}d_{\phi
}(\boldsymbol{p}(\widetilde{\boldsymbol{\theta}}),\boldsymbol{p}%
(\widehat{\boldsymbol{\theta}})). \label{5b}%
\end{equation}
If we consider $\phi(x)=x\log x-x+1$ in (\ref{5a}), we get $G^{2}$, and if we
consider $\phi(x)=\tfrac{1}{2}(x-1)^{2}$ in (\ref{5a}), we get $X^{2}$.
Test-statistics based on $\phi$-divergence measures have been used in the
framework of loglinear models for some authors, see Cressie and Pardo (2000,
2002, 2003), Mart\'{\i}n and Pardo (2006, 2008b, 2011).

\section{Asymptotic results\label{sec:Main results}}

As starting point, we shall establish the observed Fisher information matrix
associated with $\boldsymbol{\theta}$, $\mathcal{I}_{F}^{(n_{1},n_{2}%
)}(\boldsymbol{\theta})$, for a loglinear model with product-multinomial
sampling as%
\begin{equation}
\mathcal{I}_{F}^{(n_{1},n_{2})}(\boldsymbol{\theta})=\frac{1}{n}%
\boldsymbol{W}^{T}%
\begin{pmatrix}
n_{1}(\boldsymbol{D}_{\boldsymbol{\pi}_{1}(\boldsymbol{\theta})}%
-\boldsymbol{\pi}_{1}(\boldsymbol{\theta})\boldsymbol{\pi}_{1}^{T}%
(\boldsymbol{\theta})) & \boldsymbol{0}_{J\times J}\\
\boldsymbol{0}_{J\times J} & n_{2}(\boldsymbol{D}_{\boldsymbol{\pi}%
_{2}(\boldsymbol{\theta})}-\boldsymbol{\pi}_{2}(\boldsymbol{\theta
})\boldsymbol{\pi}_{2}^{T}(\boldsymbol{\theta}))
\end{pmatrix}
\boldsymbol{W}, \label{FIM1}%
\end{equation}
where $\boldsymbol{D}_{\boldsymbol{a}}$\ is the diagonal matrix of vector
$\boldsymbol{a}$. To proof (\ref{FIM1}), we take into account that the overall
observed Fisher information matrix for product multinomial sampling is the
weighted observed Fisher information matrix associated with each multinomial
sample, $\mathcal{I}_{F,i}^{(n_{1},n_{2})}(\boldsymbol{\theta})$, $i=1,2$,
i.e.%
\begin{align*}
\mathcal{I}_{F}^{(n_{1},n_{2})}(\boldsymbol{\theta})  &  =\frac{n_{1}}%
{n}\mathcal{I}_{F,1}^{(n_{1},n_{2})}(\boldsymbol{\theta})+\frac{n_{2}}%
{n}\mathcal{I}_{F,2}^{(n_{1},n_{2})}(\boldsymbol{\theta}),\\
\mathcal{I}_{F,i}^{(n_{1},n_{2})}(\boldsymbol{\theta})  &  =\boldsymbol{W}%
_{i}^{T}(\boldsymbol{D}_{\boldsymbol{\pi}_{i}(\boldsymbol{\theta}%
)}-\boldsymbol{\pi}_{i}(\boldsymbol{\theta})\boldsymbol{\pi}_{i}%
^{T}(\boldsymbol{\theta}))\boldsymbol{W}_{i},\quad i=1,2,
\end{align*}
such that $\boldsymbol{W}^{T}=(\boldsymbol{W}_{1}^{T},\boldsymbol{W}_{2}^{T}%
)$, $\log\boldsymbol{p}_{1}(\boldsymbol{\theta})=u\boldsymbol{1}_{J}%
+u_{1(1)}\boldsymbol{1}_{J}+\boldsymbol{W}_{1}\boldsymbol{\theta}$ and
$\log\boldsymbol{p}_{2}(\boldsymbol{\theta})=u\boldsymbol{1}_{J}%
+\boldsymbol{W}_{2}\boldsymbol{\theta}$.

When $\boldsymbol{\theta}\in\Theta_{0}$, we shall denote $\boldsymbol{\theta
}_{0}$\ to be the true value of the unknown parameter under $H_{0}$, and in
such a case it holds $\boldsymbol{\pi}_{1}(\boldsymbol{\theta}_{0}%
)=\boldsymbol{\pi}_{2}(\boldsymbol{\theta}_{0})=\boldsymbol{\pi}%
(\boldsymbol{\theta}_{0})=(\pi_{1}(\boldsymbol{\theta}_{0}),...,\pi
_{J}(\boldsymbol{\theta}_{0}))^{T}$, where $\boldsymbol{\pi}_{i}%
(\boldsymbol{\theta}_{0})$ is defined as the probability vector with the terms
given in (\ref{eq5}) and related to the loglinear model through
$\boldsymbol{p}_{i}(\boldsymbol{\theta}_{0})=\frac{n_{i}}{n}\boldsymbol{\pi
}_{i}(\boldsymbol{\theta}_{0})$, $i=1,2$. Notice that $\boldsymbol{\pi}%
_{i}(\boldsymbol{\theta}_{0})$ is fixed as $n_{1},n_{2}\rightarrow\infty$ and
we shall assume that%
\[
\nu_{i}=\lim_{n_{i}\rightarrow\infty}\frac{n_{i}}{n},\quad i=1,2,
\]
is fixed but unknown, i.e. $\lim_{n_{i}\rightarrow\infty}\boldsymbol{p}%
_{i}(\boldsymbol{\theta})=\nu_{i}\boldsymbol{\pi}_{i}(\boldsymbol{\theta}%
_{0})$, $i=1,2$. We shall also denote%
\[
\boldsymbol{\pi}^{\ast}(\boldsymbol{\theta}_{0})=(\pi_{1}(\boldsymbol{\theta
}_{0}),...,\pi_{J-1}(\boldsymbol{\theta}_{0}))^{T},\quad i=1,2.
\]
the $(J-1)$-dimensional vector obtained removing from $\boldsymbol{\pi
}(\boldsymbol{\theta}_{0})$\ the last element. Focussing on the parameter
structure $\boldsymbol{\theta}=(\boldsymbol{\theta}_{12}^{T}%
,\boldsymbol{\theta}_{2}^{T})^{T}$, with $\boldsymbol{\theta}_{12}%
=(\theta_{12(11)},...,\theta_{12(1,J-1)})^{T}$, $\boldsymbol{\theta}%
_{2}=(\theta_{2(1)},...,\theta_{2(J-1)})^{T}$ and the specific structure of
$\boldsymbol{W}$, see (\ref{W}), we shall establish asymptotically the
specific shape of (\ref{FIM1}), a fundamental result for the posterior theorems.

\begin{theorem}
The asymptotic Fisher information matrix of $\boldsymbol{\theta}$,
$\mathcal{I}_{F}(\boldsymbol{\theta})=\lim_{n_{1},n_{2}\rightarrow\infty
}\mathcal{I}_{F}^{(n_{1},n_{2})}(\boldsymbol{\theta})$ when
$\boldsymbol{\theta}\in\Theta_{0}$ is given by%
\begin{equation}
\mathcal{I}_{F}(\boldsymbol{\theta}_{0})=%
\begin{pmatrix}
\boldsymbol{D}_{\boldsymbol{\pi}^{\ast}(\boldsymbol{\theta}_{0})}%
-\boldsymbol{\pi}^{\ast}(\boldsymbol{\theta}_{0})\boldsymbol{\pi}^{\ast
T}(\boldsymbol{\theta}_{0}) & \nu_{1}\left(  \boldsymbol{D}_{\boldsymbol{\pi
}^{\ast}(\boldsymbol{\theta}_{0})}-\boldsymbol{\pi}^{\ast}(\boldsymbol{\theta
}_{0})\boldsymbol{\pi}^{\ast T}(\boldsymbol{\theta}_{0})\right) \\
\nu_{1}\left(  \boldsymbol{D}_{\boldsymbol{\pi}^{\ast}(\boldsymbol{\theta}%
_{0})}-\boldsymbol{\pi}^{\ast}(\boldsymbol{\theta}_{0})\boldsymbol{\pi}^{\ast
T}(\boldsymbol{\theta}_{0})\right)  & \nu_{1}\left(  \boldsymbol{D}%
_{\boldsymbol{\pi}^{\ast}(\boldsymbol{\theta}_{0})}-\boldsymbol{\pi}^{\ast
}(\boldsymbol{\theta}_{0})\boldsymbol{\pi}^{\ast T}(\boldsymbol{\theta}%
_{0})\right)
\end{pmatrix}
. \label{FIM2}%
\end{equation}

\end{theorem}

\begin{proof}
Replacing $\boldsymbol{\theta}$ by $\boldsymbol{\theta}_{0}$ and the explicit
expression of $\boldsymbol{W}$\ in the general expression of the finite sample
size Fisher information matrix for two independent multinomial samples,
(\ref{FIM1}), we obtain through the property of the Kronecker product given in
(1.22) of Harville (2008, page 341) that%
\begin{align*}
\mathcal{I}_{F}^{(n_{1},n_{2})}(\boldsymbol{\theta}_{0})  &  =\left(
\begin{pmatrix}
1 & 1\\
1 & 0
\end{pmatrix}
\otimes%
\begin{pmatrix}
\boldsymbol{I}_{J-1}\\
\boldsymbol{0}_{J-1}^{T}%
\end{pmatrix}
^{T}\right)  \left(  diag\{\tfrac{n_{i}}{n}\}_{i=1}^{2}\otimes(\boldsymbol{D}%
_{\boldsymbol{\pi}(\boldsymbol{\theta}_{0})}-\boldsymbol{\pi}%
(\boldsymbol{\theta}_{0})\boldsymbol{\pi}^{T}(\boldsymbol{\theta}%
_{0}))\right)  \left(
\begin{pmatrix}
1 & 1\\
1 & 0
\end{pmatrix}
\otimes%
\begin{pmatrix}
\boldsymbol{I}_{J-1}\\
\boldsymbol{0}_{J-1}^{T}%
\end{pmatrix}
\right) \\
&  =\left(
\begin{pmatrix}
1 & 1\\
1 & 0
\end{pmatrix}
diag\{\tfrac{n_{i}}{n}\}_{i=1}^{2}%
\begin{pmatrix}
1 & 1\\
1 & 0
\end{pmatrix}
\right)  \otimes\left(
\begin{pmatrix}
\boldsymbol{I}_{J-1}\\
\boldsymbol{0}_{J-1}^{T}%
\end{pmatrix}
^{T}(\boldsymbol{D}_{\boldsymbol{\pi}(\boldsymbol{\theta}_{0})}%
-\boldsymbol{\pi}(\boldsymbol{\theta}_{0})\boldsymbol{\pi}^{T}%
(\boldsymbol{\theta}_{0}))%
\begin{pmatrix}
\boldsymbol{I}_{J-1}\\
\boldsymbol{0}_{J-1}^{T}%
\end{pmatrix}
\right) \\
&  =%
\begin{pmatrix}
1 & \tfrac{n_{1}}{n}\\
\tfrac{n_{1}}{n} & \tfrac{n_{1}}{n}%
\end{pmatrix}
\otimes\left(  \boldsymbol{D}_{\boldsymbol{\pi}^{\ast}(\boldsymbol{\theta}%
_{0})}-\boldsymbol{\pi}^{\ast}(\boldsymbol{\theta}_{0})\boldsymbol{\pi}^{\ast
T}(\boldsymbol{\theta}_{0})\right)  ,
\end{align*}
and then%
\begin{equation}
\mathcal{I}_{F}(\boldsymbol{\theta}_{0})=%
\begin{pmatrix}
1 & \nu_{1}\\
\nu_{1} & \nu_{1}%
\end{pmatrix}
\otimes\left(  \boldsymbol{D}_{\boldsymbol{\pi}^{\ast}(\boldsymbol{\theta}%
_{0})}-\boldsymbol{\pi}^{\ast}(\boldsymbol{\theta}_{0})\boldsymbol{\pi}^{\ast
T}(\boldsymbol{\theta}_{0})\right)  . \label{iF}%
\end{equation}

\end{proof}

The following theorem establishes that the asymptotic distribution of the
families of test statistics (\ref{5a}) and (\ref{5b}) corresponds to a
$J$-dimensional chi-bar squared random variable, a mixture of $J$ chi-squared
distributions. Let $E=\{1,...,J-1\}$ be the whole set of all row-indices of
matrix $\boldsymbol{R}$, $\mathcal{F}(E)$ the family of all possible subsets
of $E$, and $\boldsymbol{R}(S\mathbf{)}$ is a submatrix of $\boldsymbol{R}%
$\ with row-\'{\i}ndices belonging to $S\in\mathcal{F}(E)$. We must not forget
that $\boldsymbol{R}=(\boldsymbol{0}_{(J-1)\times(J-1)},\boldsymbol{G}_{J-1})$
and therefore $\boldsymbol{R}(S)=(\boldsymbol{0}_{card(S)\times(J-1)}%
,\boldsymbol{G}_{J-1}(S))$.

We denote by $\boldsymbol{H}(\boldsymbol{\theta})$\ the following
$(J-1)\times(J-1)$ tridiagonal matrix%
\begin{equation}
\boldsymbol{H}(\boldsymbol{\theta})=\frac{1}{\nu_{1}\nu_{2}}%
\begin{pmatrix}
\frac{\pi_{1}(\boldsymbol{\theta})+\pi_{2}(\boldsymbol{\theta})}{\pi
_{1}(\boldsymbol{\theta})\pi_{2}(\boldsymbol{\theta})} & -\frac{1}{\pi
_{2}(\boldsymbol{\theta})} &  &  & \\
-\frac{1}{\pi_{1}(\boldsymbol{\theta})} & \frac{\pi_{2}(\boldsymbol{\theta
})+\pi_{3}(\boldsymbol{\theta})}{\pi_{2}(\boldsymbol{\theta})\pi
_{3}(\boldsymbol{\theta})} & -\frac{1}{\pi_{3}(\boldsymbol{\theta})} &  & \\
& -\frac{1}{\pi_{3}(\boldsymbol{\theta})} & \frac{\pi_{3}(\boldsymbol{\theta
})+\pi_{4}(\boldsymbol{\theta})}{\pi_{3}(\boldsymbol{\theta})\pi
_{4}(\boldsymbol{\theta})} & \ddots & \\
&  & \ddots & \ddots & -\frac{1}{\pi_{J-1}(\boldsymbol{\theta})}\\
&  &  & -\frac{1}{\pi_{J-1}(\boldsymbol{\theta})} & \frac{\pi_{J-1}%
(\boldsymbol{\theta})+\pi_{J}(\boldsymbol{\theta})}{\pi_{J-1}%
(\boldsymbol{\theta})\pi_{J}(\boldsymbol{\theta})}%
\end{pmatrix}
, \label{H}%
\end{equation}
and by $\boldsymbol{H}(S_{1},S_{2},\boldsymbol{\theta})$ the submatrix of
$\boldsymbol{H}(\boldsymbol{\theta})$\ obtained by deleting from it the
row-indices contained in the set $S_{1}$ and column-indices contained in the
set $S_{2}$.

\begin{theorem}
\label{Th1}Under $H_{0}$, the asymptotic distribution of $S_{\phi
}(\boldsymbol{p}(\widetilde{\boldsymbol{\theta}}),\boldsymbol{p}%
(\widehat{\boldsymbol{\theta}}))$\ and $T_{\phi}(\overline{\boldsymbol{p}%
},\boldsymbol{p}(\widetilde{\boldsymbol{\theta}}),\boldsymbol{p}%
(\widehat{\boldsymbol{\theta}}))$\ is%
\[
\lim_{n_{1},n_{2}\rightarrow\infty}\Pr\left(  S_{\phi}(\boldsymbol{p}%
(\widetilde{\boldsymbol{\theta}}),\boldsymbol{p}(\widehat{\boldsymbol{\theta}%
}))\leq x\right)  =\lim_{n_{1},n_{2}\rightarrow\infty}\Pr\left(  T_{\phi
}(\overline{\boldsymbol{p}},\boldsymbol{p}(\widetilde{\boldsymbol{\theta}%
}),\boldsymbol{p}(\widehat{\boldsymbol{\theta}}))\leq x\right)  =\sum
_{j=0}^{J-1}w_{j}(\boldsymbol{\theta}_{0})\Pr\left(  \chi_{(J-1)-j}^{2}\leq
x\right)
\]
where $\chi_{0}^{2}=0$ a.s. and $\{w_{j}(\boldsymbol{\theta}_{0}%
)\}_{j=0}^{J-1}$ is the set of weights such that $\sum_{j=0}^{J-1}%
w_{j}(\boldsymbol{\theta}_{0})=1$ and%
\begin{equation}
w_{j}(\boldsymbol{\theta}_{0})=\sum_{S\in\mathcal{F}(E),\mathrm{card}(S)=j}%
\Pr\left(  \boldsymbol{Z}_{1}(S)\geq\boldsymbol{0}_{j}\right)  \Pr\left(
\boldsymbol{Z}_{2}(S)\geq\boldsymbol{0}_{(J-1)-j}\right)  , \label{eqw}%
\end{equation}
where%
\begin{align*}
\boldsymbol{Z}_{1}(S)  &  \sim\mathcal{N}\left(  \boldsymbol{0}_{\mathrm{card}%
(S)},\boldsymbol{H}^{-1}(S,S,\boldsymbol{\theta}_{0})\right)  ,\\
\boldsymbol{Z}_{2}(S)  &  \sim\mathcal{N}\left(  \boldsymbol{0}%
_{(J-1)-\mathrm{card}(S)},\boldsymbol{H}(S^{C},S^{C},\boldsymbol{\theta}%
_{0})-\boldsymbol{H}(S^{C},S,\boldsymbol{\theta}_{0})\boldsymbol{H}%
^{-1}(S,S,\boldsymbol{\theta}_{0})\boldsymbol{H}^{T}(S^{C}%
,S,\boldsymbol{\theta}_{0})\right)  ,
\end{align*}
$S^{C}=E-S$ and $\mathrm{card}(S)$ denotes the cardinal of the set $S$.
\end{theorem}

\begin{proof}
By following similar arguments of Mart\'{\i}n and Balakrishnan we obtain
$\boldsymbol{H}(S,S,\boldsymbol{\theta}_{0})=\boldsymbol{R}(S\mathbf{)}%
\mathcal{I}_{F}^{-1}(\boldsymbol{\theta}_{0})\boldsymbol{R}^{T}(S\mathbf{)}$
(see Appendix \ref{ProofTh1ContrA}, for the details). In particular,
$\boldsymbol{H}(\boldsymbol{\theta}_{0})=\boldsymbol{H}(S,S,\boldsymbol{\theta
}_{0})$ with $S=E$, i.e.%
\begin{align*}
\boldsymbol{H}(\boldsymbol{\theta}_{0})  &  =\boldsymbol{R}(E\mathbf{)}%
\mathcal{I}_{F}^{-1}(\boldsymbol{\theta}_{0})\boldsymbol{R}^{T}(E\mathbf{)}\\
&  \mathbf{=}(\boldsymbol{0}_{(J-1)\times(J-1)},\boldsymbol{G}_{J-1}%
)\mathcal{I}_{F}^{-1}(\boldsymbol{\theta}_{0})(\boldsymbol{0}_{(J-1)\times
(J-1)},\boldsymbol{G}_{J-1})^{T},
\end{align*}
where $\mathcal{I}_{F}(\boldsymbol{\theta}_{0})$ is (\ref{iF}). By following
the properties of the inverse of the Kronecker product for calculating the
inverse of (\ref{iF}),%
\begin{align*}
\mathcal{I}_{F}^{-1}(\boldsymbol{\theta}_{0})  &  =%
\begin{pmatrix}
1 & \nu_{1}\\
\nu_{1} & \nu_{1}%
\end{pmatrix}
^{-1}\otimes\left(  \boldsymbol{D}_{\boldsymbol{\pi}^{\ast}(\boldsymbol{\theta
}_{0})}-\boldsymbol{\pi}^{\ast}(\boldsymbol{\theta}_{0})\boldsymbol{\pi}^{\ast
T}(\boldsymbol{\theta})\right)  ^{-1}\\
&  =%
\begin{pmatrix}
\frac{1}{\nu_{2}} & -\frac{1}{\nu_{2}}\\
-\frac{1}{\nu_{2}} & \frac{1}{\nu_{1}\nu_{2}}%
\end{pmatrix}
\otimes\left(  \boldsymbol{D}_{\boldsymbol{\pi}^{\ast}(\boldsymbol{\theta}%
_{0})}^{-1}+\frac{1}{\pi_{J}(\boldsymbol{\theta}_{0})}\boldsymbol{1}%
_{J-1}\boldsymbol{1}_{J-1}^{T}\right)  ,
\end{align*}
and replacing it in the previous expression of $\boldsymbol{H}%
(\boldsymbol{\theta}_{0})$,%
\begin{align*}
\boldsymbol{H}(\boldsymbol{\theta}_{0})  &  =\frac{1}{\nu_{1}\nu_{2}%
}\boldsymbol{G}_{J-1}\left(  \boldsymbol{D}_{\boldsymbol{\pi}^{\ast
}(\boldsymbol{\theta}_{0})}^{-1}+\frac{1}{\pi_{J}(\boldsymbol{\theta}_{0}%
)}\boldsymbol{1}_{J-1}\boldsymbol{1}_{J-1}^{T}\right)  \boldsymbol{G}%
_{J-1}^{T}\\
&  =\frac{1}{\nu_{1}\nu_{2}}\left(  \boldsymbol{G}_{J-1}\boldsymbol{D}%
_{\boldsymbol{\pi}^{\ast}(\boldsymbol{\theta}_{0})}^{-1}\boldsymbol{G}%
_{J-1}^{T}+\frac{1}{\pi_{J}(\boldsymbol{\theta}_{0})}\boldsymbol{e}%
_{J-1}\boldsymbol{e}_{J-1}^{T}\right)  ,
\end{align*}
which is equal to (\ref{H}).\medskip
\end{proof}

Even though there is an equality in (\ref{4b}), $\boldsymbol{\theta}$ is not a
fixed vector under the null hypothesis since such an equality is effective
only for $\boldsymbol{\theta}_{12}$,\ and thus $\boldsymbol{\theta}_{2}$ is a
vector of nuisance parameters. This means that we have a composite null
hypothesis which requires estimation of $\boldsymbol{\theta}\in\Theta_{0}$,
through $\widehat{\boldsymbol{\theta}}$ and we cannot use directly the results
based on Theorem \ref{Th1}. The tests performed replacing the parameter
$\boldsymbol{\theta}_{0}$\ of the asymptotic distribution by
$\widehat{\boldsymbol{\theta}}$ are called \textquotedblleft local
tests\textquotedblright\ (see Dardanoni and Forcina (1998)) and they are
usually considered to be good approximations of the theoretical tests.

In relation to the weights, $\{w_{j}(\boldsymbol{\theta}_{0})\}_{j=1,...,J}$,
there are explicit expressions when $J\in\{2,3,4\}$ based on the matrix given
in (\ref{H}) and formulas (3.24), (3.25) and (3.26) in Silvapulle and Sen
(2005, page 80). When $J=2$, $w_{0}(\boldsymbol{\theta}_{0})=w_{1}%
(\boldsymbol{\theta}_{0})=\frac{1}{2}$. When $J=3$, the estimators of the
weights are%
\begin{equation}
\left\{
\begin{array}
[c]{l}%
w_{0}(\widehat{\boldsymbol{\theta}})=\tfrac{1}{2}-w_{2}%
(\widehat{\boldsymbol{\theta}}),\\
w_{1}(\widehat{\boldsymbol{\theta}})=\frac{1}{2},\\
w_{2}(\widehat{\boldsymbol{\theta}})=\tfrac{1}{2\pi}\arccos\widehat{\rho}%
_{12},
\end{array}
\right.  \label{weightsJ=3}%
\end{equation}
where%
\begin{equation}
\widehat{\rho}_{ij}=\tfrac{\widehat{\sigma}_{ij}}{\sqrt{\widehat{\sigma}%
_{ii}\widehat{\sigma}_{jj}}}=-\sqrt{\frac{(N_{1i}+N_{2i})(N_{1,j+1}%
+N_{2,j+1})}{(N_{1i}+N_{2i}+N_{1j}+N_{2j})(N_{1j}+N_{2j}+N_{1,j+1}+N_{2j+1})}%
}, \label{cor}%
\end{equation}
is the correlation associated with the $i$-th and $j$-th variable of a central
random variable with variance-covariance matrix
\[
\boldsymbol{H}(\widehat{\boldsymbol{\theta}})=\frac{1}{\widehat{\nu}%
_{1}\widehat{\nu}_{2}}%
\begin{pmatrix}
\frac{\pi_{1}(\widehat{\boldsymbol{\theta}})+\pi_{2}%
(\widehat{\boldsymbol{\theta}})}{\pi_{1}(\widehat{\boldsymbol{\theta}})\pi
_{2}(\widehat{\boldsymbol{\theta}})} & -\frac{1}{\pi_{2}%
(\widehat{\boldsymbol{\theta}})}\\
-\frac{1}{\pi_{2}(\widehat{\boldsymbol{\theta}})} & \frac{\pi_{2}%
(\widehat{\boldsymbol{\theta}})+\pi_{3}(\widehat{\boldsymbol{\theta}})}%
{\pi_{2}(\widehat{\boldsymbol{\theta}})\pi_{3}(\widehat{\boldsymbol{\theta}})}%
\end{pmatrix}
,
\]
where $\pi_{j}(\widehat{\boldsymbol{\theta}})=\frac{N_{1j}+N_{2j}}{n}$. When
$J=4$,%
\begin{equation}
\left\{
\begin{array}
[c]{l}%
w_{0}(\widehat{\boldsymbol{\theta}})=\tfrac{1}{4\pi}\left(  2\pi
-\arccos\widehat{\rho}_{12}-\arccos\widehat{\rho}_{13}-\arccos\widehat{\rho
}_{23}\right)  ,\\
w_{1}(\widehat{\boldsymbol{\theta}})=\tfrac{1}{4\pi}\left(  3\pi
-\arccos\widehat{\rho}_{12\cdot3}-\arccos\widehat{\rho}_{13\cdot2}%
-\arccos\widehat{\rho}_{23\cdot1}\right)  ,\\
w_{2}(\widehat{\boldsymbol{\theta}})=\tfrac{1}{2}-w_{0}%
(\widehat{\boldsymbol{\theta}}),\\
w_{3}(\widehat{\boldsymbol{\theta}})=\tfrac{1}{2}-w_{1}%
(\widehat{\boldsymbol{\theta}}),
\end{array}
\right.  \label{weightsJ=4}%
\end{equation}
which depend on the estimation of the marginal (\ref{cor}) and conditional
correlations%
\[
\widehat{\rho}_{ij\cdot k}=\tfrac{\widehat{\rho}_{ij}-\widehat{\rho}%
_{ik}\widehat{\rho}_{kj}}{\sqrt{(1-\widehat{\rho}_{ik}^{2})(1-\widehat{\rho
}_{kj}^{2})}},
\]
associated with the $i$-th and $j$-th variable, given a value of the $k$-th
variable, of a central random variable with variance-covariance matrix%
\[
\boldsymbol{H}(\widehat{\boldsymbol{\theta}})=\frac{1}{\widehat{\nu}%
_{1}\widehat{\nu}_{2}}%
\begin{pmatrix}
\frac{\pi_{1}(\widehat{\boldsymbol{\theta}})+\pi_{2}%
(\widehat{\boldsymbol{\theta}})}{\pi_{1}(\widehat{\boldsymbol{\theta}})\pi
_{2}(\widehat{\boldsymbol{\theta}})} & -\frac{1}{\pi_{2}%
(\widehat{\boldsymbol{\theta}})} & 0\\
-\frac{1}{\pi_{2}(\widehat{\boldsymbol{\theta}})} & \frac{\pi_{2}%
(\widehat{\boldsymbol{\theta}})+\pi_{3}(\widehat{\boldsymbol{\theta}})}%
{\pi_{2}(\widehat{\boldsymbol{\theta}})\pi_{3}(\widehat{\boldsymbol{\theta}})}
& -\frac{1}{\pi_{3}(\widehat{\boldsymbol{\theta}})}\\
0 & -\frac{1}{\pi_{3}(\widehat{\boldsymbol{\theta}})} & \frac{\pi
_{3}(\widehat{\boldsymbol{\theta}})+\pi_{4}(\widehat{\boldsymbol{\theta}}%
)}{\pi_{3}(\widehat{\boldsymbol{\theta}})\pi_{4}(\widehat{\boldsymbol{\theta}%
})}%
\end{pmatrix}
.
\]
It is interesting to point out that the factor related to the sample size in
each multinomial sample, $\frac{1}{\widehat{\nu}_{1}\widehat{\nu}_{2}}$, have
no effect in the expression of estimator for the weights of the chi-bar
squared distribution These formulas will be considered in the forthcoming
sections. It is worthwhile to mention that the normal orthant probabilities
for the weights given in (\ref{eqw}), can also be computed for any value of
$J$ using the \texttt{mvtnorm} R package (see
\hyperref{http://CRAN.R-project.org/package=mvtnorm}{}{}%
{http://CRAN.R-project.org/package=mvtnorm}%
, for details).

\section{Numerical example\label{sec:Numerical example}}

In this section the data set of the introduction (Table \ref{tttt1}), where
$J=4$, is analyzed. The sample, a realization of $\boldsymbol{N}$, is
summarized in the following vector%
\[
\boldsymbol{n}=(n_{11},n_{12},n_{13},n_{14},n_{21},n_{22},n_{23},n_{24}%
)^{T}=(11,8,8,5,6,4,10,12)^{T}.
\]
The order restricted MLE under likelihood ratio order, obtained through the
\texttt{E04UCF} subroutine of\ \texttt{NAG} Fortran library (%
\hyperref{http://www.nag.co.uk/numeric/fl/FLdescription.asp}{}{}%
{http://www.nag.co.uk/numeric/fl/FLdescription.asp}%
), is%
\[
\widetilde{\boldsymbol{\theta}}%
=(-0.7164,-1.0647,-0.1823,1.5173,1.5173,0.6523)^{T}.
\]
The estimation of the probability vectors of interest is%
\begin{align*}
\overline{\boldsymbol{p}}  &
=(0.1719,0.1250,0.1250,0.0781,0.0938,0.0625,0.1563,0.1875)^{T},\\
\boldsymbol{p}(\widetilde{\boldsymbol{\theta}})  &
=(0.1740,0.1228,0.1250,0.0781,0.0916,0.0647,0.1563,0.1875)^{T},\\
\boldsymbol{p}(\widehat{\boldsymbol{\theta}})  &
=(0.1328,0.0938,0.1406,0.1328,0.1328,0.0938,0.1406,0.1328)^{T},
\end{align*}
and the estimation of the weights, based on (\ref{weightsJ=4}), are%
\[
w_{0}(\widehat{\boldsymbol{\theta}})=0.0381,\quad w_{1}%
(\widehat{\boldsymbol{\theta}})=0.2420,\quad w_{2}(\widehat{\boldsymbol{\theta
}})=0.461\,8,\quad w_{3}(\widehat{\boldsymbol{\theta}})=0.2580.
\]
In order to solve analytically the example we shall consider a particular
function $\phi$ in (\ref{5a}) and (\ref{5b}). Taking%
\[
\phi_{\lambda}(x)=\frac{x^{\lambda+1}-x-\lambda(x-1)}{\lambda(\lambda+1)},
\]
we get the \textquotedblleft the power divergence family\textquotedblright%
\[
d_{\phi_{\lambda}}(\boldsymbol{p},\boldsymbol{q})=\frac{1}{\lambda(\lambda
+1)}\left(
{\displaystyle\sum\limits_{i=1}^{2}}
{\displaystyle\sum\limits_{j=1}^{J}}
\tfrac{p_{ij}^{\lambda+1}}{q_{ij}^{\lambda}(\widehat{\boldsymbol{\theta}}%
)}-1\right)
\]
in such a way that for each $\lambda\in%
\mathbb{R}
-\{-1,0\}$\ a different divergence measure is obtained, and thus%
\begin{align}
T_{\lambda}  &  =T_{\phi_{\lambda}}(\overline{\boldsymbol{p}},\boldsymbol{p}%
(\widetilde{\boldsymbol{\theta}}),\boldsymbol{p}(\widehat{\boldsymbol{\theta}%
}))=\frac{2n}{\lambda(\lambda+1)}\left(
{\displaystyle\sum\limits_{i=1}^{2}}
{\displaystyle\sum\limits_{j=1}^{J}}
\frac{\overline{p}_{ij}^{\lambda+1}}{p_{ij}^{\lambda}%
(\widehat{\boldsymbol{\theta}})}-%
{\displaystyle\sum\limits_{i=1}^{2}}
{\displaystyle\sum\limits_{j=1}^{J}}
\frac{\overline{p}_{ij}^{\lambda+1}}{p_{ij}^{\lambda}%
(\widetilde{\boldsymbol{\theta}})}\right)  ,\label{PD1}\\
S_{\lambda}  &  =S_{\phi_{\lambda}}(\boldsymbol{p}%
(\widetilde{\boldsymbol{\theta}}),\boldsymbol{p}(\widehat{\boldsymbol{\theta}%
}))=\frac{2n}{\lambda(\lambda+1)}\left(
{\displaystyle\sum\limits_{i=1}^{2}}
{\displaystyle\sum\limits_{j=1}^{J}}
\frac{p_{ij}^{\lambda+1}(\widetilde{\boldsymbol{\theta}})}{p_{ij}^{\lambda
}(\widehat{\boldsymbol{\theta}})}-1\right)  . \label{PD2}%
\end{align}
It is also possible to cover the real line for $\lambda$, by defining
\[
d_{\phi_{\lambda}}(\boldsymbol{p},\boldsymbol{q})=\lim_{\ell\rightarrow
\lambda}d_{\phi_{\ell}}(\boldsymbol{p},\boldsymbol{q}),\quad\lambda
\in\{-1,0\},
\]
and by considering $T_{\lambda}=\lim_{\lambda\rightarrow\ell}T_{\ell}$,
$S_{\lambda}=\lim_{\lambda\rightarrow\ell}S_{\ell}$, for $\lambda\in\{0,-1\}$,
i.e.
\begin{align}
T_{0}  &  =T_{\phi_{0}}(\overline{\boldsymbol{p}},\boldsymbol{p}%
(\widetilde{\boldsymbol{\theta}}),\boldsymbol{p}(\widehat{\boldsymbol{\theta}%
}))=G^{2}=2n%
{\displaystyle\sum\limits_{i=1}^{2}}
{\displaystyle\sum\limits_{j=1}^{J}}
\overline{p}_{ij}\log\frac{p_{ij}(\widetilde{\boldsymbol{\theta}})}%
{p_{ij}(\widehat{\boldsymbol{\theta}})},\label{PD3}\\
T_{-1}  &  =T_{\phi_{-1}}(\overline{\boldsymbol{p}},\boldsymbol{p}%
(\widetilde{\boldsymbol{\theta}}),\boldsymbol{p}(\widehat{\boldsymbol{\theta}%
}))=2n\left(
{\displaystyle\sum\limits_{i=1}^{2}}
{\displaystyle\sum\limits_{j=1}^{J}}
p_{ij}(\widehat{\boldsymbol{\theta}})\log\frac{p_{ij}%
(\widehat{\boldsymbol{\theta}})}{\overline{p}_{ij}}-%
{\displaystyle\sum\limits_{i=1}^{2}}
{\displaystyle\sum\limits_{j=1}^{J}}
p_{ij}(\widetilde{\boldsymbol{\theta}})\log\frac{p_{ij}%
(\widetilde{\boldsymbol{\theta}})}{\overline{p}_{ij}}\right)  \label{PD4}%
\end{align}
and
\begin{align}
S_{0}  &  =S_{\phi_{0}}(\boldsymbol{p}(\widetilde{\boldsymbol{\theta}%
}),\boldsymbol{p}(\widehat{\boldsymbol{\theta}}))=2nd_{Kull}(\boldsymbol{p}%
(\widetilde{\boldsymbol{\theta}}),\boldsymbol{p}(\widehat{\boldsymbol{\theta}%
}))=2n%
{\displaystyle\sum\limits_{i=1}^{2}}
{\displaystyle\sum\limits_{j=1}^{J}}
p_{ij}(\widetilde{\boldsymbol{\theta}})\log\frac{p_{ij}%
(\widetilde{\boldsymbol{\theta}})}{p_{ij}(\widehat{\boldsymbol{\theta}}%
)},\label{PD5}\\
S_{-1}  &  =S_{\phi_{-1}}(\boldsymbol{p}(\widetilde{\boldsymbol{\theta}%
}),\boldsymbol{p}(\widehat{\boldsymbol{\theta}}))=2nd_{Kull}(\boldsymbol{p}%
(\widehat{\boldsymbol{\theta}}),\boldsymbol{p}(\widetilde{\boldsymbol{\theta}%
}))=2n%
{\displaystyle\sum\limits_{j=1}^{J}}
p_{ij}(\widehat{\boldsymbol{\theta}})\log\frac{p_{ij}%
(\widehat{\boldsymbol{\theta}})}{p_{ij}(\widetilde{\boldsymbol{\theta}})}.
\label{PD6}%
\end{align}
It is well known that $d_{\phi_{0}}(\boldsymbol{p},\boldsymbol{q}%
)=d_{Kull}(\boldsymbol{p},\boldsymbol{q})$ and $d_{\phi_{1}}(\boldsymbol{p}%
,\boldsymbol{q})=d_{Pearson}(\boldsymbol{p},\boldsymbol{q})$, which is very
interesting since $G^{2}$ and $X^{2}$ are members of the power divergence
based test-statistics. It is also worthwhile to mention that $d_{\phi_{-1}%
}(\boldsymbol{p},\boldsymbol{q})=d_{Kull}(\boldsymbol{q},\boldsymbol{p})$.

In Table \ref{t1}, the power divergence based test-statistics for some values
of $\lambda$ in $\Lambda=\{-1.5,-1,-\frac{1}{2},0,\frac{2}{3},1,1.5,2,3\}$,
and their corresponding asymptotic $p$-values are shown. In all of them it is
concluded, with a significance level equal to $0.05$, that an equal effect of
both treatments is rejected and hence the treatment is more effective than the
control to heal the ulcer.
\begin{table}[htbp]  \tabcolsep2.8pt  \centering
$%
\begin{tabular}
[c]{cccccccccc}\hline\hline
test-statistic & $\lambda=-1.5$ & $\lambda=-1$ & $\lambda=-\frac{1}{2}$ &
$\lambda=0$ & $\lambda=\frac{2}{3}$ & $\lambda=1$ & $\lambda=1.5$ &
$\lambda=2$ & $\lambda=3$\\\hline
\multicolumn{1}{l}{$\overset{}{T_{\lambda}}$} & 6.5323 & 6.3215 & 6.1562 &
\textbf{6.0323} & 5.9261 & 5.8965 & 5.8803 & 5.8965 & 6.0244\\
$p$\textrm{-}$\mathrm{value}(T_{\lambda})$ & 0.0175 & 0.0194 & 0.0211 &
\textbf{0.0225} & 0.0238 & 0.0241 & 0.0243 & 0.0241 & 0.0226\\
\multicolumn{1}{l}{$S_{\lambda}$} & 6.5277 & 6.3189 & 6.1551 & 6.0323 &
5.9270 & \textbf{5.8977} & 5.8815 & 5.8977 & 6.0244\\
$p$\textrm{-}$\mathrm{value}(S_{\lambda})$ & 0.0175 & 0.0195 & 0.0212 &
0.0225 & 0.0238 & \textbf{0.0241} & 0.0243 & 0.0241 & 0.0226\\\hline\hline
\end{tabular}
\ \ \ \ \ \ \ \ \ \ \ \ $%
\caption{Power divergence based test-statistics and asymptotic p-values for the data given Table \ref{tttt1}.\label{t1}}%
\end{table}%
\bigskip

The $p$-values given in Table \ref{t1} were obtained by the following
algorithm:\newline Let $T\in\{{T_{\lambda},S_{\lambda}}\}_{\lambda\in\Lambda}$
be the test-statistic associated with (\ref{4}). In the following steps the
corresponding asymptotic $p$-value, based on the asymptotic distribution of
Theorem \ref{Th1}, is calculated once it is suppose we have $\{w_{j}%
(\widehat{\boldsymbol{\theta}})\}_{j=0}^{J-1}$:

\noindent\texttt{STEP 1: Using }$\boldsymbol{n}$\texttt{\ calculate
}$\boldsymbol{p}(\widehat{\boldsymbol{\theta}})$\texttt{\ taking into account
(\ref{ind}).}\newline\texttt{STEP 2: Using }$\boldsymbol{p}%
(\widehat{\boldsymbol{\theta}})$\texttt{\ calculate value }$t$\texttt{ of
test-statistic }$T$\texttt{ using the corresponding expression in
(\ref{PD1})-(\ref{PD6}).}\newline\texttt{STEP 3: If }$T\leq0$ \texttt{then
compute }$p$\textrm{-}$\mathrm{value}(T):=1$ \texttt{and STOP, otherwise
compute }$p$\textrm{-}$\mathrm{value}(T):=0$.\newline\texttt{STEP 4: }For
$j=0,...,J-2$\texttt{, do }$p$\textrm{-}$\mathrm{value}(T):=p$\textrm{-}%
$\mathrm{value}(T)+w_{j}(\widehat{\boldsymbol{\theta}})\Pr\left(
\chi_{(J-1)-j}^{2}>t\right)  $.\texttt{\newline\hspace*{1.6cm}E.g., the NAG
Fortran library subroutine G01ECF can be useful.}

Recently, Shan and Ma (2014) have studied a similar problem as (2a)-(2b), but
considering different alternative hypotheses, since they consider odds ratios
based on cumulative probabilities. Focussed on probabilities rather than
cumulative probabilities, we are going to include the asymptotic version of
their test-statistic in our numerical study as well as later, in the
simulation study: the two sample Wilcoxon test-statistic for discrete data
(ties), also known as Wilcoxon mid-rank test-statistic. Metha et al. (1984)
proposed such a test-statistic for solving exactly the same alternative
hypothesis studied in this paper either as a permutation or as asymptotic
test. Our null and alternative hypotheses are a particular case of their
hypotheses, taking in their Section 4 $\phi^{\ast}=1$. The expression of the
Wilcoxon mid-rank test-statistic is%
\begin{equation}
W=%
{\textstyle\sum\limits_{j=1}^{J}}
r_{j}n_{1j}, \label{wilc}%
\end{equation}
where $r_{1}=(n_{\bullet1}+2)/2$ and $r_{j}=%
{\textstyle\sum\nolimits_{\ell=1}^{j-1}}
n_{\bullet\ell}+\left.  \left(  n_{\bullet j}+1\right)  \right/  2$,
$j=2,...,J$, $n_{\bullet j}=n_{1j}+n_{2j}$, and the corresponding asymptotic
distribution is normal with mean $\mu_{W}=\left.  n_{1}\left(  n+1\right)
\right/  2$ and variance
\[
\sigma_{W}^{2}=n_{1}n_{2}\frac{n+1-\frac{1}{n(n-1)}%
{\textstyle\sum\nolimits_{j=1}^{J}}
(n_{\bullet j}^{3}-n_{\bullet j})}{12}.
\]
The Wilcoxon mid-rank test-statistic for the data of Table \ref{tttt1} is
$W=875$\ and with the corresponding $p$-value, $0.01094$, the same conclusion
is obtained, i.e. rejecting the hypothesis of equal effect of both treatments
with $5\%$ significance level.

\section{Simulation study\label{sec:Simulation Study}}

\subsection{2x2 table: one sided in comparison with the two sided test}

In this section we illustrate in what sense the likelihood ratio test given in
(\ref{LRT}),%
\begin{equation}
G^{2}=2n%
{\displaystyle\sum\limits_{i=1}^{2}}
{\displaystyle\sum\limits_{j=1}^{2}}
\overline{p}_{ij}\log\frac{p_{ij}(\widetilde{\boldsymbol{\theta}})}%
{p_{ij}(\widehat{\boldsymbol{\theta}})}=2%
{\displaystyle\sum\limits_{i=1}^{2}}
{\displaystyle\sum\limits_{j=1}^{2}}
n_{ij}\log\frac{\pi_{ij}(\widetilde{\boldsymbol{\theta}})}{\pi_{ij}%
(\widehat{\boldsymbol{\theta}})}, \label{G 1}%
\end{equation}
is different from the one for the non order restricted alternative hypothesis
(two sided test, in $2\times2$ tables)%
\begin{equation}
\bar{G}^{2}=2n%
{\displaystyle\sum\limits_{i=1}^{2}}
{\displaystyle\sum\limits_{j=1}^{2}}
\overline{p}_{ij}\log\frac{\overline{p}_{ij}}{p_{ij}%
(\widehat{\boldsymbol{\theta}})}=2%
{\displaystyle\sum\limits_{i=1}^{2}}
{\displaystyle\sum\limits_{j=1}^{2}}
n_{ij}\log\frac{n_{ij}/n_{i}}{\pi_{ij}(\widehat{\boldsymbol{\theta}})}=2%
{\displaystyle\sum\limits_{i=1}^{2}}
{\displaystyle\sum\limits_{j=1}^{2}}
n_{ij}\log\frac{n_{ij}/n_{i}}{n_{\bullet j}/n}. \label{G bar}%
\end{equation}
For simplicity the case of $J=2$ is taken into account, where the (simple
null) one sided test%
\begin{equation}
H_{0}:\;\vartheta_{1}=1\text{,\qquad vs.}\qquad H_{1}:\;\vartheta
_{1}>1\text{,} \label{tt1}%
\end{equation}
with $\vartheta_{1}=\pi_{11}\pi_{22}/\pi_{21}\pi_{12}=\pi_{11}(1-\pi_{21}%
)/\pi_{21}(1-\pi_{11})$, or
\begin{subequations}
\[
H_{0}:\;\pi_{11}=\pi_{21}\text{,\qquad vs.}\qquad H_{1}:\;\pi_{11}>\pi
_{21}\text{,}%
\]
is tested with (\ref{G 1}), and on the other hand the two sided test
\end{subequations}
\begin{equation}
H_{0}:\;\vartheta_{1}=1\text{,\qquad vs.}\qquad H_{1}:\;\vartheta_{1}%
\neq1\text{,} \label{tt2b}%
\end{equation}
or
\begin{subequations}
\[
H_{0}:\;\pi_{11}=\pi_{21}\text{,\qquad vs.}\qquad H_{1}:\;\pi_{11}\neq\pi
_{21}\text{,}%
\]
is carried out with (\ref{G bar}). The same procedure would be possible to
perform for any $\phi$-divergence based test considered in this paper. We also
consider the mid-rank Wilcoxon test for both version of the alternative
hypothesis. To clarify the parameter space in both tests, we shall rewrite
(\ref{tt1}) and (\ref{tt2b}) as follows
\end{subequations}
\begin{subequations}
\[
H_{0}:\;\vartheta_{1}\in\Psi_{0}\text{,\qquad vs.}\qquad H_{1}:\;\vartheta
_{1}\in\Psi_{1}\text{,}%
\]
where $\Psi_{0}=\{1\}$, $\Psi_{1}=(1,+\infty)$,
\end{subequations}
\begin{subequations}
\[
H_{0}:\;\vartheta_{1}\in\Psi_{0}\text{,\qquad vs.}\qquad H_{1}^{\prime
}:\;\vartheta_{1}\in\Psi_{1}^{\prime}\text{,}%
\]
where $\Psi_{1}^{\prime}=(-\infty,1)\cup(1,+\infty)$. The parameter spaces for
(\ref{tt1}) and (\ref{tt2b})\ are $\Psi=\Psi_{0}\cup\Psi_{1}=[1,+\infty)$ and
$\Psi^{\prime}=\Psi_{0}\cup\Psi_{1}^{\prime}=%
\mathbb{R}
$, respectively. The same hypotheses in term of probabilities are given by
\end{subequations}
\begin{subequations}
\[
H_{0}:\;(\pi_{11},\pi_{21})\in\Lambda_{0}\text{,\qquad vs.}\qquad H_{1}%
:\;(\pi_{11},\pi_{21})\in\Lambda_{1}\text{,}%
\]
where $\Lambda_{0}=\left\{  (\pi_{11},\pi_{21})\in(0,1)\times(0,1):\pi
_{11}=\pi_{21}\right\}  $, $\Lambda_{1}=\left\{  (\pi_{11},\pi_{21}%
)\in(0,1)\times(0,1):\pi_{11}>\pi_{21}\right\}  $,
\end{subequations}
\begin{subequations}
\[
H_{0}:\;(\pi_{11},\pi_{21})\in\Lambda_{0}\text{,\qquad vs.}\qquad H_{1}%
:\;(\pi_{11},\pi_{21})\in\Lambda_{1}^{\prime}\text{,}%
\]
where $\Lambda_{1}^{\prime}=\left\{  (\pi_{11},\pi_{21})\in(0,1)\times
(0,1):\pi_{11}\neq\pi_{21}\right\}  $. The corresponding parameter spaces in
term of probabilities are given by
\end{subequations}
\begin{align*}
\Lambda &  =\Lambda_{0}\cup\Lambda_{1}=\left\{  (\pi_{11},\pi_{21}%
)\in(0,1)\times(0,1):\pi_{11}\geq\pi_{21}\right\}  ,\\
\Lambda^{\prime}  &  =\Lambda_{0}\cup\Lambda_{1}^{\prime}=(0,1)\times(0,1).
\end{align*}

The likelihood ratio test-statistics for (\ref{tt1}) and (\ref{tt2b}) are
different since in the numerator of (\ref{G 1}), $\pi_{ij}%
(\widetilde{\boldsymbol{\theta}})$, is obtained maximizing the likelihood
function in $\Lambda$, while the numerator of (\ref{G bar}), $n_{ij}/n$, is
maximized in $\Lambda^{\prime}$.\ Even though both estimators are different,
in practice they require a similar computation:\newline$\bullet$ If $\bar{\pi
}_{11}=\frac{n_{11}}{n_{1}}>\bar{\pi}_{21}=\frac{n_{21}}{n_{2}}$, then
$\pi_{11}(\widetilde{\boldsymbol{\theta}})=\frac{n_{11}}{n_{1}}>\pi
_{21}(\widetilde{\boldsymbol{\theta}})=\frac{n_{21}}{n_{2}}$ and $G^{2}%
=\bar{G}^{2}=2%
{\textstyle\sum\nolimits_{i=1}^{2}}
{\textstyle\sum\nolimits_{j=1}^{2}}
n_{ij}\log\frac{n_{ij}/n_{i}}{n_{\bullet j}/n}$;\newline$\bullet$ If $\bar
{\pi}_{11}=\frac{n_{11}}{n_{1}}\leq\bar{\pi}_{21}=\frac{n_{21}}{n_{2}}$, then
$\pi_{11}(\widetilde{\boldsymbol{\theta}})=\pi_{11}%
(\widehat{\boldsymbol{\theta}})=\frac{n_{\bullet1}}{n}\leq\pi_{21}%
(\widetilde{\boldsymbol{\theta}})=\pi_{21}(\widehat{\boldsymbol{\theta}%
})=\frac{n_{\bullet1}}{n}$ and $G^{2}=0$.\newline Hence, taking into account
the asymptotic distributions, i.e. $\frac{1}{2}\chi_{0}^{2}+\frac{1}{2}%
\chi_{1}^{2}$ for (\ref{tt1}) and $\chi_{1}^{2}$\ for (\ref{tt2b}), we obtain%
\[
p\mathrm{-}value(G^{2})=\left\{
\begin{array}
[c]{ll}%
\frac{1}{2}\Pr\left(  \chi_{1}^{2}>2%
{\displaystyle\sum\limits_{i=1}^{2}}
{\displaystyle\sum\limits_{j=1}^{2}}
n_{ij}\log\frac{n_{ij}/n_{i}}{n_{\bullet j}/n}\right)  , & \text{if }%
\frac{n_{11}}{n_{1}}>\frac{n_{21}}{n_{2}},\\
1, & \text{if }\frac{n_{11}}{n_{1}}\leq\frac{n_{21}}{n_{2}},
\end{array}
\right.
\]
and
\[
p\mathrm{-}value(\bar{G}^{2})=\Pr\left(  \chi_{1}^{2}>2%
{\displaystyle\sum\limits_{i=1}^{2}}
{\displaystyle\sum\limits_{j=1}^{2}}
n_{ij}\log\frac{n_{ij}}{n_{\bullet j}}\right)  .
\]
A third test is the composite null one sided test%
\begin{align}
H_{0}  &  :\;\vartheta_{1}\leq1\text{,\quad(}\vartheta_{1}\in\Psi_{0}^{\prime
}\text{)\qquad vs.}\qquad H_{1}:\;\vartheta_{1}>1\text{,}\quad\text{(}%
\vartheta_{1}\in\Psi_{1}\text{)}\label{tt3}\\
H_{0}  &  :\;\pi_{11}\leq\pi_{21}\text{,}\quad\text{(}(\pi_{11},\pi_{21}%
)\in\Lambda_{0}^{\prime}\text{)\qquad vs.}\qquad H_{1}:\;\pi_{11}>\pi
_{21}\text{,}\quad\text{(}(\pi_{11},\pi_{21})\in\Lambda_{1}\text{),}\nonumber
\end{align}
with $\Psi_{0}^{\prime}=(-\infty,1]$ and $\Lambda_{0}^{\prime}=\left\{
(\pi_{11},\pi_{21})\in(0,1)\times(0,1):\pi_{11}\leq\pi_{21}\right\}  $. For
the corresponding test-statistic,%
\begin{equation}
\widetilde{G}^{2}=2n%
{\displaystyle\sum\limits_{i=1}^{2}}
{\displaystyle\sum\limits_{j=1}^{2}}
\overline{p}_{ij}\log\frac{\overline{p}_{ij}}{p_{ij}%
(\widetilde{\boldsymbol{\theta}})}=2%
{\displaystyle\sum\limits_{i=1}^{2}}
{\displaystyle\sum\limits_{j=1}^{2}}
n_{ij}\log\frac{n_{ij}/n_{i}}{\pi_{ij}(\widetilde{\boldsymbol{\theta}})}:
\label{G tilde}%
\end{equation}
$\bullet$ If $\bar{\pi}_{11}=\frac{n_{11}}{n_{1}}\geq\bar{\pi}_{21}%
=\frac{n_{21}}{n_{2}}$, then $\pi_{11}(\widetilde{\boldsymbol{\theta}}%
)=\frac{n_{\bullet1}}{n}\geq\pi_{12}(\widetilde{\boldsymbol{\theta}}%
)=\frac{n_{\bullet1}}{n}$ and $\widetilde{G}^{2}=2%
{\textstyle\sum\nolimits_{i=1}^{2}}
{\textstyle\sum\nolimits_{j=1}^{2}}
n_{ij}\log\frac{n_{ij}/n_{i}}{n_{\bullet j}/n}$;\newline$\bullet$ If $\bar
{\pi}_{11}=\frac{n_{11}}{n_{1}}<\bar{\pi}_{21}=\frac{n_{21}}{n_{2}}$, then
$\pi_{11}(\widetilde{\boldsymbol{\theta}})=\frac{n_{11}}{n_{1}}<\pi
_{21}(\widetilde{\boldsymbol{\theta}})=\frac{n_{21}}{n_{2}}$ and
$\widetilde{G}^{2}=0$.\newline Hence, both one sided test-statistics, the
composite null one, $\widetilde{G}^{2}$, and the simple null one, $G^{2}$, are
almost equal and
\[
p\mathrm{-}value(\widetilde{G}^{2})=\left\{
\begin{array}
[c]{ll}%
\frac{1}{2}\Pr\left(  \chi_{1}^{2}>2%
{\displaystyle\sum\limits_{i=1}^{2}}
{\displaystyle\sum\limits_{j=1}^{2}}
n_{ij}\log\frac{n_{ij}/n_{i}}{n_{\bullet j}/n}\right)  , & \text{if }%
\frac{n_{11}}{n_{1}}\geq\frac{n_{21}}{n_{2}},\\
1, & \text{if }\frac{n_{11}}{n_{1}}<\frac{n_{21}}{n_{2}}.
\end{array}
\right.
\]

The mid-rank $W$ test-statistic for (\ref{tt1}) and (\ref{tt2b}) is the same,
(\ref{wilc}), as well as the distribution under the null, but%
\[
p\mathrm{-}value(W)=\Pr\left(  \mathcal{N}(0,1)<-\frac{\left(  r_{1}%
n_{11}+r_{2}n_{12}\right)  -\left.  n_{1}\left(  n+1\right)  \right/  2}%
{\sqrt{n_{1}n_{2}\frac{n+1-\frac{1}{n(n-1)}\left[  \left(  (n_{\bullet1}%
^{3}-n_{\bullet1})\right)  +\left(  (n_{\bullet2}^{3}-n_{\bullet2})\right)
\right]  }{12}}}\right)
\]
for (\ref{tt1}) and
\[
p\mathrm{-}value(W)=2\Pr\left(  \mathcal{N}(0,1)>\frac{\left\vert \left(
r_{1}n_{11}+r_{2}n_{12}\right)  -\left.  n_{1}\left(  n+1\right)  \right/
2\right\vert }{\sqrt{n_{1}n_{2}\frac{n+1-\frac{1}{n(n-1)}\left[  \left(
(n_{\bullet1}^{3}-n_{\bullet1})\right)  +\left(  (n_{\bullet2}^{3}%
-n_{\bullet2})\right)  \right]  }{12}}}\right)
\]
\ for (\ref{tt2b}).%

\begin{figure}[htbp]  \tabcolsep2.8pt  \centering
\begin{tabular}
[c]{c}%
{\includegraphics[
height=2.8764in,
width=5.6273in
]%
{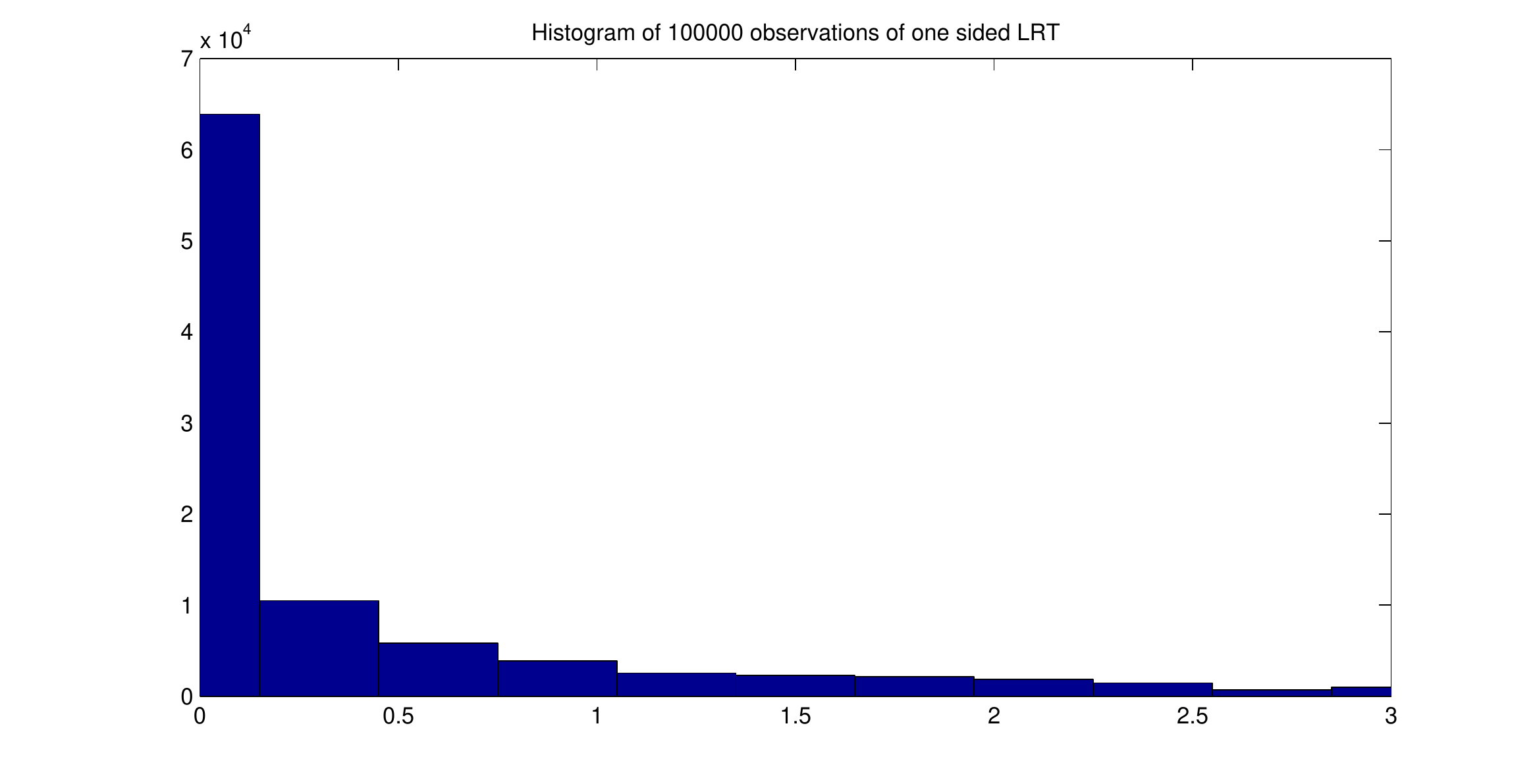}%
}
\\%
{\includegraphics[
height=2.8764in,
width=5.6273in
]%
{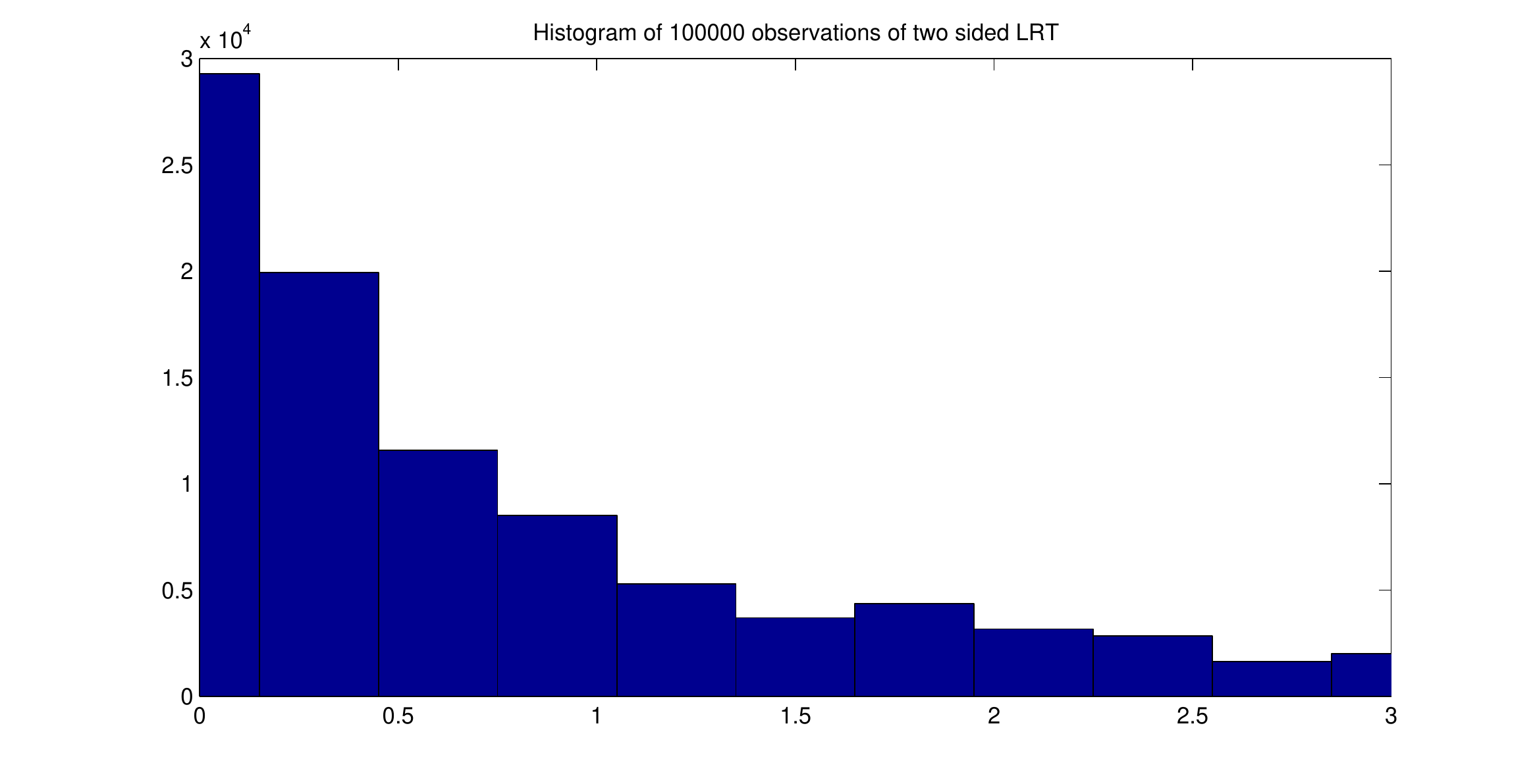}%
}
\\%
{\includegraphics[
height=2.8764in,
width=5.6273in
]%
{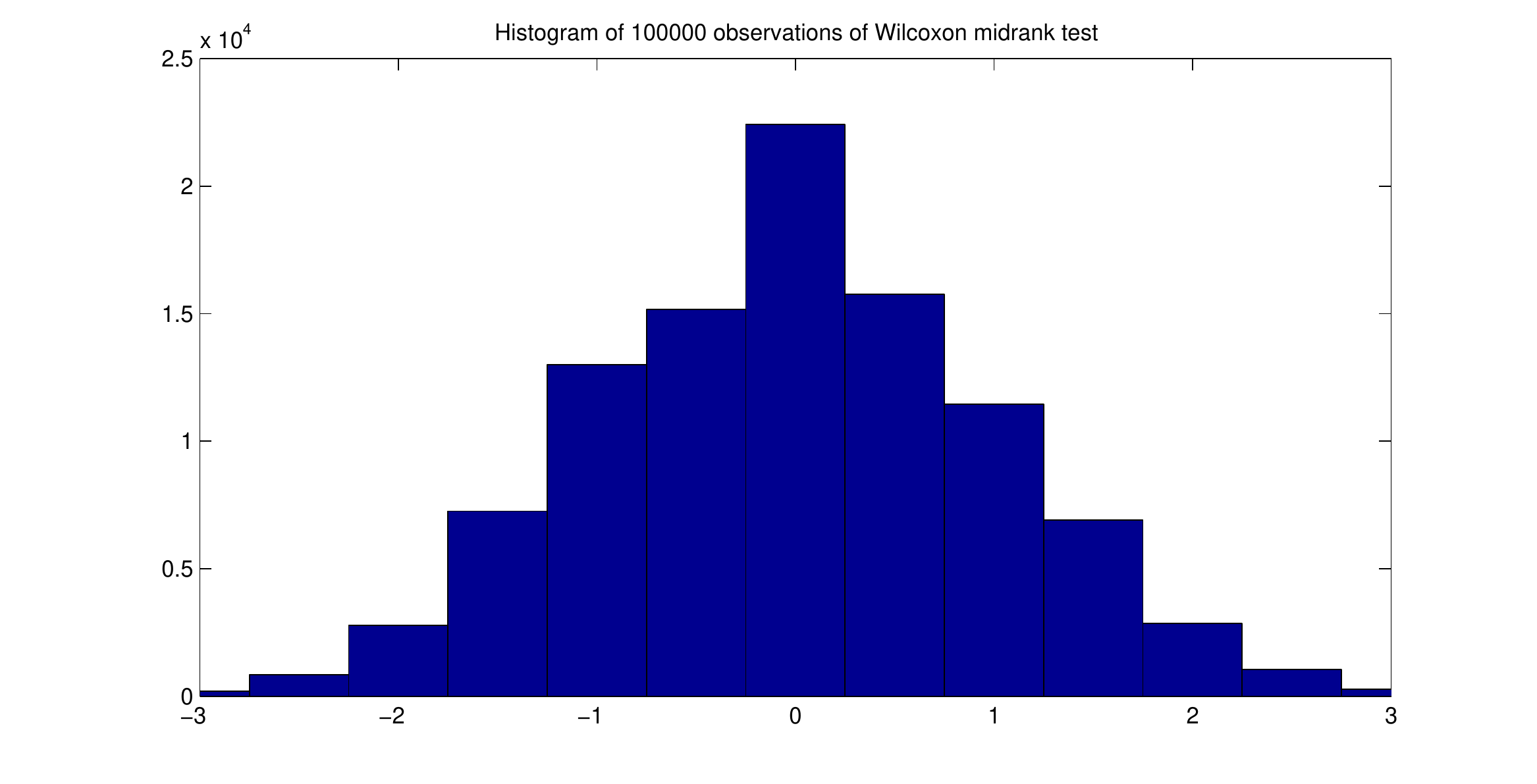}%
}
\end{tabular}
\caption{Histograms of $G^{2}$, $\bar{G}^{2}$ and $W$ with $n_1=40$, $n_2=20$ and $\pi _{i1}=0.35$, $i=1,2$. \label{figHH}}%
\end{figure}%

The following short simulation study considers $R=100,000$\ realizations,
$n_{i1}^{(h)}$, $i=1,2$, $h=1,...,R$, of%
\[
N_{i1}\overset{ind}{\sim}\mathcal{B}in(n_{i},\pi_{i1}),\qquad i=1,2,
\]
with $\pi_{11}=\pi_{21}=0.35$ and $n_{1}=40$ and $n_{2}=20$. In Figure
\ref{figHH} a histogram of $G^{2}$, $\bar{G}^{2}$ and $W$\ is shown where the
shape of the density function of each can be recognized. In Table \ref{ttHH},
the simulated significance levels ($\widehat{\alpha}$) and powers
($\widehat{\beta}$) are calculated as the proportion of statistics with
$p$-values smaller than the nominal level $\alpha=0.05$. The test-statistic
based on the Hellinger distance $S_{-1/2}$, given in (\ref{hel}), is also
included. From this simulation study it is concluded that the $G^{2}$
likelihood ratio test-statistic and the $W$\ Wilcoxon mid-rank test for
$2\times2$ contingency tables, are specific procedures for the one sided test
(\ref{tt1}) since the parameter spaces are different, but are strongly related
with the two sided test (\ref{tt2b}) in the way of calculating the value of
the test-statistic and the corresponding $p$-value. It is remarkable that the
simulated significance level for the one-sided $W$\ Wilcoxon mid-rank test for
$2\times2$ contingency tables exhibits a slightly better approximation of the
nominal level\ in comparison with the likelihood ratio test $G^{2}$ for the
one sided test (\ref{tt1}), and the likelihood ratio test $G^{2}$ slightly
better than the test-statistic based on the Hellinger distance $S_{-1/2}$. The
powers of the test-statistics are calculated for $\pi_{11}=0.45>\pi_{21}%
=0.35$. The test-statistic based on the Hellinger distance $S_{-1/2}$ has the
greatest power and the $W$\ Wilcoxon mid-rank test the smallest power for the
one sided test (\ref{tt1}). In Section \ref{Sim} a more extensive simulation
study is considered with a criterion to select the best test-statistic within
a broader class of power divergence based test-statistics. Finally, the two
sided test-statistics, $\bar{G}^{2}$ and $W$, exhibit a worse power than the
one sided test-statistics. This behaviour was obviously expected, since being
$\Psi\subset\Psi^{\prime}$ or equivalently $\Lambda\subset\Lambda^{\prime}$,
the one sided tests have always a better power than the two sided tests.%

\begin{table}[htbp]  \tabcolsep2.8pt  \centering
\begin{tabular}
[c]{cccccccccccc}\hline
&  & $S_{-1/2}$ (one sided) &  & $G^{2}$ (one sided) &  & $\bar{G}^{2}$ (two
sided) &  & one sided $W$ &  & two sided $W$ & \\\hline
$\widehat{\alpha}$ &  & $0.0567$ &  & $0.0559$ &  & $0.0533$ &  & $0.0495$ &
& $0.0489$ & \\
$\widehat{\beta}$ &  & $0.2027$ &  & $0.2025$ &  & $0.1186$ &  & $0.1865$ &  &
$0.1149$ & \\\hline
\end{tabular}
\caption{Simulated significance levels ($\pi _{11}=\pi _{21}=0.35$), $\widehat{\alpha
}$, and powers ($\pi _{11}=0.45$, $\pi _{21}=0.35$), $\widehat{\beta }$, for
$S_{-1/2}$, $G^{2}$, $\bar{G}^{2}$ and $W$ test-statistics with $n_{1}=40>n_{2}=20$.\label{ttHH}}%
\end{table}%

\subsection{Power divergence test-statistics: simulated size and
powers\label{Sim}}

In this Section the performance of the power divergence test statistics
(\ref{PD1})-(\ref{PD6}) is studied in terms of the simulated exact size and
simulated power of the test, based on small and moderate sample sizes. A
simulation experiment with seven scenarios is designed in Table \ref{tt},
taking into account the sample sizes of the two independent samples. The pairs
of scenarios (A,G), (B,F) and (C,E) should have very similar exact
significance levels, since the sample sizes of the two samples are symmetrical
(the ratio of one sample is the inverse of the other one). With respect to the
choice of $\lambda$, the parameters for the power divergence test statistics,
the interest is focused on the interval $[-1.5,3]$. Note that the
test-statistics applied in the numerical example are covered as particular cases.%

\begin{table}[htbp]   \centering
$%
\begin{tabular}
[c]{cccccccc}\hline
scenarios & sc. A & sc. B & sc. C & sc. D & sc. E & sc. F & sc. G\\\hline
$n_{1}$ & $20$ & $20$ & $20$ & $20$ & $16$ & $10$ & $4$\\
$n_{2}$ & $4$ & $10$ & $16$ & $20$ & $20$ & $20$ & $20$\\\hline
ratio & $5$ & $2$ & $1.25$ & $1$ & $0.8$ & $0.5$ & $0.2$\\\hline
\end{tabular}
\ \ \ \ \ \ \ \ \ $%
\caption{Scenarios,  based on sample sizes, for the simulation stydy in a contingency table $2\times 3$.\label{tt}}
\end{table}%
\bigskip

The algorithm described in Section \ref{sec:Numerical example} is taken into
account to calculate the $p$-value of each test-statistic ${T\in\{T_{\lambda
},S_{\lambda}\}}_{\lambda\in\lbrack-1.5,3]}$, with a sample $\boldsymbol{N}$,
and this is repeated independently $R=25\,000$ times. The simulated exact
power was computed as%
\[
{\widehat{\beta}}_{T}={\widehat{\beta}}_{T}(\delta)=\frac{\text{number of
replications of }T\,\text{for which the }p\text{-value is less than }\alpha
}{R},
\]
for the probability vectors%
\begin{align*}
\boldsymbol{\pi}_{i}(\boldsymbol{\theta}(\delta))  &  =(\pi_{i1}%
(\boldsymbol{\theta}(\delta)),\pi_{i2}(\boldsymbol{\theta}(\delta)),\pi
_{i3}(\boldsymbol{\theta}(\delta)))^{T}\\
\pi_{ij}(\boldsymbol{\theta}(\delta))  &  =\frac{1}{3}\frac{1+i(j-1)\delta
}{1+i\delta},\quad i=1,2,\quad j=1,2,3,
\end{align*}
for $\delta\in\Xi=\{0.1,0.5,1.0,1.5\}$. The simulated exact size was computed
as%
\[
{\widehat{\alpha}}_{T}=\frac{\text{number of replications of }T\,\text{for
which the }p\text{-value is less than }\alpha}{R},
\]
for the probability vectors
\begin{align*}
\boldsymbol{\pi}_{i}(\boldsymbol{\theta}_{0})  &  =(\pi_{i1}%
(\boldsymbol{\theta}_{0}),\pi_{i2}(\boldsymbol{\theta}_{0}),\pi_{i3}%
(\boldsymbol{\theta}_{0}))^{T}\\
\pi_{ij}(\boldsymbol{\theta}_{0})  &  =\frac{1}{3},\quad i=1,2,\quad j=1,2,3,
\end{align*}
which corresponds to the case of $\delta=0$ for $\boldsymbol{\pi}%
_{i}(\boldsymbol{\theta}(\delta))$.

In Table \ref{tt2} the local odds ratios,%
\[
\vartheta_{j}=\vartheta_{j}(\delta)=\frac{1+(j-1)\delta}{1+2(j-1)\delta}%
\frac{1+2j\delta}{1+j\delta},
\]
$j=1,2$, are shown for $\delta\in\{0\}\cup\Xi$. Notice that in
$\boldsymbol{\vartheta}=\boldsymbol{\vartheta}(\delta)=(\vartheta_{1}%
(\delta),\vartheta_{2}(\delta))^{T}$\ some of the components are further from
$\boldsymbol{\vartheta}(0)=\boldsymbol{1}_{2}$ (null hypothesis), as the value
of $\delta>0$ is further from $0$. This means that a greater value of the
estimation of the power function might be obtained, as $\delta>0$ is greater.
This claim is supported by the fact that some values of the components of
$\boldsymbol{\vartheta}=\boldsymbol{\vartheta}(\delta)$ decrease as $\delta>0$
increases but more slowly than the others increase. In addition, for a fixed
value of $\delta>0$, it is expected a greater value of $\widehat{\beta}%
_{T}(\delta)$, as $n$ is greater (the worst powers in Scenario A and the best
powers in Scenario D). We have also added in Table \ref{tt2}\ the last three
rows for two reasons, first, to show that for any fixed value of $\delta$,
$\pi_{2j}(\boldsymbol{\theta}(\delta))/\pi_{1j}(\boldsymbol{\theta}(\delta))$
is non-decreasing as $j$, the ordinal category, increases and second, to
clarify the meaning of the two asterisks contained in the table. It is clear
that for a big value of $\delta$, $\pi_{i1}(\boldsymbol{\theta}(\delta))>0$
goes to zero on the right for $i=1,2$, but in the practice, due to the empty
cells in the contingency table, the estimator of the ratio $\pi_{21}%
(\boldsymbol{\theta}(\delta))/\pi_{11}(\boldsymbol{\theta}(\delta))$ becomes
$1$ rather than $\frac{1}{2}$ (and $\vartheta_{1}(\delta)$ becomes $1$). This
was our experience when we used values of $\delta$ bigger than $1.5$, i.e. the
power becomes quite little in the practice.%

\begin{table}[htbp]  \tabcolsep2.8pt  \centering
\begin{tabular}
[c]{ccccccccccccc}\hline
&  & $\delta=0$ &  & $\delta=0.1$ &  & $\delta=0.5$ &  & $\delta=1$ &  &
$\delta=1.5$ &  & $\delta=\infty$\\\hline
$\vartheta_{1}=\vartheta_{1}(\delta)$ &  & $1.000$ &  & $1.091$ &  & $1.333$ &
& $1.500$ &  & $1.600$ &  & $2.00^{\ast}$\\
$\vartheta_{2}=\vartheta_{2}(\delta)$ &  & $1.000$ &  & $1.069$ &  & $1.125$ &
& $1.111$ &  & $1.094$ &  & $1.00$\\\hline
$\pi_{21}(\boldsymbol{\theta}(\delta))/\pi_{11}(\boldsymbol{\theta}(\delta))$
&  & $0.33/0.33$ &  & $0.28/0.30$ &  & $0.17/0.22$ &  & $0.11/0.17$ &  &
$0.08/0.13$ &  & $0.50^{\ast}$\\
$\pi_{22}(\boldsymbol{\theta}(\delta))/\pi_{12}(\boldsymbol{\theta}(\delta))$
&  & $0.33/0.33$ &  & $0.33/0.33$ &  & $0.33/0.33$ &  & $0.33/0.33$ &  &
$0.33/0.33$ &  & $1.00$\\\hline
\end{tabular}
\caption{Theoretical local odd ratios for the Monte Carlo study.\label{tt2}}%
\end{table}%

Once a nominal size $\alpha=0.05$ is established, Table \ref{alfas} summarizes
the simulated exact sizes in all the scenarios for the test-statistic
${T\in\{T_{\lambda},S_{\lambda},W\}}_{\lambda\in\Lambda}$, with $\Lambda
=\{-1.5,-1,-\frac{1}{2},0,\frac{2}{3},1,1.5,2,3\}$. We have plotted $3\times2$
graphs in Figures \ref{fig2}-\ref{fig7} and we refer them as plots in three
rows. In the first row of Figures \ref{fig1}-\ref{fig7} we can see on the left
the exact power in all the scenarios for the test-statistic $\{{T_{\lambda
},W\}}_{\lambda\in\lbrack-1.5,3]}$ and on the right for the test-statistic
$\{{S_{\lambda},W\}}_{\lambda\in\lbrack-1.5,3]}$. In order to make a
comparison of exact powers, we cannot directly proceed without considering the
exact sizes. For this reason we are going to give a procedure based on two
steps, for scenarios B-G.

\noindent\textit{Step 1}: We are going to check for all the power divergence
based test-statistics the criterion given by Dale (1986), i.e.,
\begin{equation}
|\,\text{logit}(1-{\widehat{\alpha}}_{T})-\text{logit}(1-\alpha)\,|\leq e
\label{con1}%
\end{equation}
with $\mathrm{logit}\left(  p\right)  =\log\left(  \frac{p}{1-p}\right)  $. We
only consider the values of $\lambda$\ such that ${\widehat{\alpha}}_{T}%
$\ satisfies (\ref{con1}) with $e=0.35$, then we shall only consider the
test-statistics such that ${\widehat{\alpha}}_{T}\in\left[
0.0357,0.0695\right]  $, in all the scenarios. This criterion has been
considered for some authors, see for instance Cressie et al. (2003) and
Mart\'{\i}n and Pardo (2012). The cases satisfying the criterion are marked in
bold in Table \ref{alfas}, and comprise those values in the abscissa of the
plot between the dashed band (the dashed line in the middle represents the
nominal size), and we can conclude that we must not consider in our study
${T\in\{T_{\lambda},S_{\lambda},W\}}_{\lambda\in\lbrack-1.5,-0.4)}$.

\noindent\textit{Step 2}: We compare all the test statistics obtained in Step
1 with the classical likelihood ratio test ($G^{2}=T_{0}$) as well as the
classical Pearson test statistic ($X^{2}=S_{1}$). To do so, we have calculated
the relative local efficiencies%
\[
\widehat{\rho}_{T}=\widehat{\rho}_{T}(\delta)=\frac{({\widehat{\beta}}%
_{T}(\delta)-{\widehat{\alpha}}_{T})-({\widehat{\beta}}_{T_{0}}(\delta
)-{\widehat{\alpha}}_{T_{0}})}{{\widehat{\beta}}_{T_{0}}(\delta
)-{\widehat{\alpha}}_{T_{0}}},\qquad\widehat{\rho}_{T}^{\ast}=\widehat{\rho
}_{T}^{\ast}(\delta)=\frac{({\widehat{\beta}}_{T}(\delta)-{\widehat{\alpha}%
}_{T})-({\widehat{\beta}}_{S_{1}}(\delta)-{\widehat{\alpha}}_{S_{1}}%
)}{{\widehat{\beta}}_{S_{1}}(\delta)-{\widehat{\alpha}}_{S_{1}}}.
\]
In Figures \ref{fig2}-\ref{fig7} the powers and the relative local
efficiencies are summarized. The second rows of the figures represent
$\widehat{\rho}_{T}$, while in the third row is plotted $\widehat{\rho}%
_{T}^{\ast}$, on the left it is considered ${T\in}\{{T_{\lambda},W\}}%
_{\lambda\in\lbrack-1.5,3]}$ and ${T\in}\{{S_{\lambda},W\}}_{\lambda\in
\lbrack-1.5,3]}$ on the right. In Figure \ref{fig1} we show only one row since
it represents the atypical case in which the exact powers are less that the
exact significance level for the values of $\lambda$ satisfying the Dale's
criterion and so, it does not make sense to compare the powers.%

\begin{table}[htbp]  \tabcolsep2.8pt  \centering
\begin{tabular}
[c]{l}%
$%
\begin{tabular}
[c]{ccccccccccc}\hline\hline
sc & ${\widehat{\alpha}}_{T_{-1.5}}$ & ${\widehat{\alpha}}_{T_{-1}}$ &
${\widehat{\alpha}}_{T_{-1/2}}$ & ${\widehat{\alpha}}_{T_{0}}$ &
${\widehat{\alpha}}_{T_{2/3}}$ & ${\widehat{\alpha}}_{T_{1}}$ &
${\widehat{\alpha}}_{T_{1.5}}$ & ${\widehat{\alpha}}_{T_{2}}$ &
${\widehat{\alpha}}_{T_{3}}$ & ${\widehat{\alpha}}_{W}$\\\hline
$A$ & 0.0013 & 0.0359 & 0.1725 & 0.0745 & \textbf{0.0468} & \textbf{0.0460} &
\textbf{0.0517} & \textbf{0.0586} & 0.0949 & \textbf{0.0509}\\
$B$ & \textbf{0.0670} & \textbf{0.0612} & \textbf{0.0664} & \textbf{0.0597} &
\textbf{0.0541} & \textbf{0.0503} & \textbf{0.0511} & \textbf{0.0536} &
\textbf{0.0619} & \textbf{0.0509}\\
$C$ & 0.0747 & \textbf{0.0686} & \textbf{0.0608} & \textbf{0.0537} &
\textbf{0.0494} & \textbf{0.0485} & \textbf{0.0478} & \textbf{0.0492} &
\textbf{0.0573} & \textbf{0.0485}\\
$D$ & \textbf{0.0688} & \textbf{0.0653} & \textbf{0.0631} & \textbf{0.0577} &
\textbf{0.0538} & \textbf{0.0528} & \textbf{0.0522} & \textbf{0.0530} &
\textbf{0.0572} & \textbf{0.0495}\\
$E$ & 0.0751 & \textbf{0.0691} & \textbf{0.0610} & \textbf{0.0548} &
\textbf{0.0511} & \textbf{0.0502} & \textbf{0.0494} & \textbf{0.0509} &
\textbf{0.0591} & \textbf{0.0512}\\
$F$ & \textbf{0.0665} & \textbf{0.0614} & \textbf{0.0681} & \textbf{0.0616} &
\textbf{0.0554} & \textbf{0.0518} & \textbf{0.0517} & \textbf{0.0539} &
\textbf{0.0615} & \textbf{0.0506}\\
$G$ & 0.0013 & 0.0363 & 0.1802 & 0.0775 & \textbf{0.0477} & \textbf{0.0466} &
\textbf{0.0526} & \textbf{0.0602} & 0.0965 & \textbf{0.0541}\\\hline\hline
\end{tabular}
\ \ $\\
$%
\begin{tabular}
[c]{ccccccccccc}\hline\hline
sc & ${\widehat{\alpha}}_{S_{-1.5}}$ & ${\widehat{\alpha}}_{S_{-1}}$ &
${\widehat{\alpha}}_{S_{-1/2}}$ & ${\widehat{\alpha}}_{S_{0}}$ &
${\widehat{\alpha}}_{S_{2/3}}$ & ${\widehat{\alpha}}_{S_{1}}$ &
${\widehat{\alpha}}_{S_{1.5}}$ & ${\widehat{\alpha}}_{S_{2}}$ &
${\widehat{\alpha}}_{S_{3}}$ & ${\widehat{\alpha}}_{W}$\\\hline
$A$ & 0.2106 & 0.2055 & 0.1572 & 0.0745 & \textbf{0.0429} & \textbf{0.0430} &
\textbf{0.0499} & \textbf{0.0507} & 0.0752 & \textbf{0.0509}\\
$B$ & 0.0799 & 0.0762 & \textbf{0.0638} & \textbf{0.0596} & \textbf{0.0543} &
\textbf{0.0497} & \textbf{0.0509} & \textbf{0.0524} & \textbf{0.0584} &
\textbf{0.0509}\\
$C$ & 0.0729 & \textbf{0.0676} & \textbf{0.0581} & \textbf{0.0537} &
\textbf{0.0505} & \textbf{0.0492} & \textbf{0.0491} & \textbf{0.0501} &
\textbf{0.0583} & \textbf{0.0485}\\
$D$ & \textbf{0.0675} & \textbf{0.0656} & \textbf{0.0620} & \textbf{0.0577} &
\textbf{0.0552} & \textbf{0.0543} & \textbf{0.0541} & \textbf{0.0543} &
\textbf{0.0577} & \textbf{0.0495}\\
$E$ & 0.0745 & \textbf{0.0683} & \textbf{0.0584} & \textbf{0.0547} &
\textbf{0.0518} & \textbf{0.0507} & \textbf{0.0504} & \textbf{0.0515} &
\textbf{0.0598} & \textbf{0.0512}\\
$F$ & 0.0814 & 0.0780 & \textbf{0.0656} & \textbf{0.0616} & \textbf{0.0551} &
\textbf{0.0509} & \textbf{0.0516} & \textbf{0.0528} & \textbf{0.0572} &
\textbf{0.0506}\\
$G$ & 0.2170 & 0.2123 & 0.1653 & 0.0775 & \textbf{0.0446} & \textbf{0.0450} &
\textbf{0.0510} & \textbf{0.0516} & 0.0782 & \textbf{0.0541}\\\hline\hline
\end{tabular}
\ \ \ $%
\end{tabular}
\caption{${\widehat{\alpha}}_{T}$, for ${T\in\{T_{\lambda},S_{\lambda},W\}}_{\lambda\in\Lambda}$ in scenarios of Table \ref{tt}. \label{alfas}}%
\end{table}%
%

\begin{figure}[htbp]  \tabcolsep2.8pt  \centering
\begin{tabular}
[c]{cc}%
${T_{\lambda}}$ & ${S_{\lambda}}$\\%
{\includegraphics[
height=2.463in,
width=3.3667in
]%
{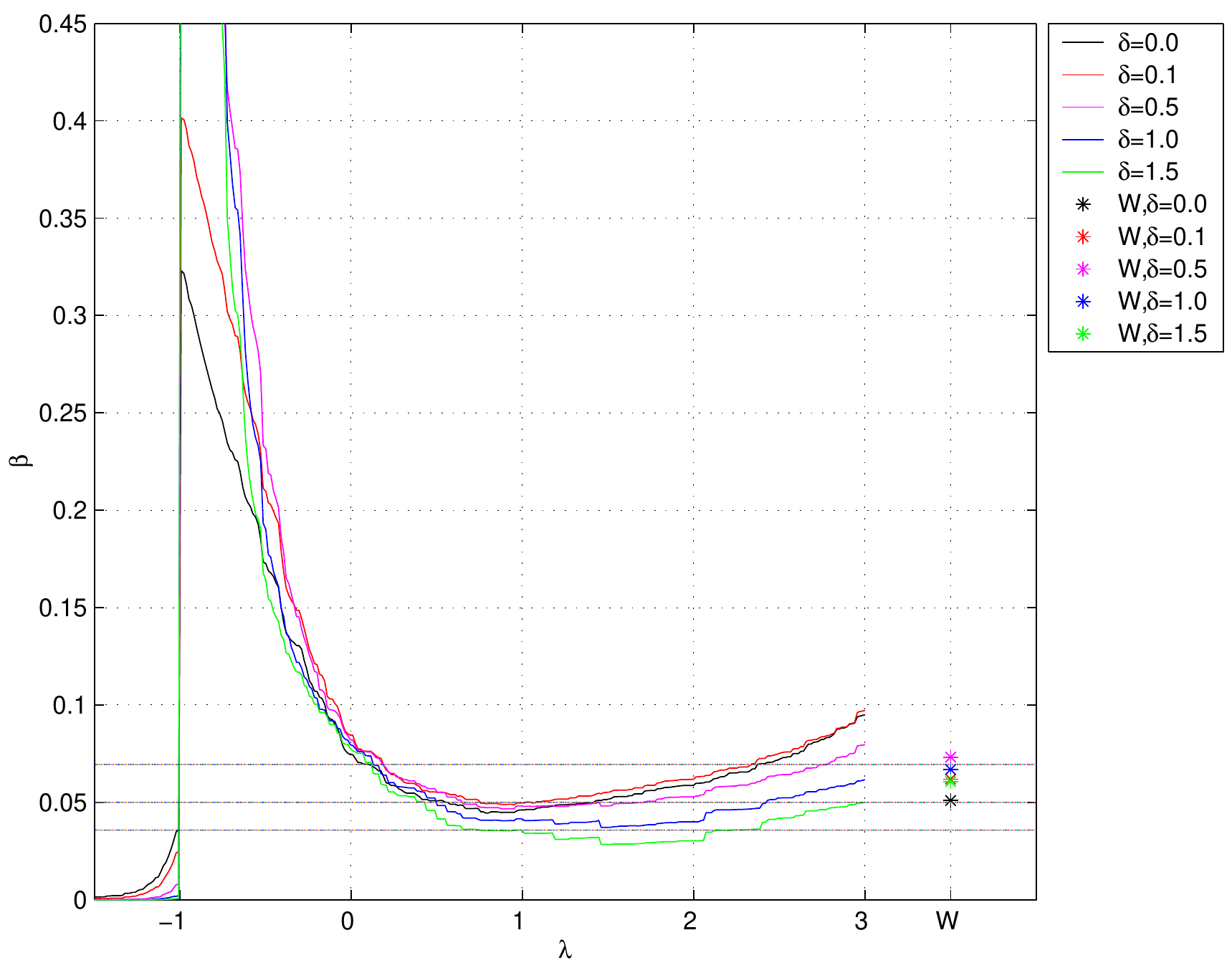}%
}
&
{\includegraphics[
height=2.463in,
width=3.3667in
]%
{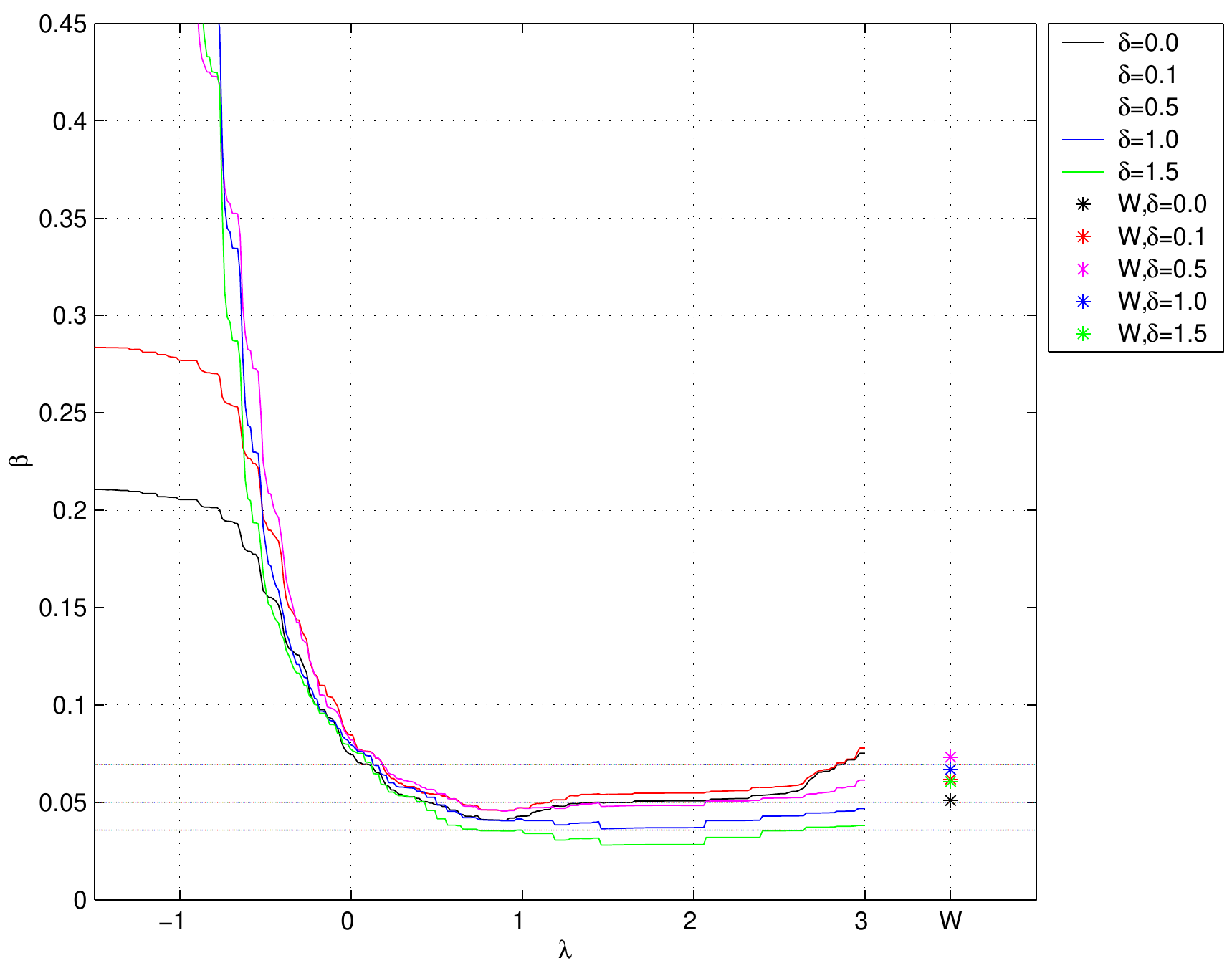}%
}
\end{tabular}
\caption{Powers for $T_{\lambda}$, $S_{\lambda}$  and $W$  in scenario A. \label{fig1}}%
\end{figure}%
%

\begin{figure}[htbp]  \tabcolsep2.8pt  \centering
\begin{tabular}
[c]{cc}%
${T_{\lambda}}$ & ${S_{\lambda}}$\\%
{\includegraphics[
height=2.4561in,
width=3.1202in
]%
{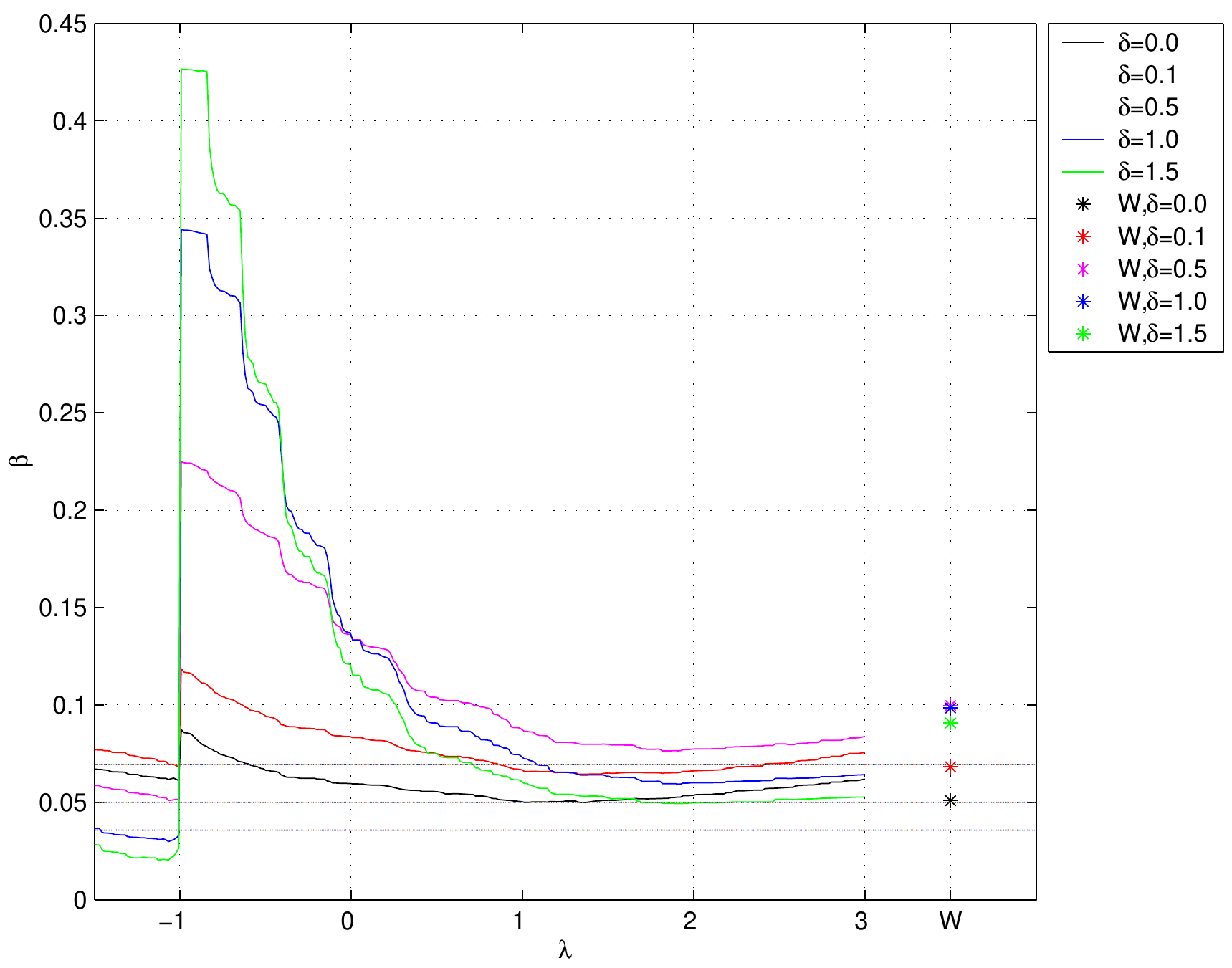}%
}
&
{\includegraphics[
height=2.4561in,
width=3.1202in
]%
{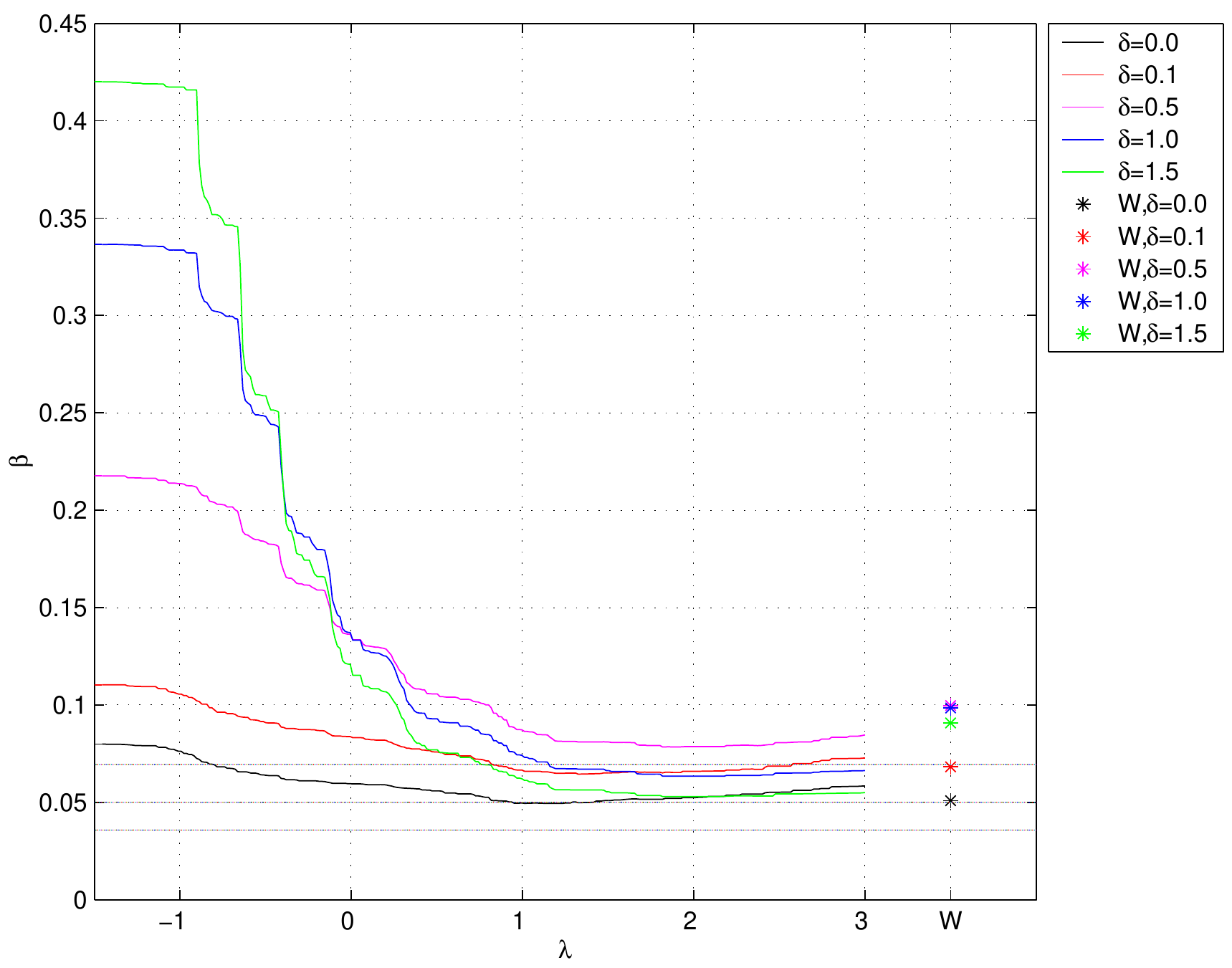}%
}
\\%
{\includegraphics[
height=2.4561in,
width=3.0701in
]%
{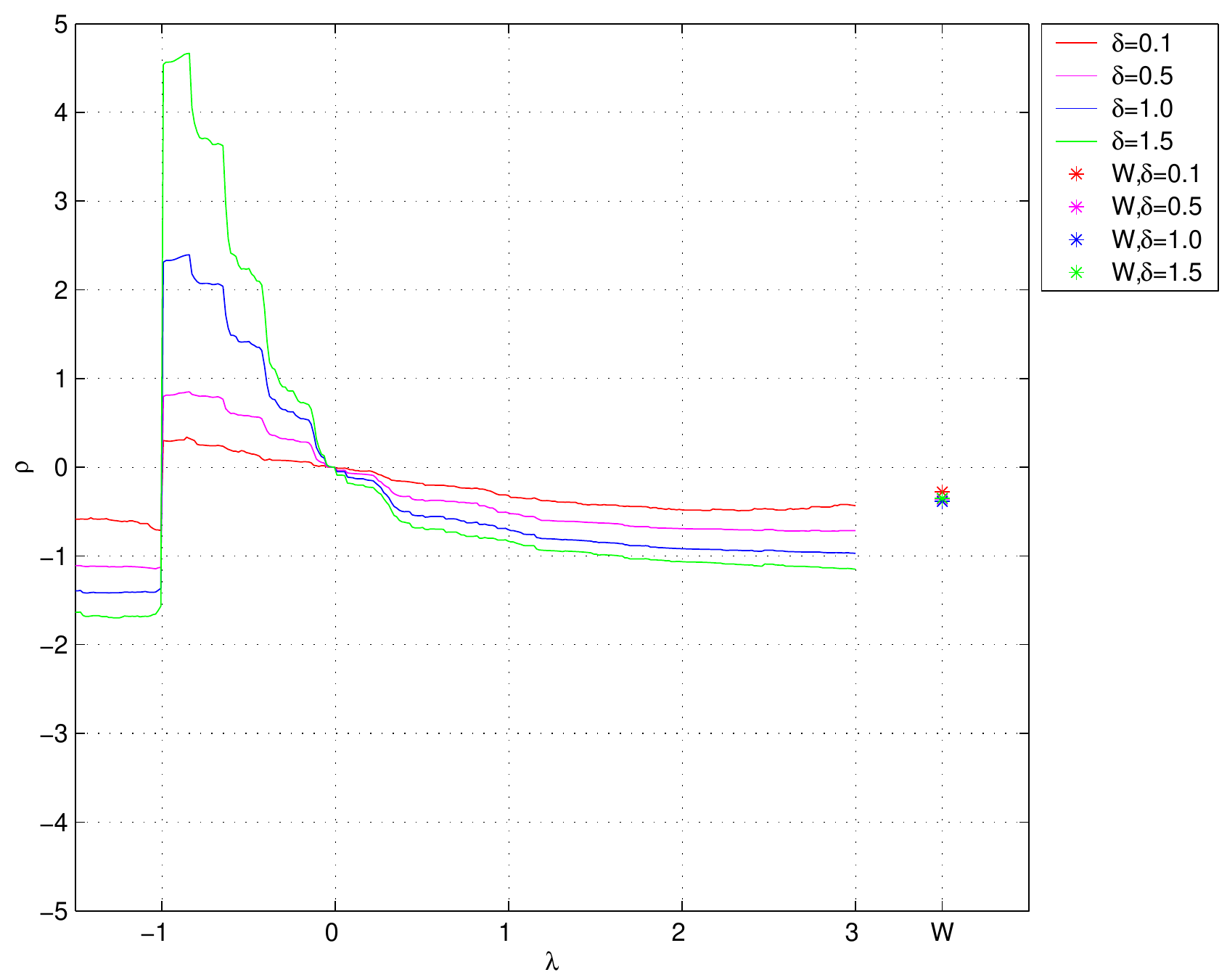}%
}
&
{\includegraphics[
height=2.4561in,
width=3.0701in
]%
{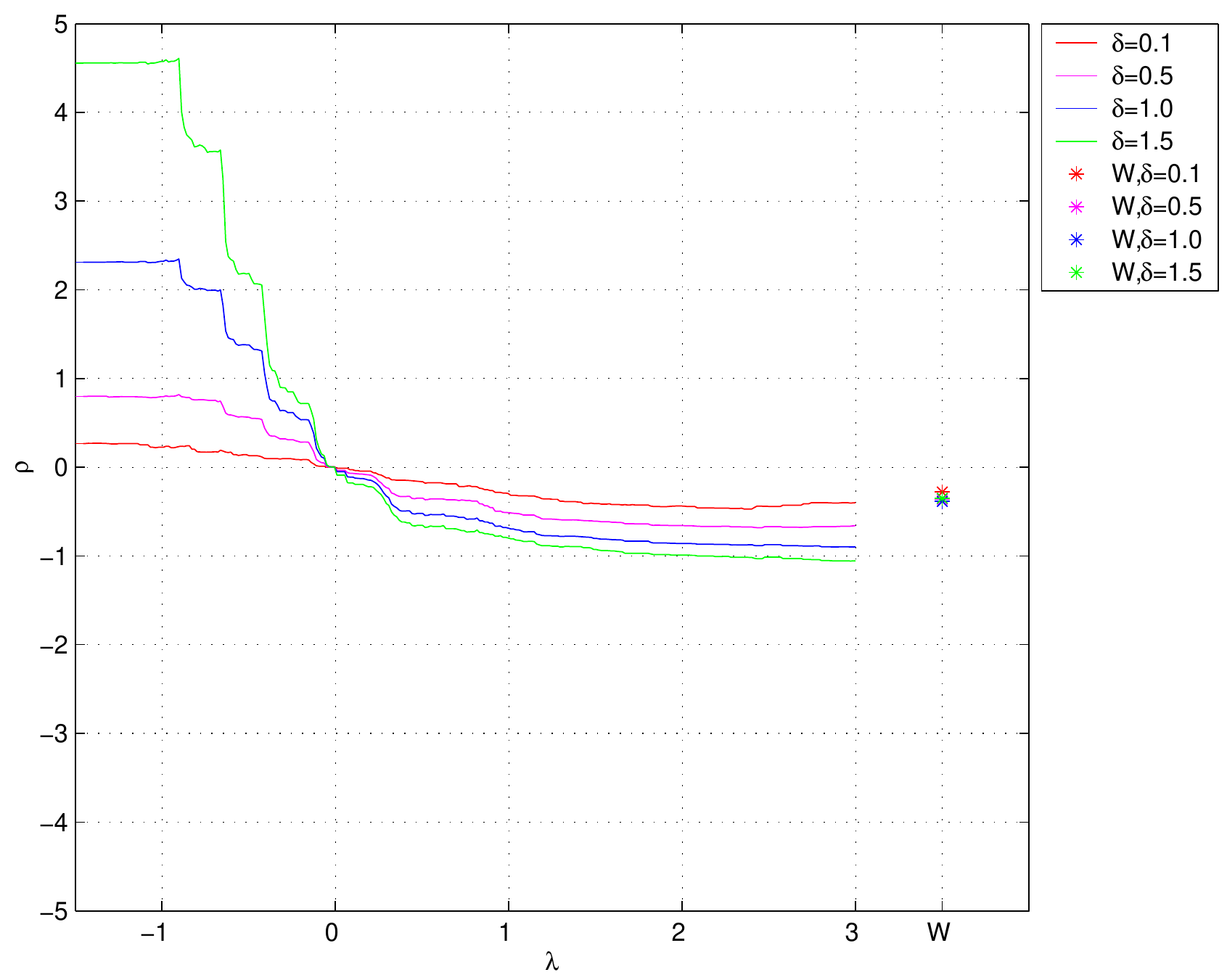}%
}
\\%
{\includegraphics[
height=2.4561in,
width=3.0701in
]%
{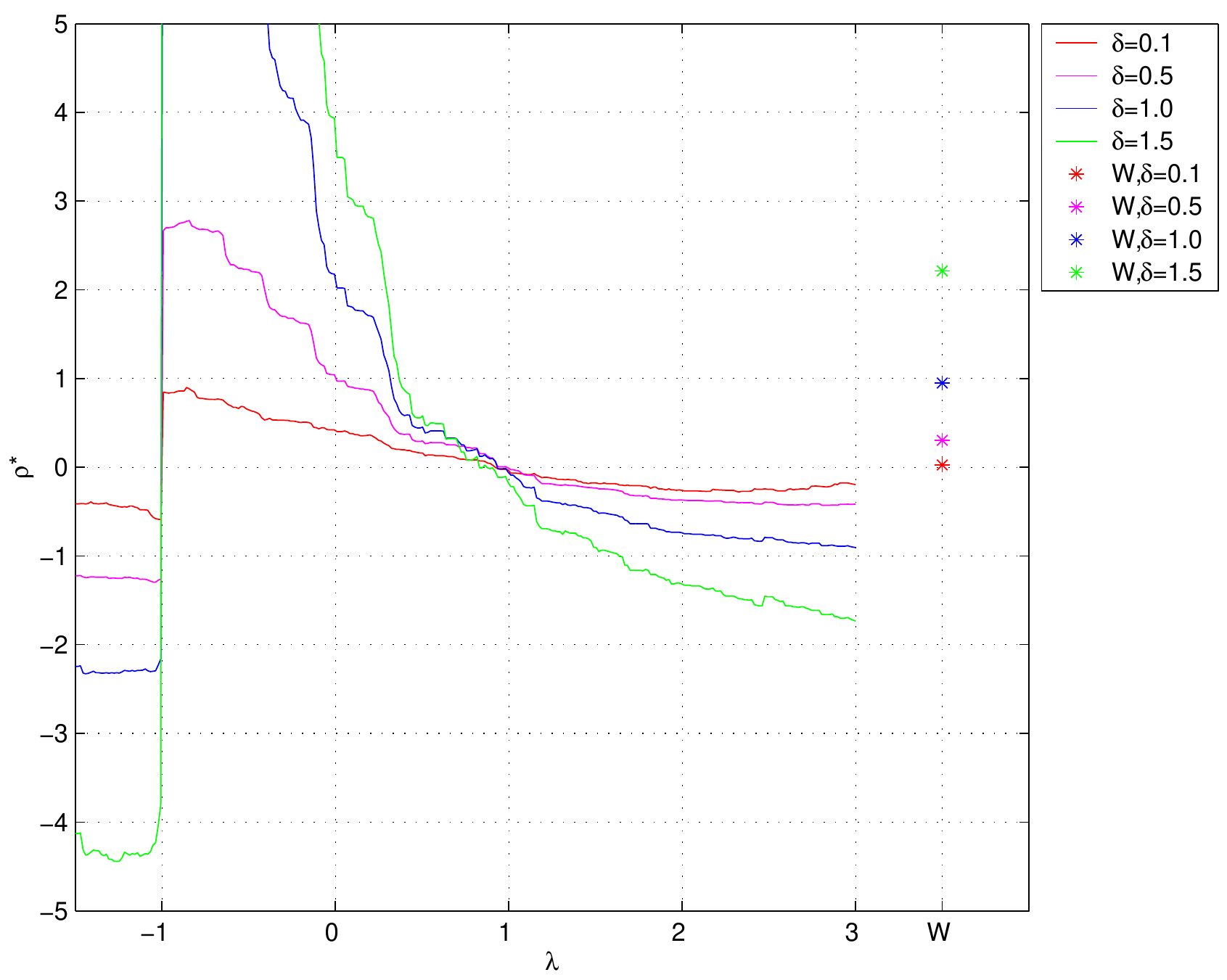}%
}
&
{\includegraphics[
height=2.4561in,
width=3.0701in
]%
{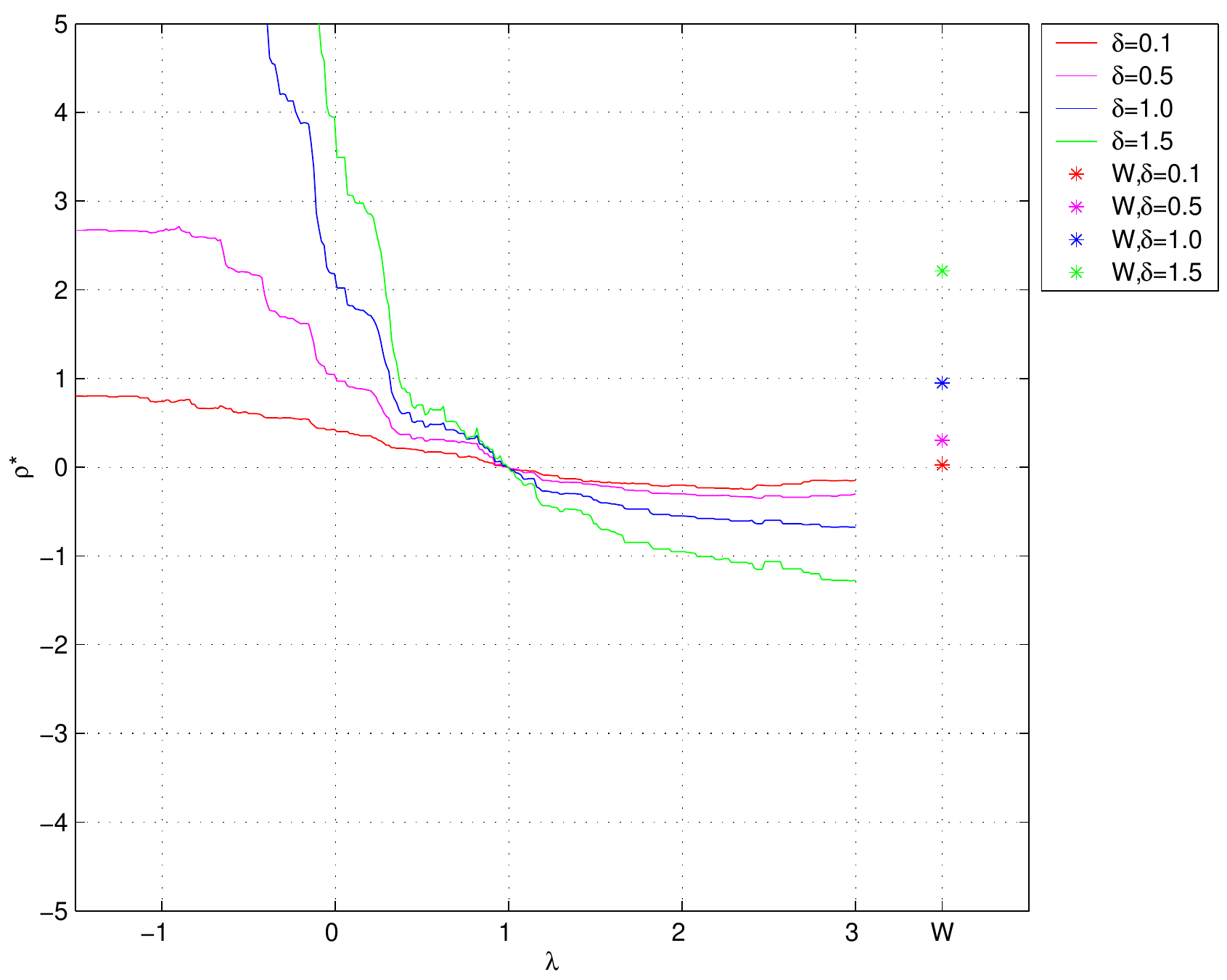}%
}
\end{tabular}
\caption{Power and relative local efficiencies for $T_{\lambda}$, $S_{\lambda}$ and $W$ in scenario B. \label{fig2}}%
\end{figure}%
%

\begin{figure}[htbp]  \tabcolsep2.8pt  \centering
\begin{tabular}
[c]{cc}%
${T_{\lambda}}$ & ${S_{\lambda}}$\\%
{\includegraphics[
height=2.4561in,
width=3.1202in
]%
{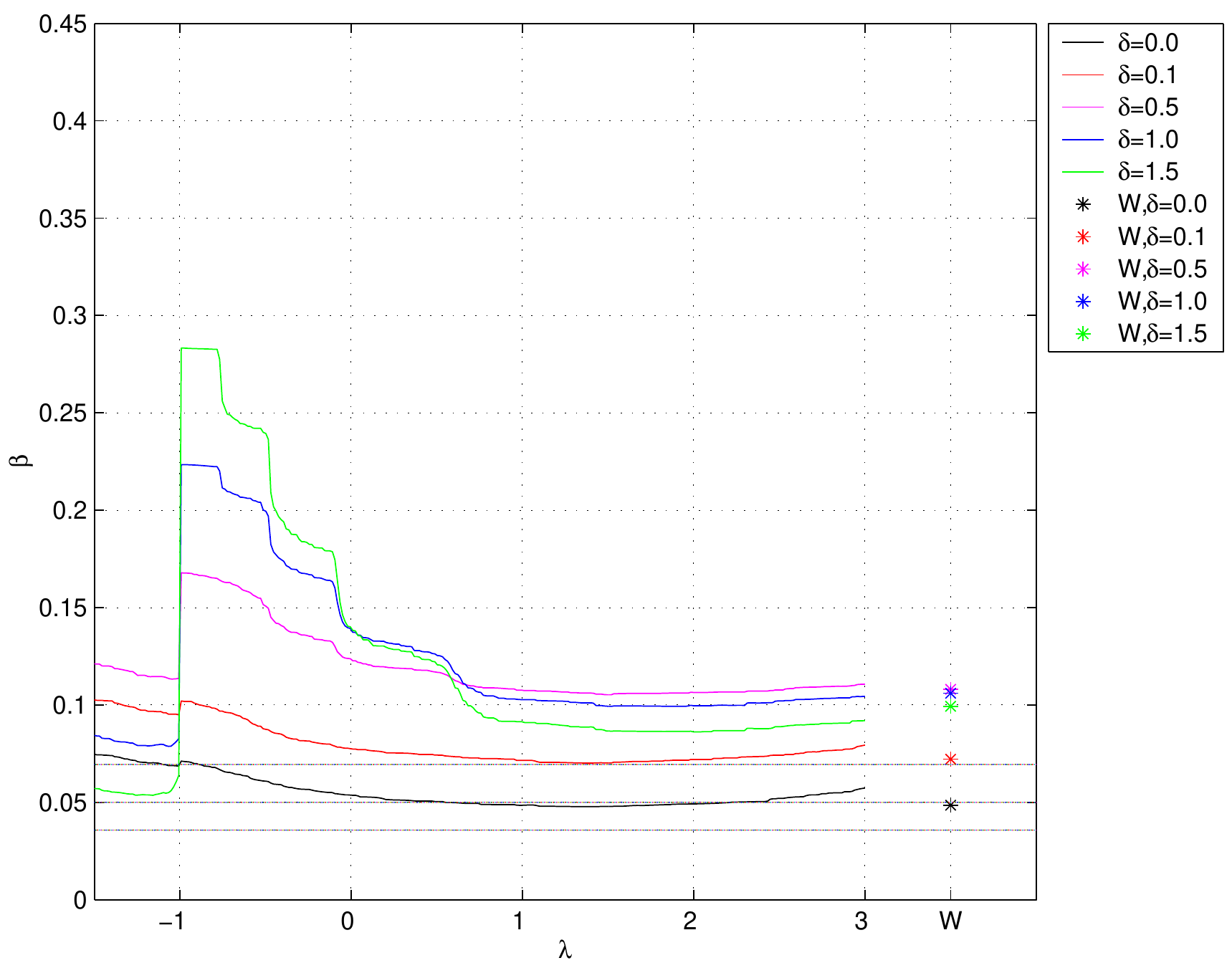}%
}
&
{\includegraphics[
height=2.4561in,
width=3.1202in
]%
{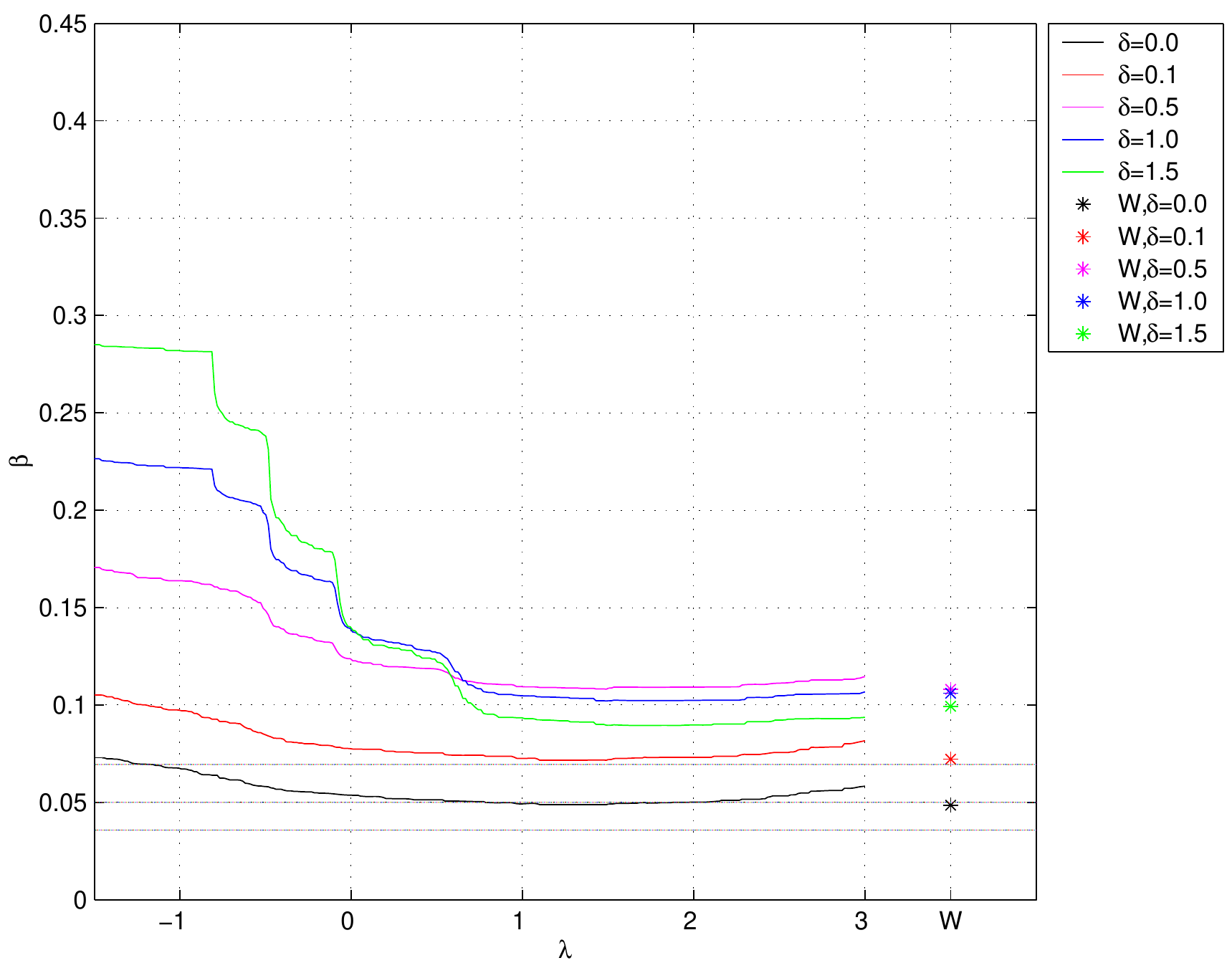}%
}
\\%
{\includegraphics[
height=2.4561in,
width=3.0701in
]%
{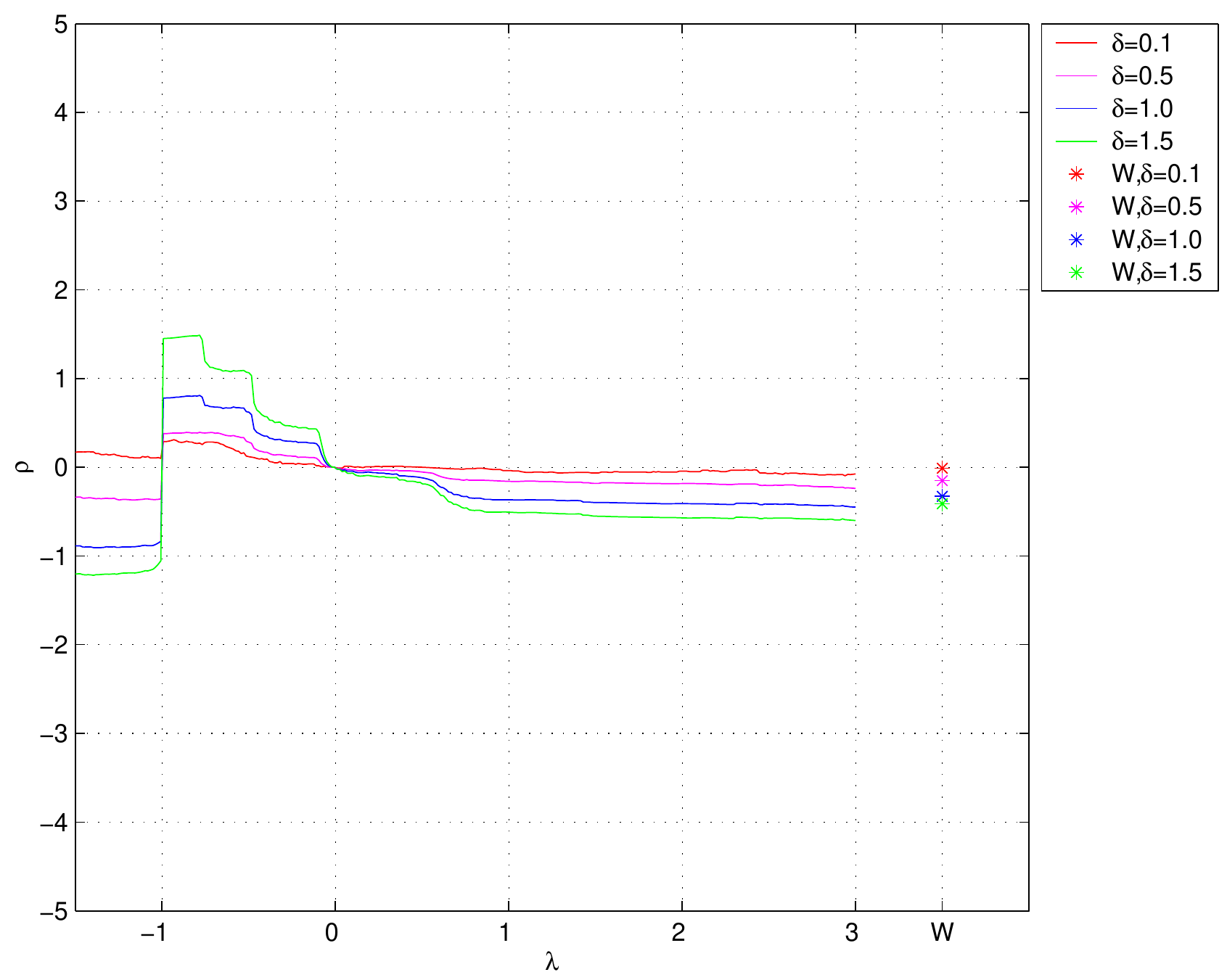}%
}
&
{\includegraphics[
height=2.4561in,
width=3.0701in
]%
{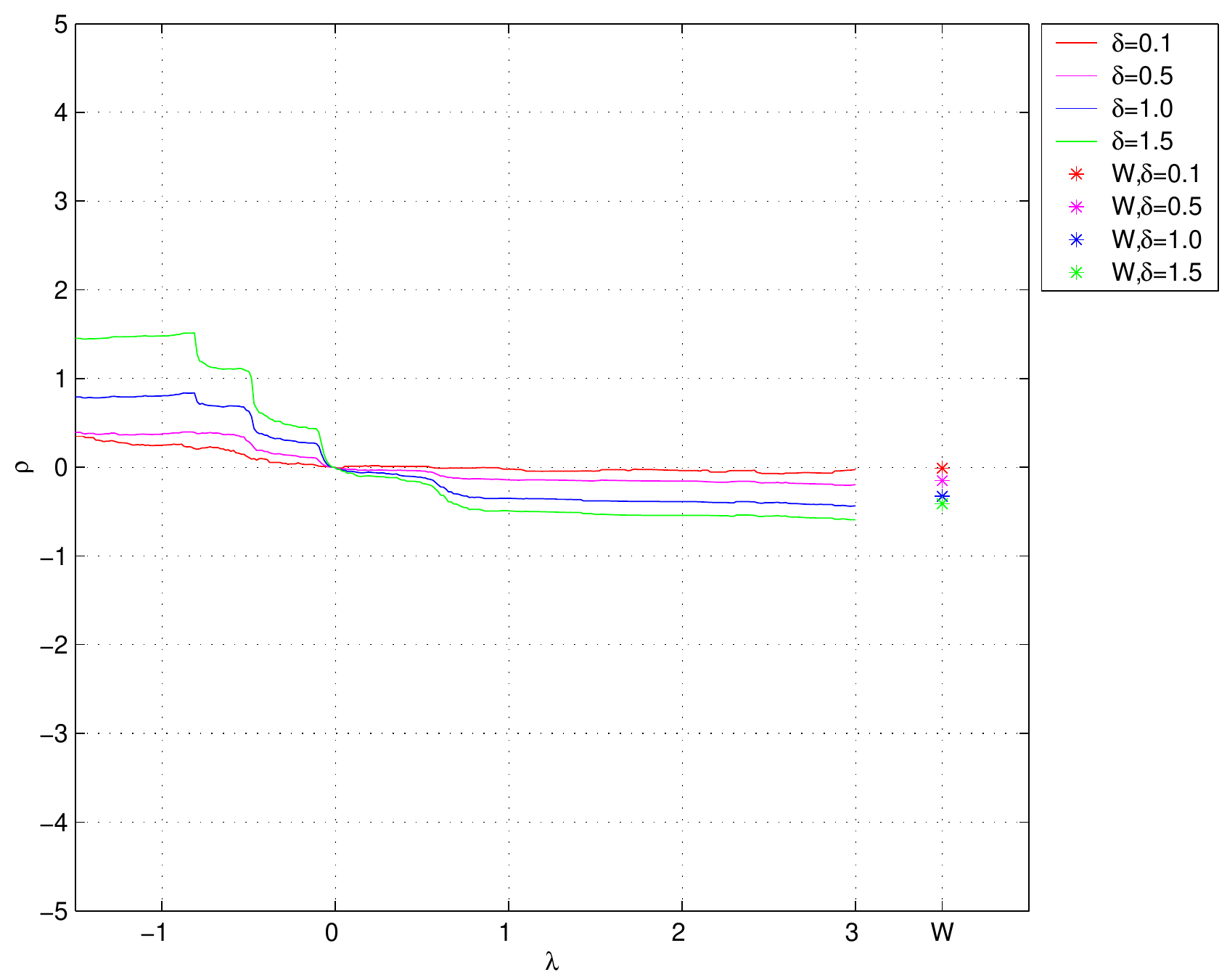}%
}
\\%
{\includegraphics[
height=2.4561in,
width=3.0701in
]%
{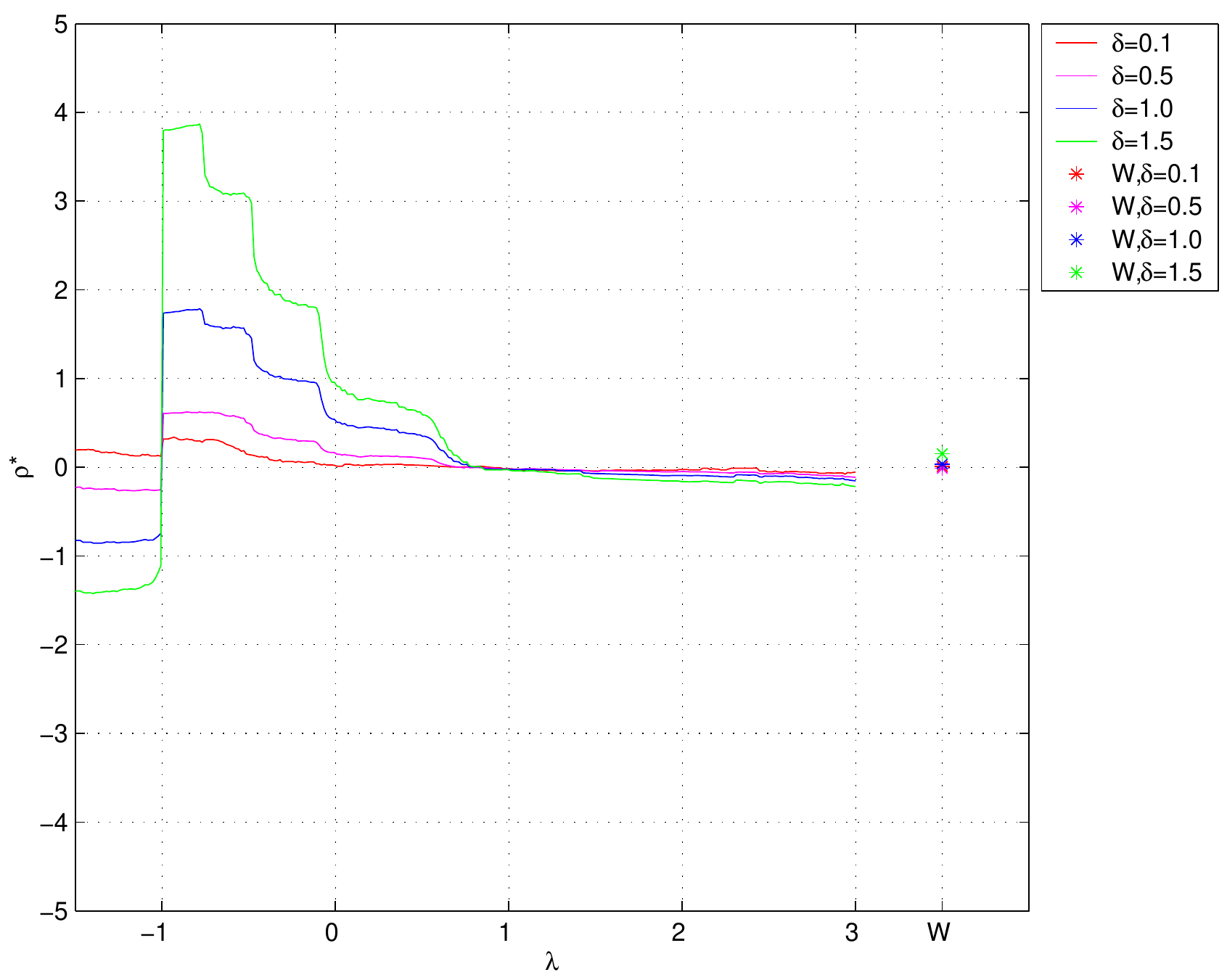}%
}
&
{\includegraphics[
height=2.4561in,
width=3.0701in
]%
{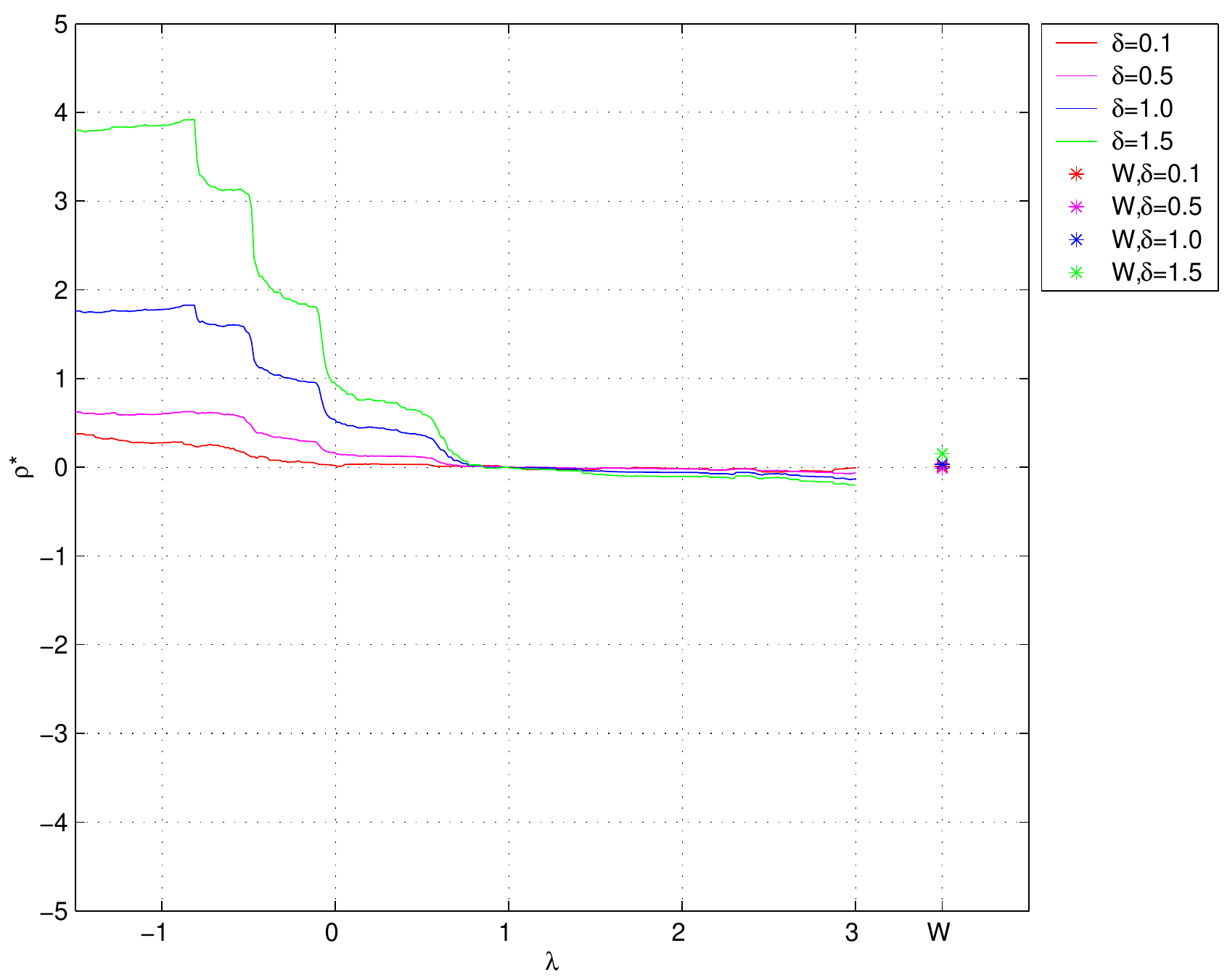}%
}
\end{tabular}
\caption{Power and relative local efficiencies for $T_{\lambda}$, $S_{\lambda}$ and $W$ in scenario C. \label{fig3}}%
\end{figure}%
%

\begin{figure}[htbp]  \tabcolsep2.8pt  \centering
\begin{tabular}
[c]{cc}%
${T_{\lambda}}$ & ${S_{\lambda}}$\\%
{\includegraphics[
height=2.4561in,
width=3.1202in
]%
{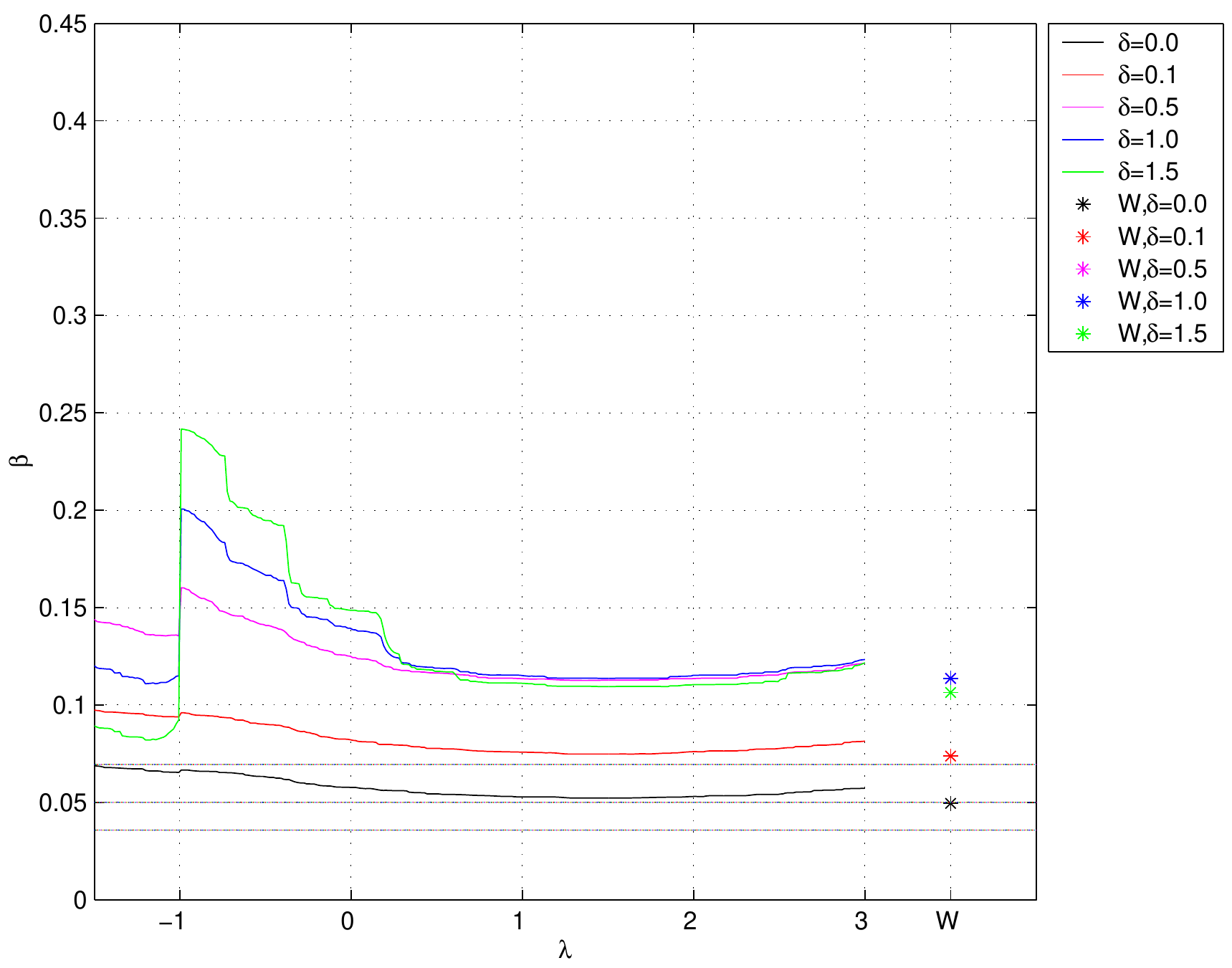}%
}
&
{\includegraphics[
height=2.4561in,
width=3.1202in
]%
{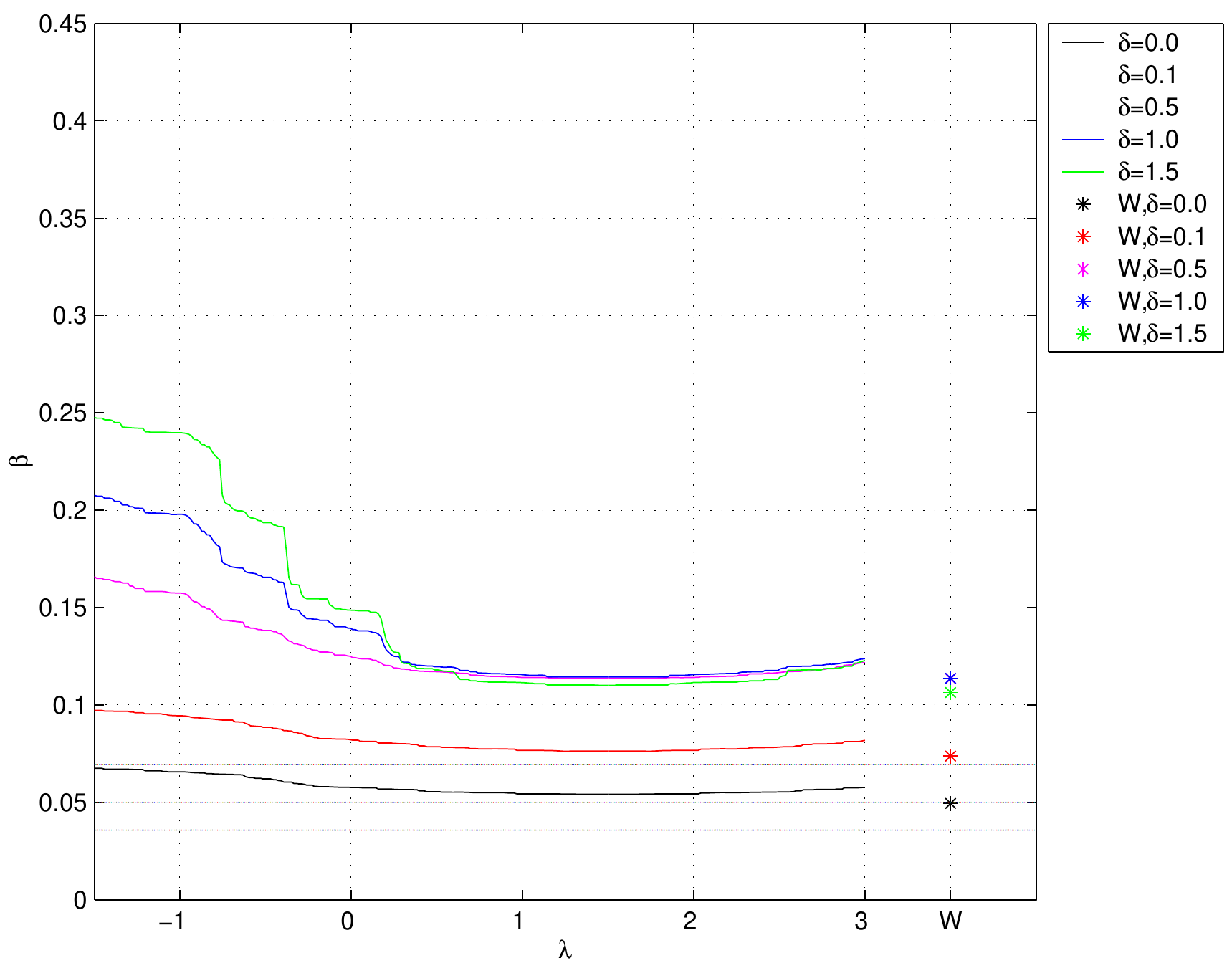}%
}
\\%
{\includegraphics[
height=2.4561in,
width=3.0701in
]%
{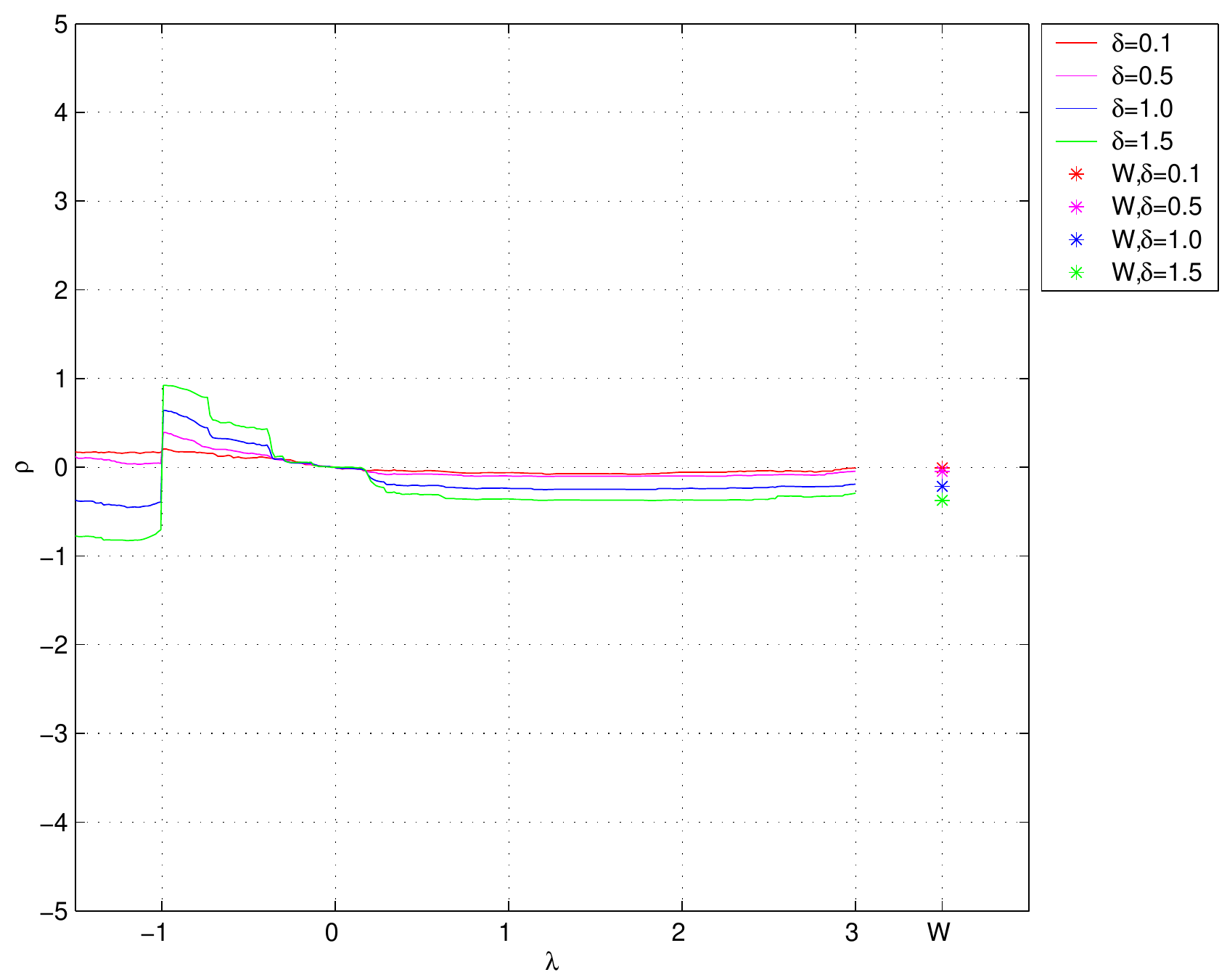}%
}
&
{\includegraphics[
height=2.4561in,
width=3.0701in
]%
{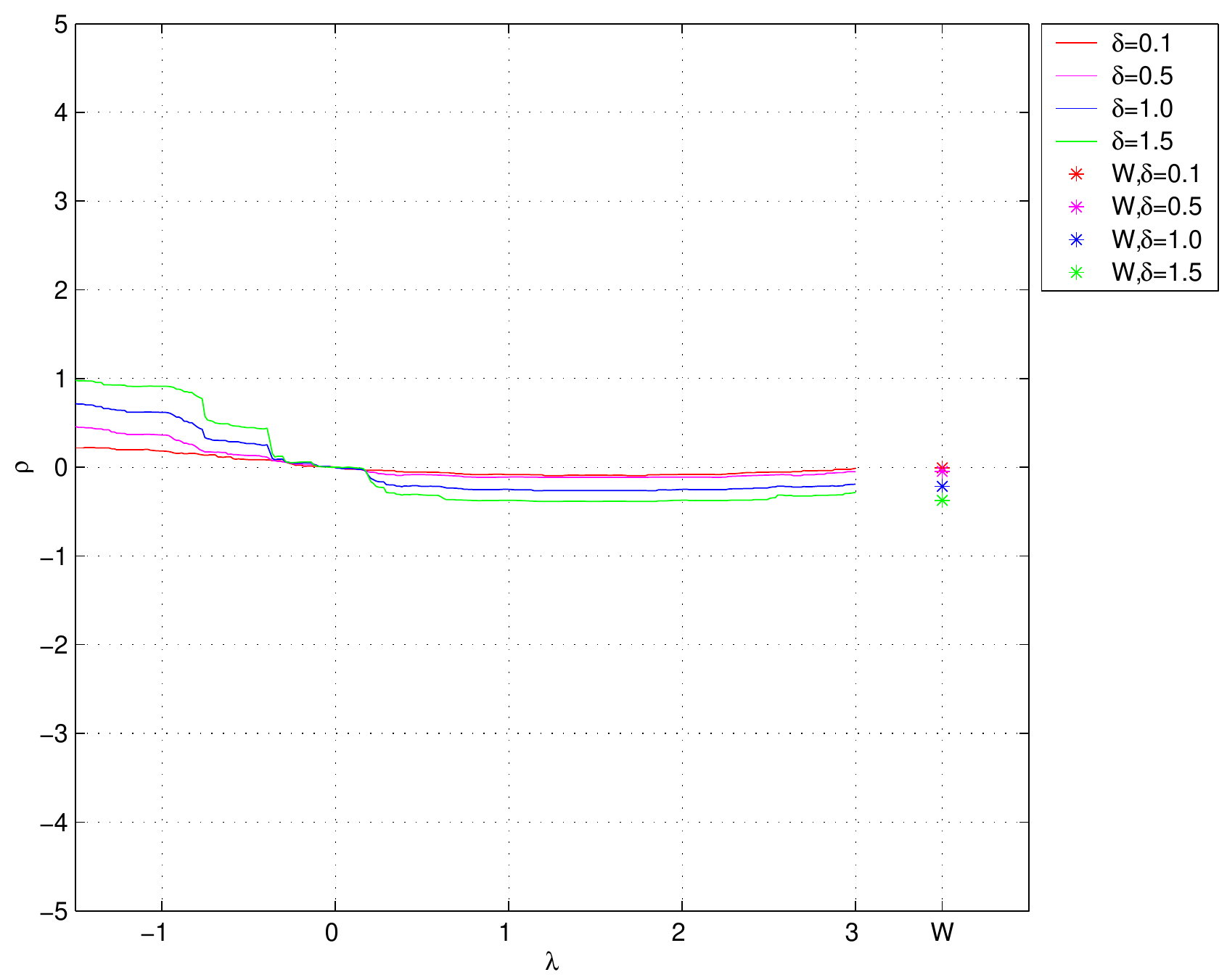}%
}
\\%
{\includegraphics[
height=2.4561in,
width=3.0701in
]%
{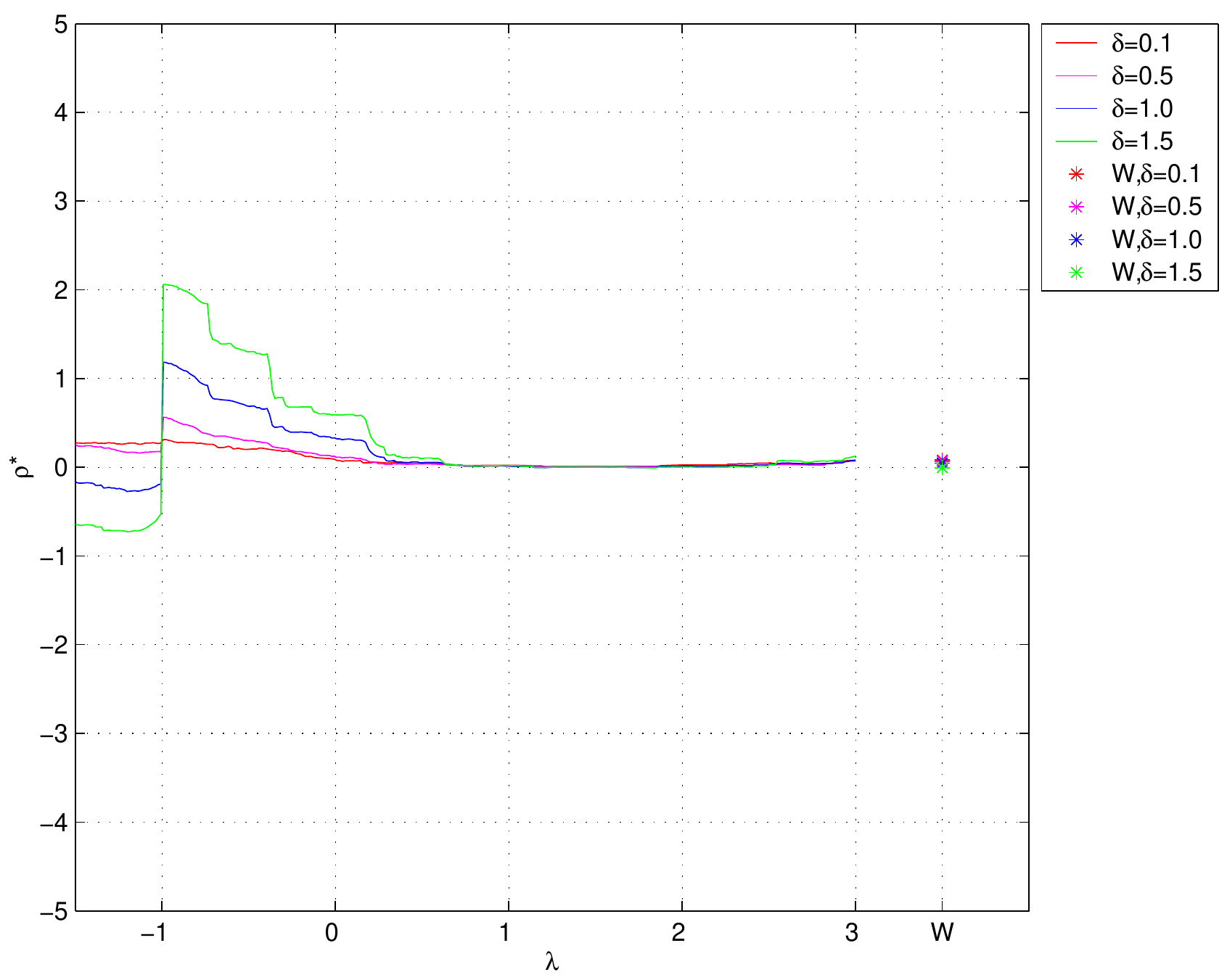}%
}
&
{\includegraphics[
height=2.4561in,
width=3.0701in
]%
{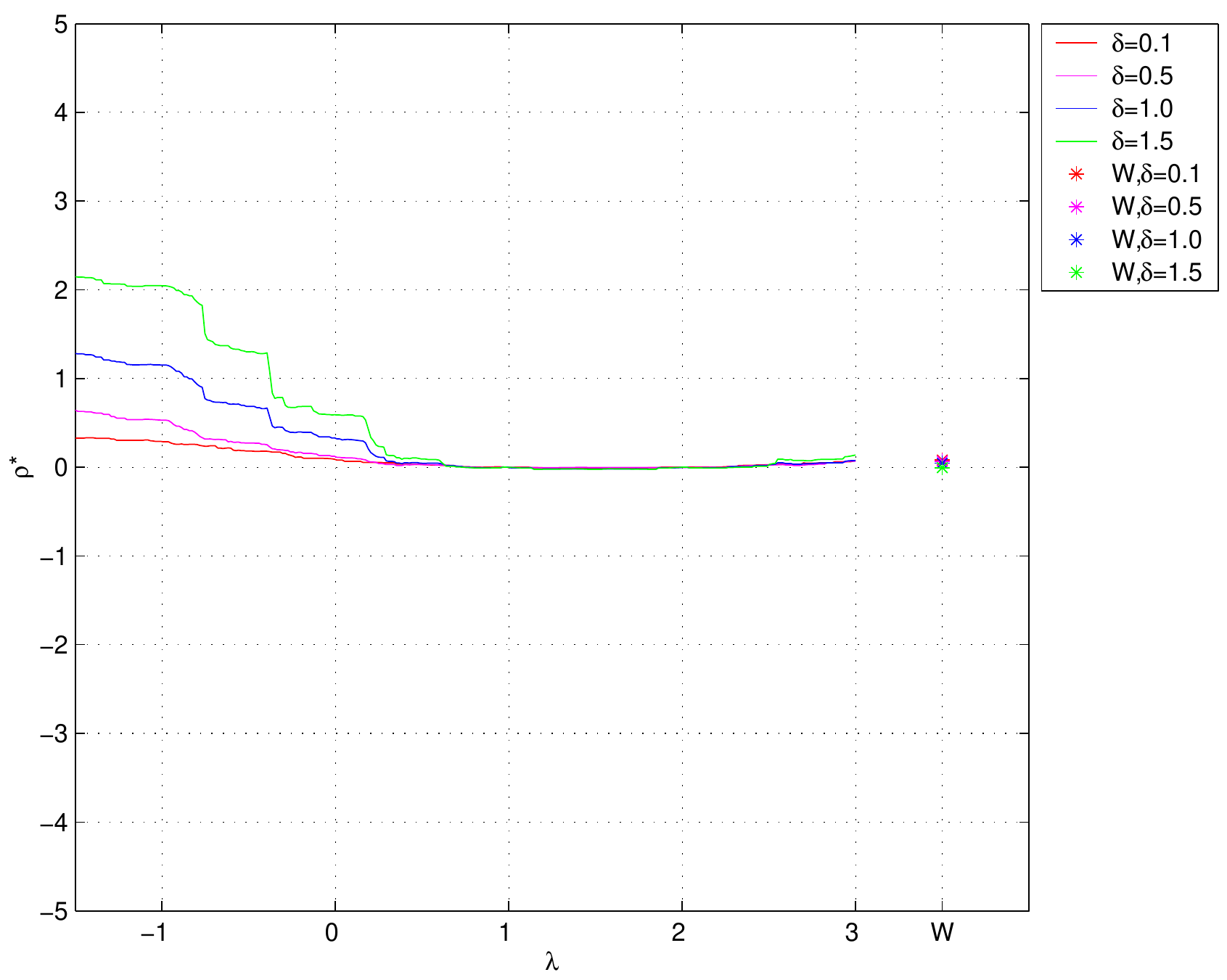}%
}
\end{tabular}
\caption{Power and relative local efficiencies for $T_{\lambda}$, $S_{\lambda}$ and $W$ in scenario D. \label{fig4}}%
\end{figure}%
%

\begin{figure}[htbp]  \tabcolsep2.8pt  \centering
\begin{tabular}
[c]{cc}%
${T_{\lambda}}$ & ${S_{\lambda}}$\\%
{\includegraphics[
height=2.463in,
width=3.3667in
]%
{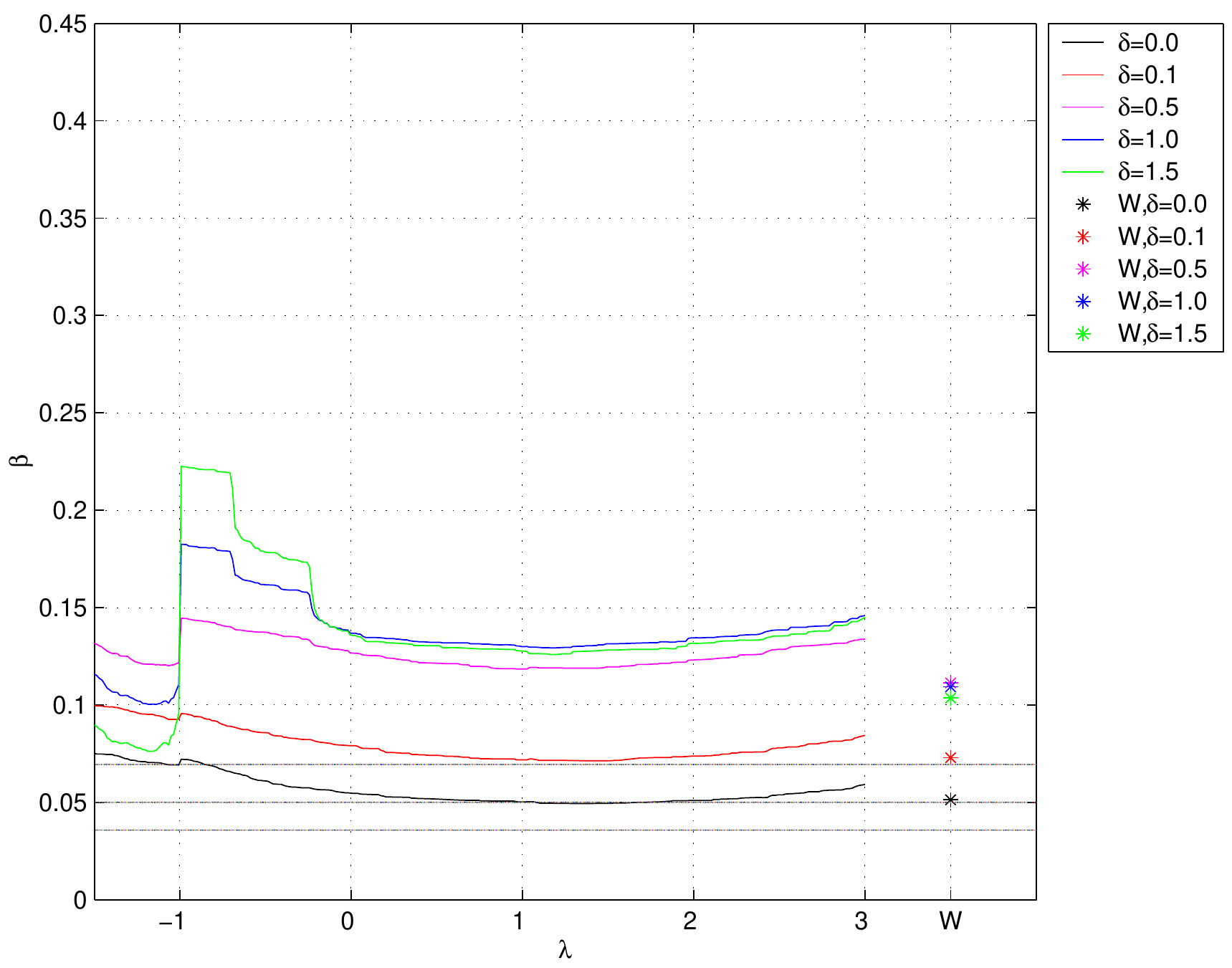}%
}
&
{\includegraphics[
height=2.463in,
width=3.3667in
]%
{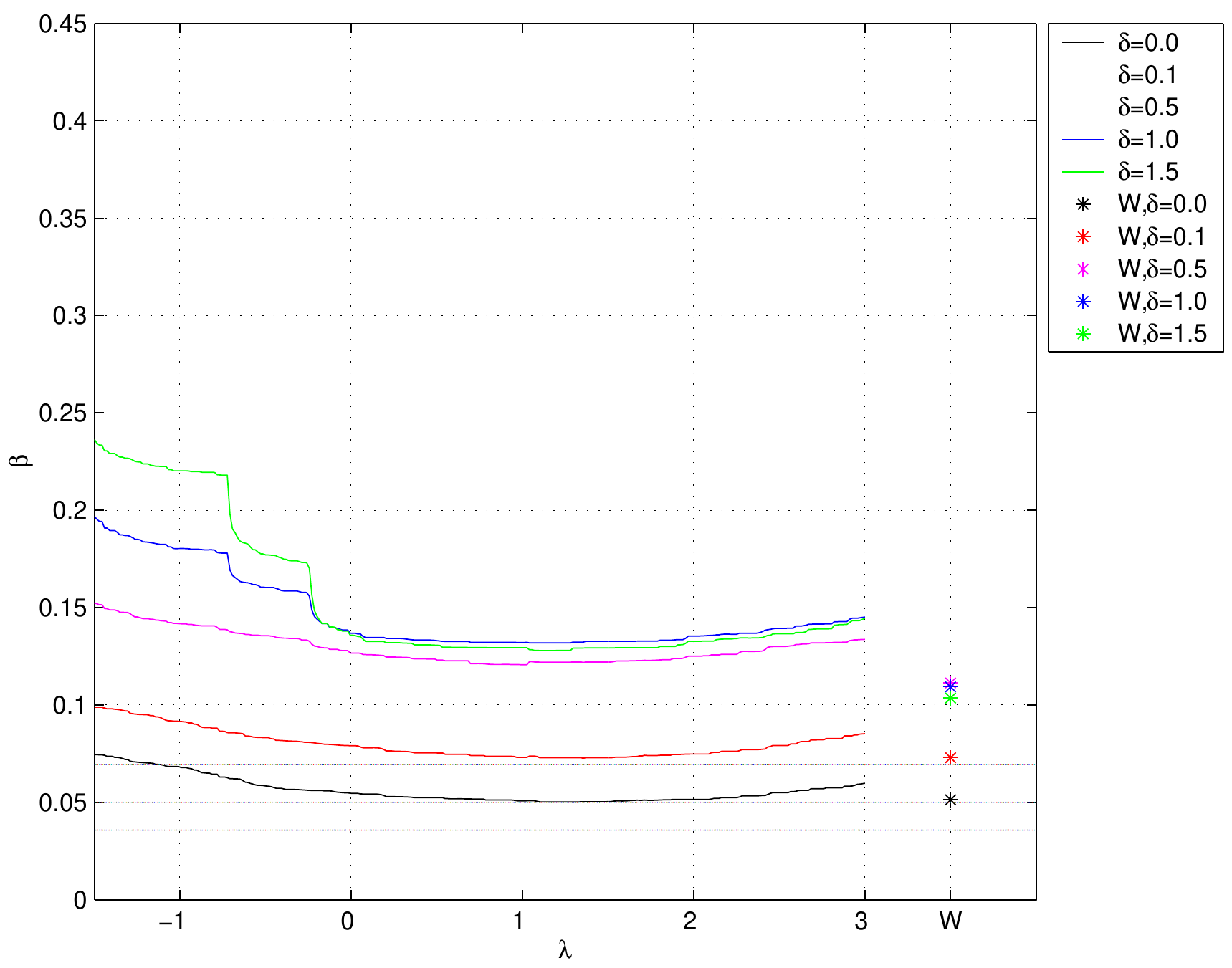}%
}
\\%
{\includegraphics[
height=2.463in,
width=3.3667in
]%
{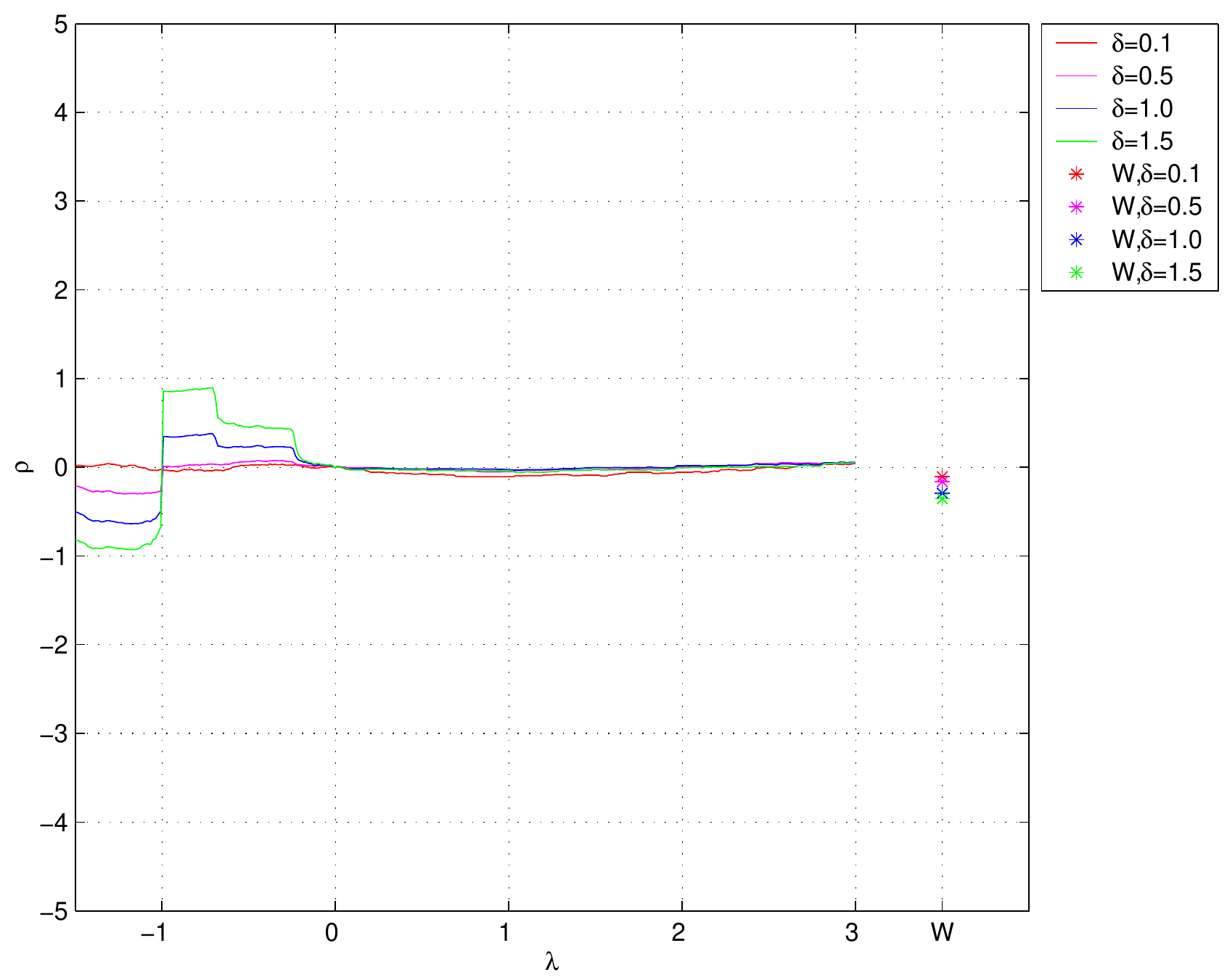}%
}
&
{\includegraphics[
height=2.463in,
width=3.3667in
]%
{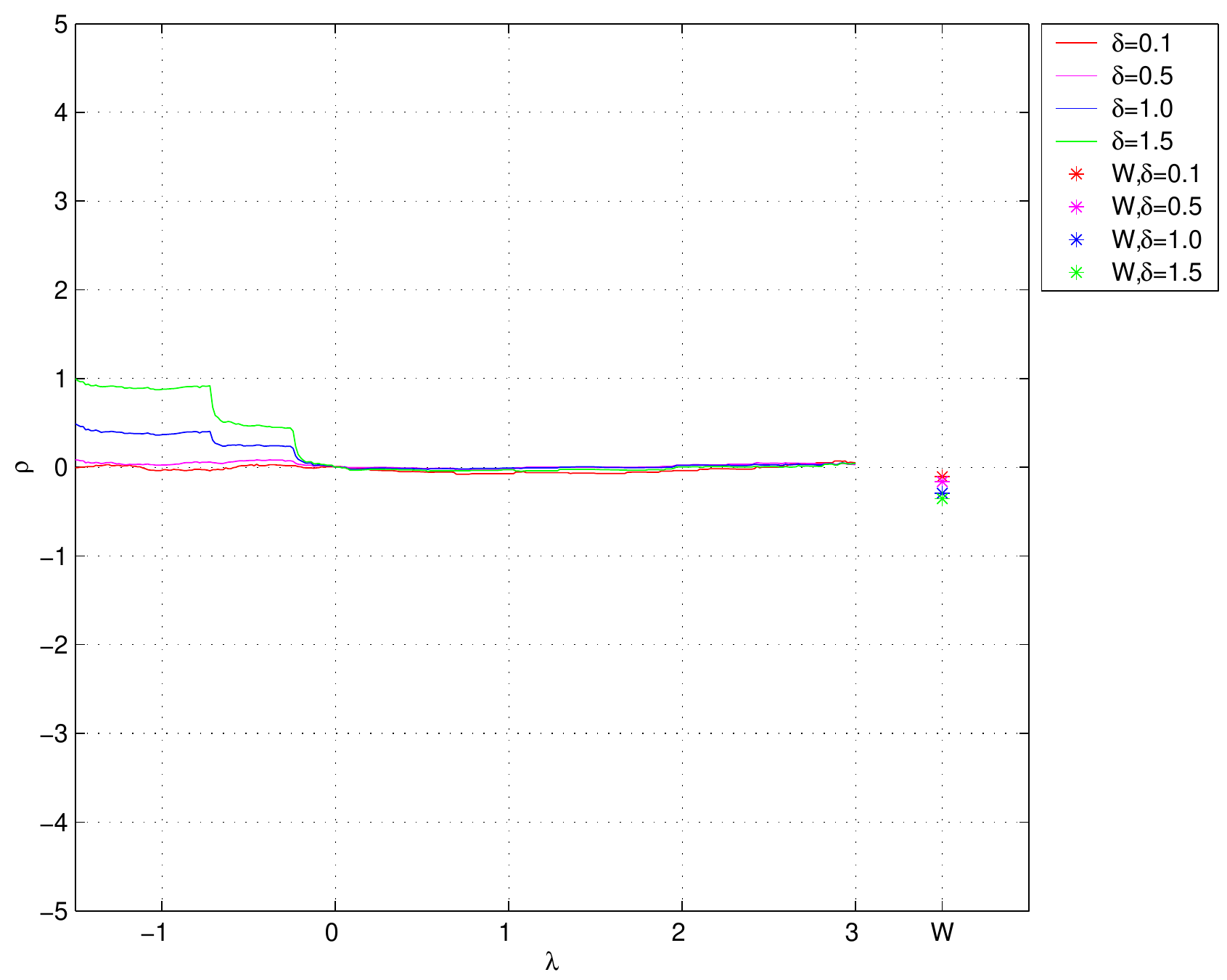}%
}
\\%
{\includegraphics[
height=2.463in,
width=3.3667in
]%
{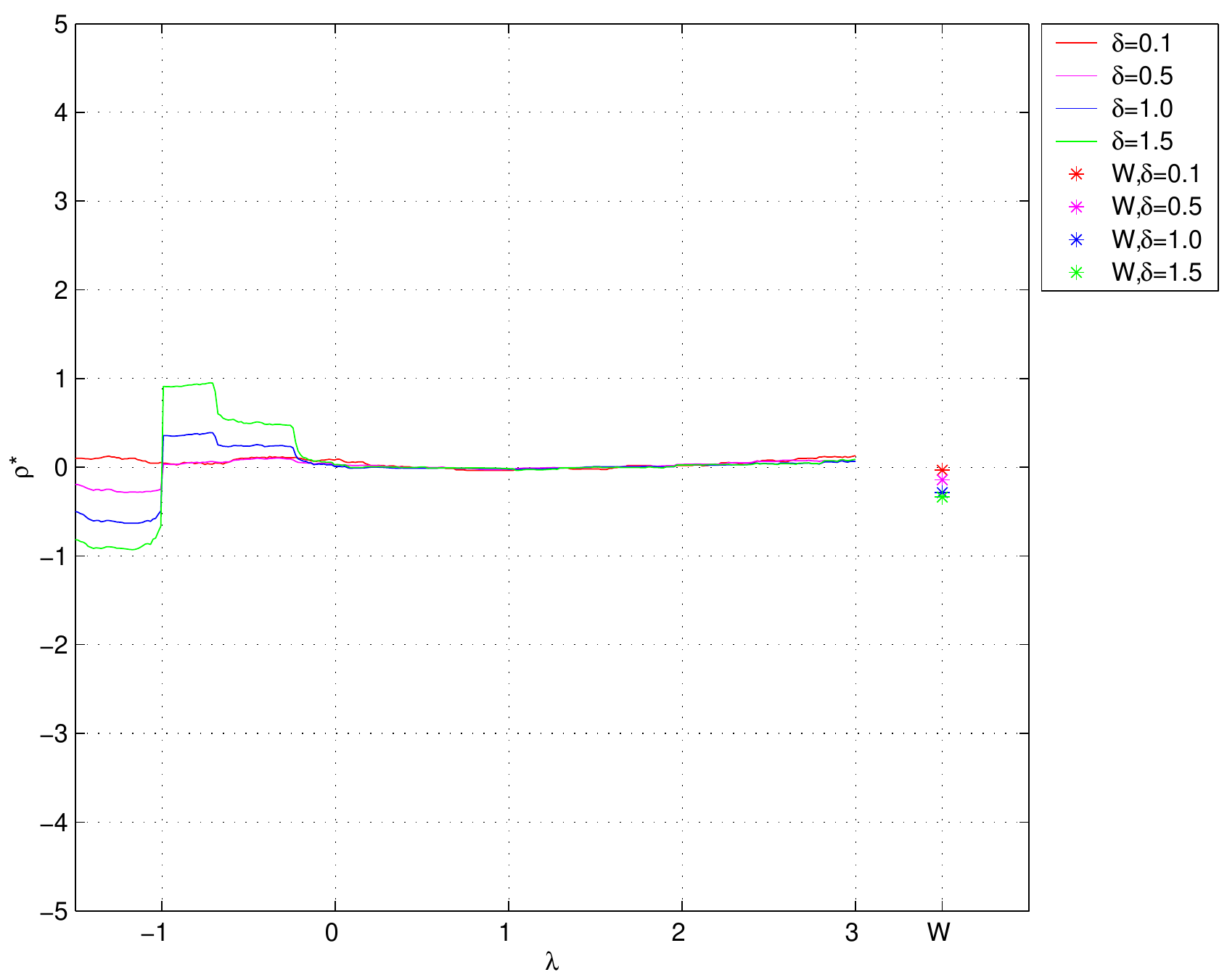}%
}
&
{\includegraphics[
height=2.463in,
width=3.3667in
]%
{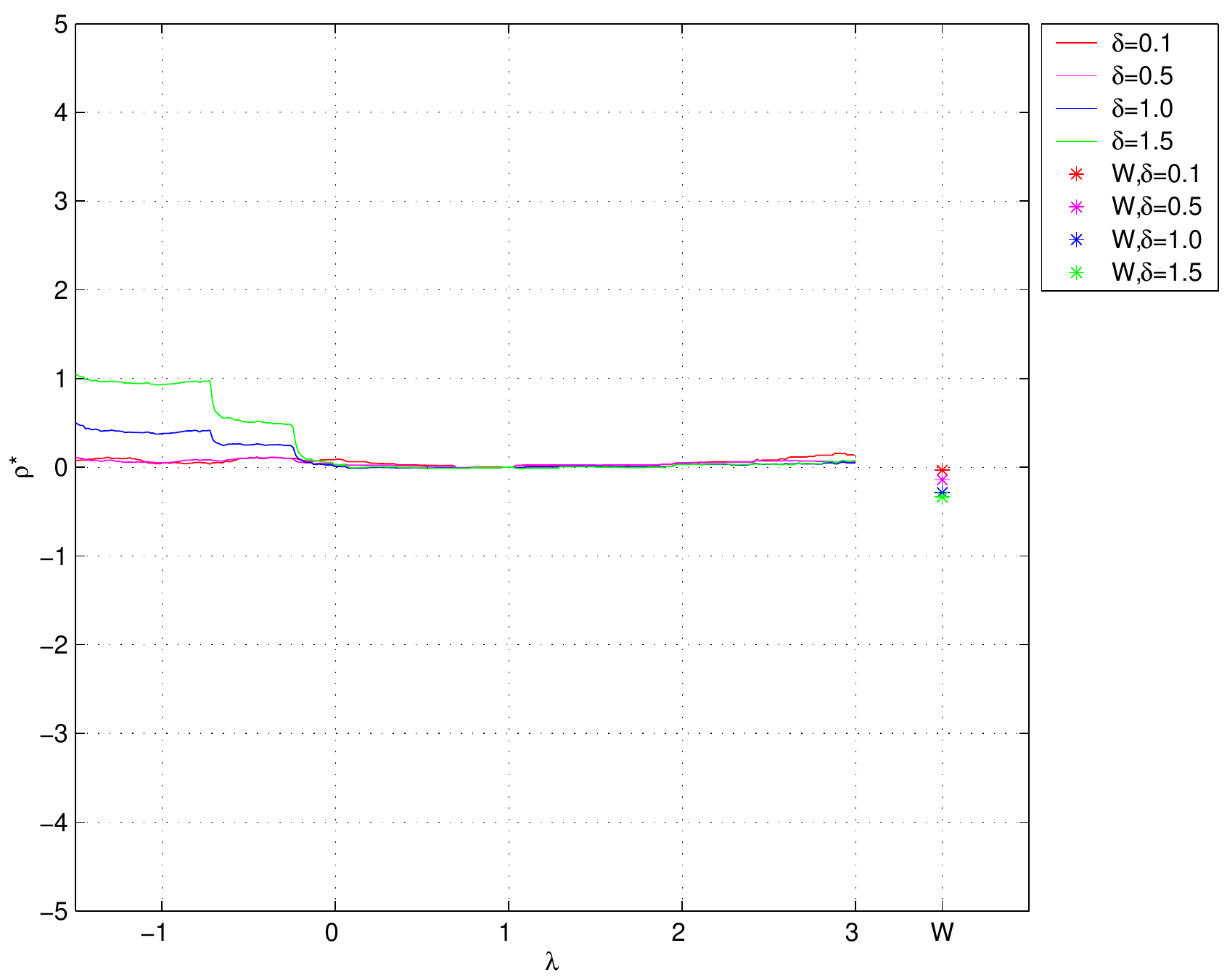}%
}
\end{tabular}
\caption{Power and relative local efficiencies for $T_{\lambda}$, $S_{\lambda}$ and $W$ in scenario E. \label{fig5}}%
\end{figure}%
%

\begin{figure}[htbp]  \tabcolsep2.8pt  \centering
\begin{tabular}
[c]{cc}%
${T_{\lambda}}$ & ${S_{\lambda}}$\\%
{\includegraphics[
height=2.463in,
width=3.3667in
]%
{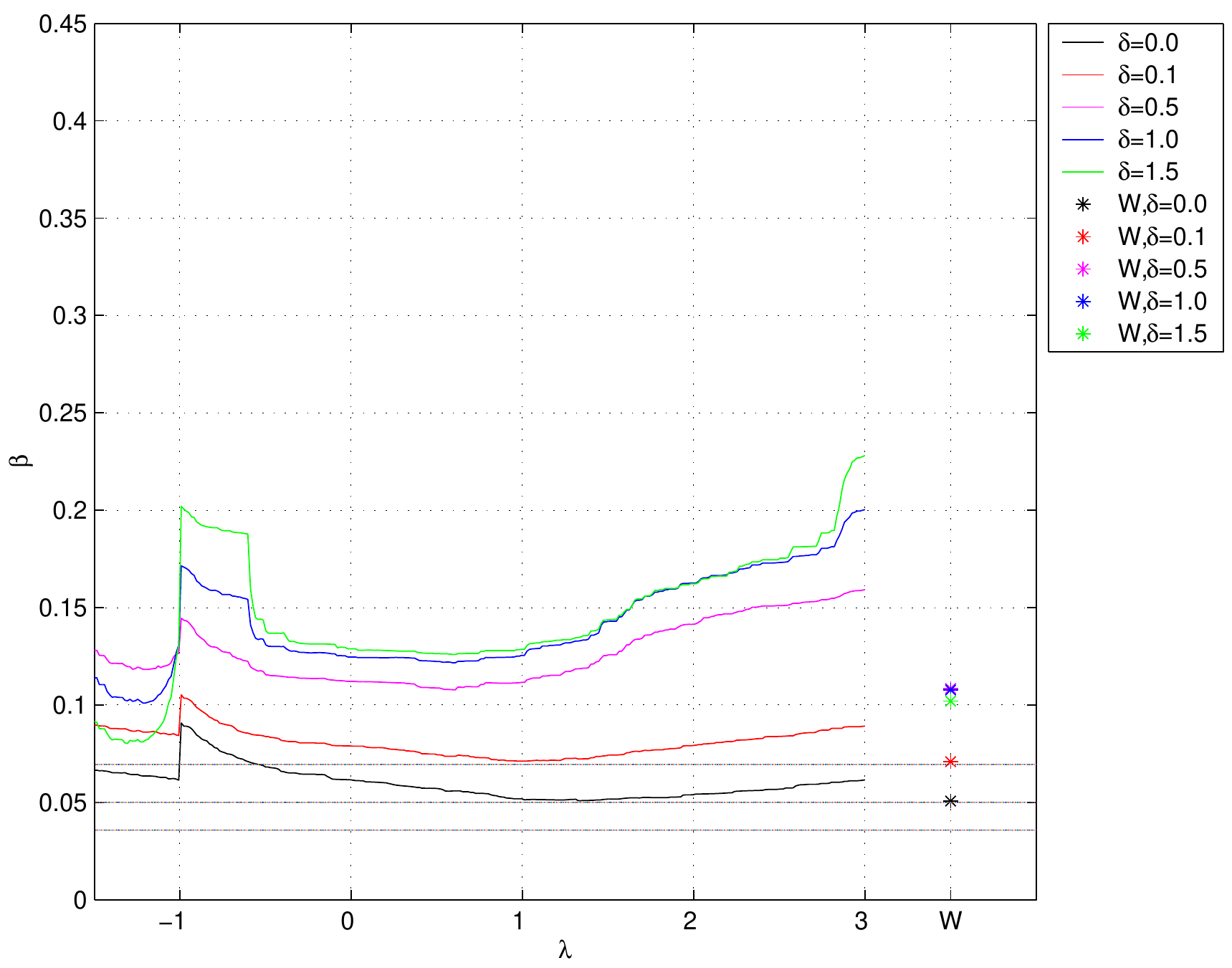}%
}
&
{\includegraphics[
height=2.463in,
width=3.3667in
]%
{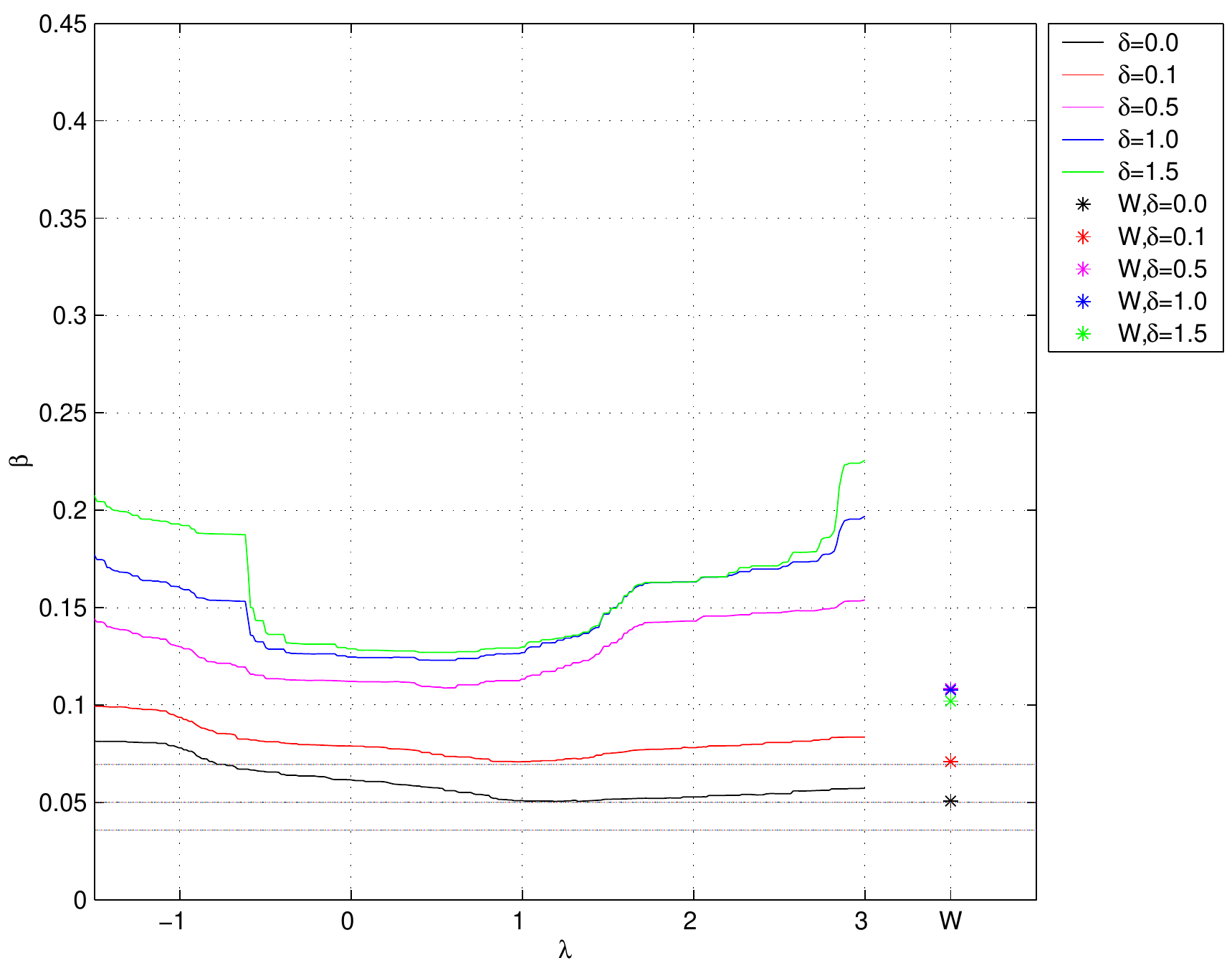}%
}
\\%
{\includegraphics[
height=2.463in,
width=3.3667in
]%
{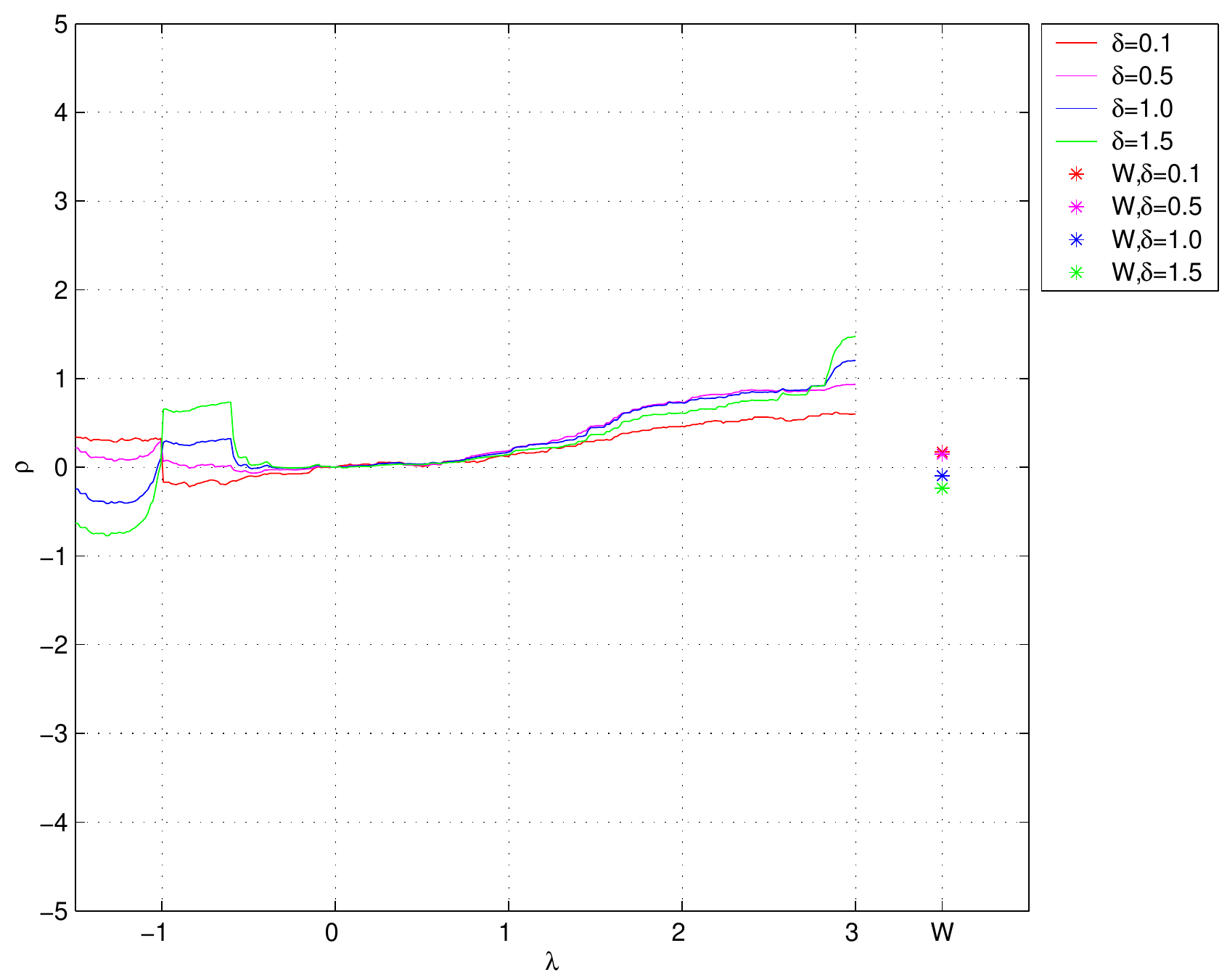}%
}
&
{\includegraphics[
height=2.463in,
width=3.3667in
]%
{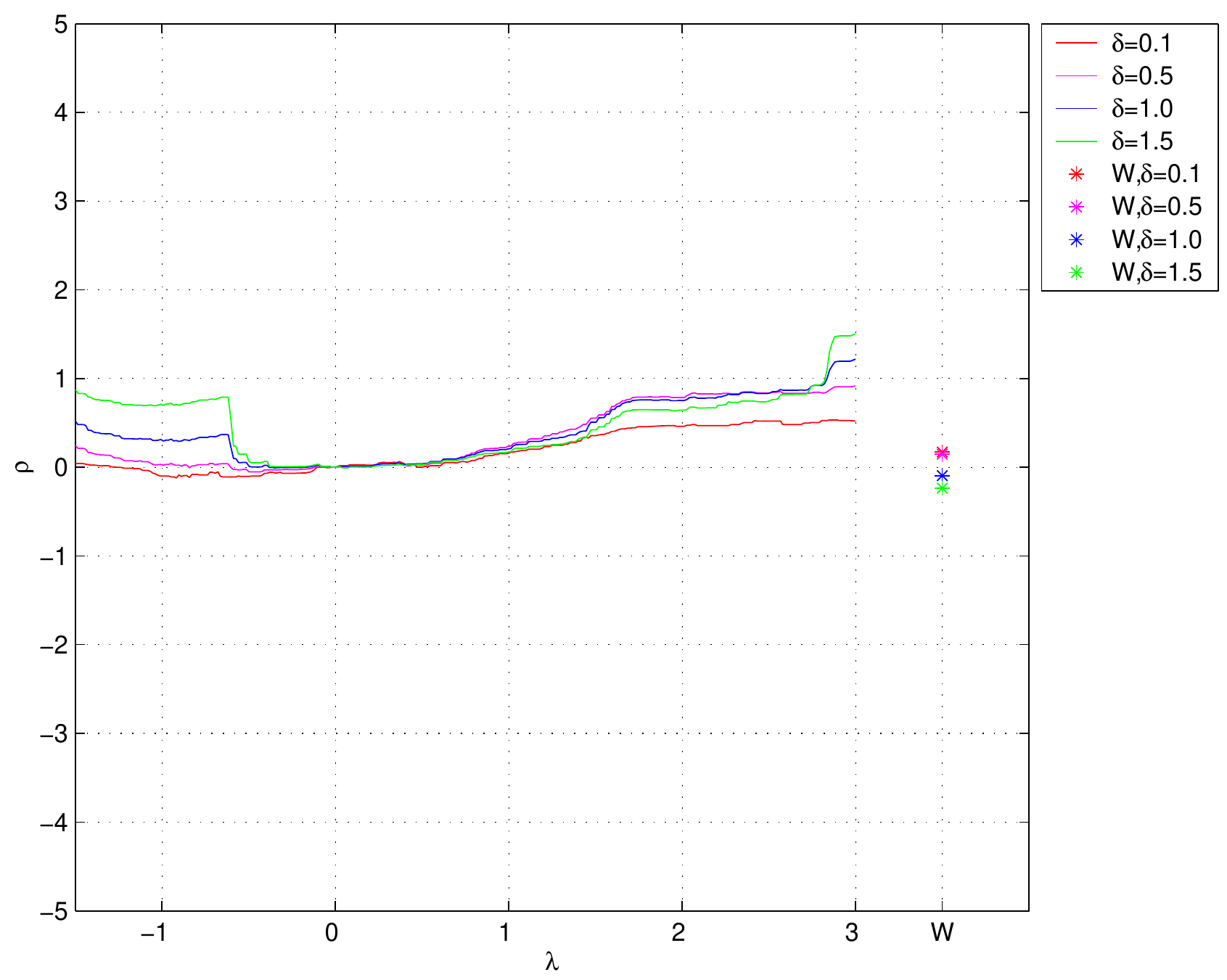}%
}
\\%
{\includegraphics[
height=2.463in,
width=3.3667in
]%
{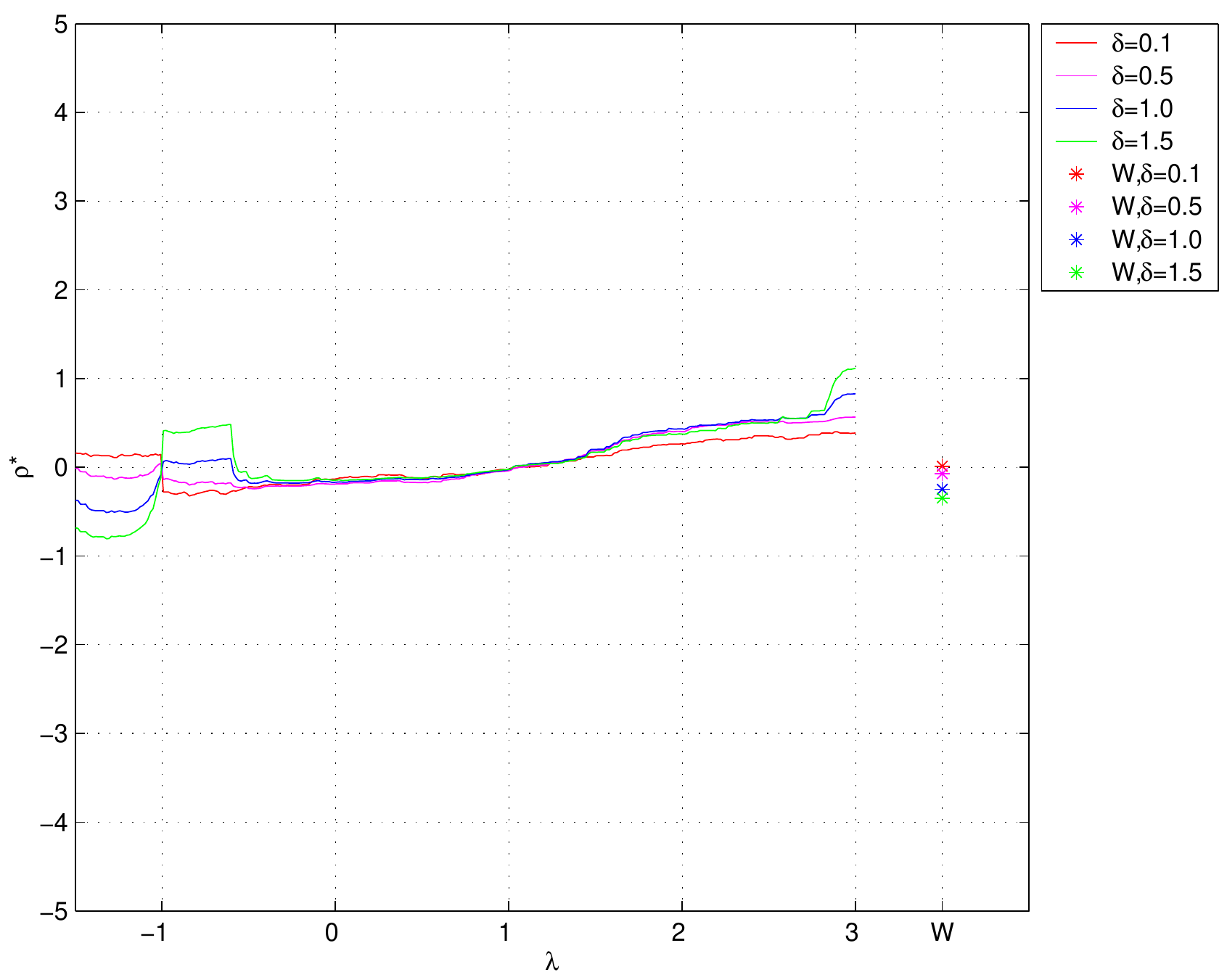}%
}
&
{\includegraphics[
height=2.463in,
width=3.3667in
]%
{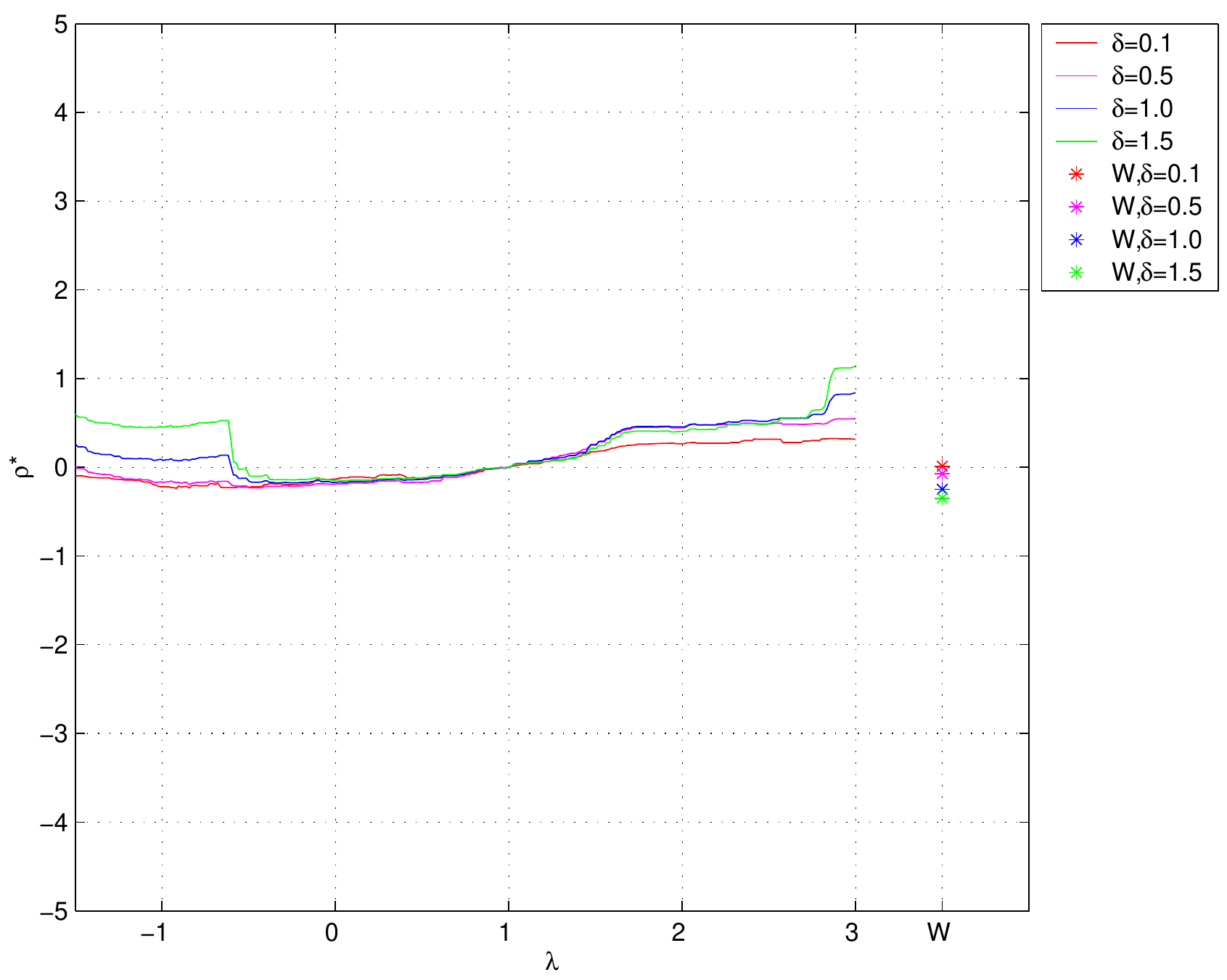}%
}
\end{tabular}
\caption{Power and relative local efficiencies for $T_{\lambda}$, $S_{\lambda}$ and $W$ in scenario F. \label{fig6}}%
\end{figure}%
%

\begin{figure}[htbp]  \tabcolsep2.8pt  \centering
\begin{tabular}
[c]{cc}%
${T_{\lambda}}$ & ${S_{\lambda}}$\\%
{\includegraphics[
height=2.4561in,
width=3.1202in
]%
{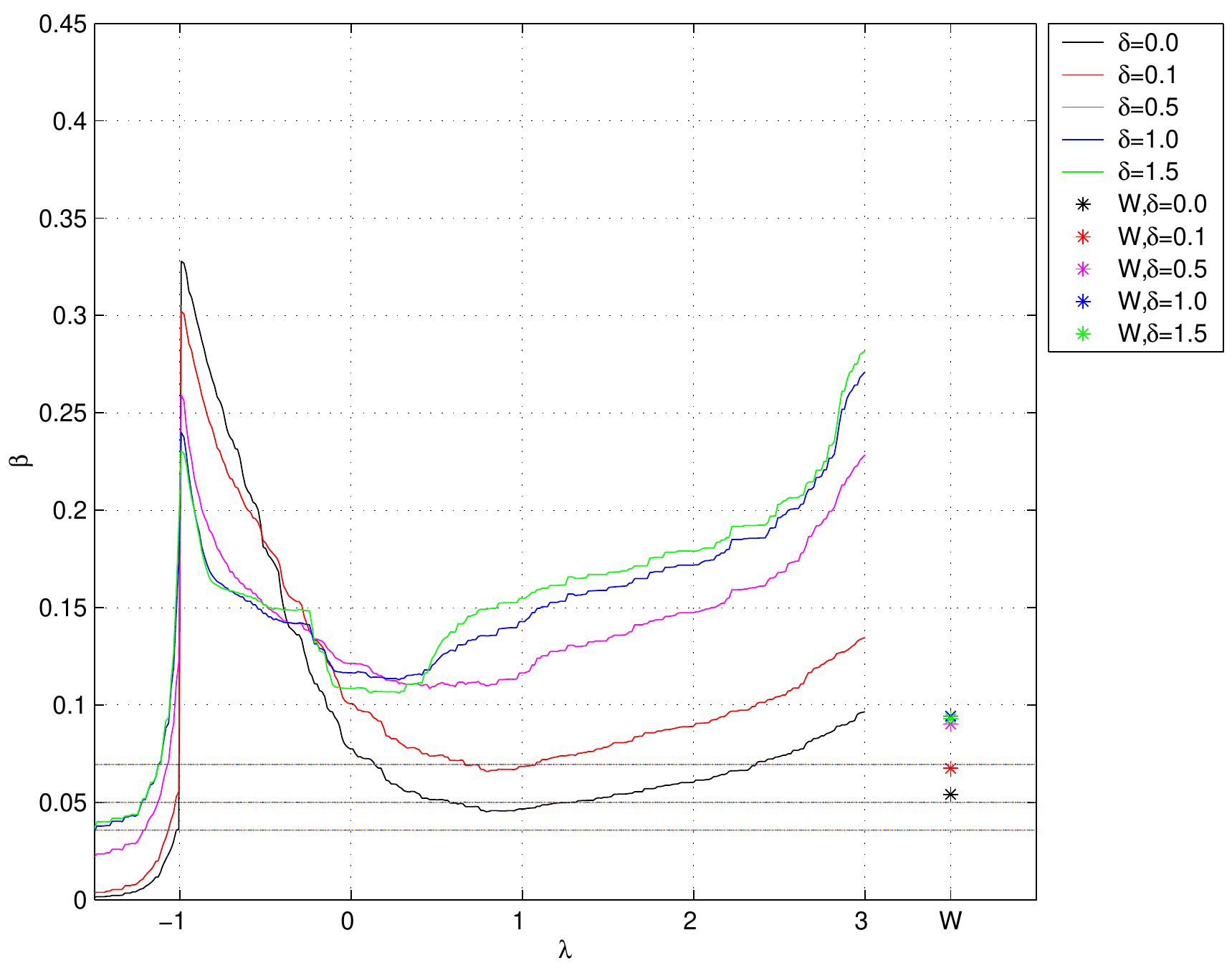}%
}
&
{\includegraphics[
height=2.4552in,
width=3.1211in
]%
{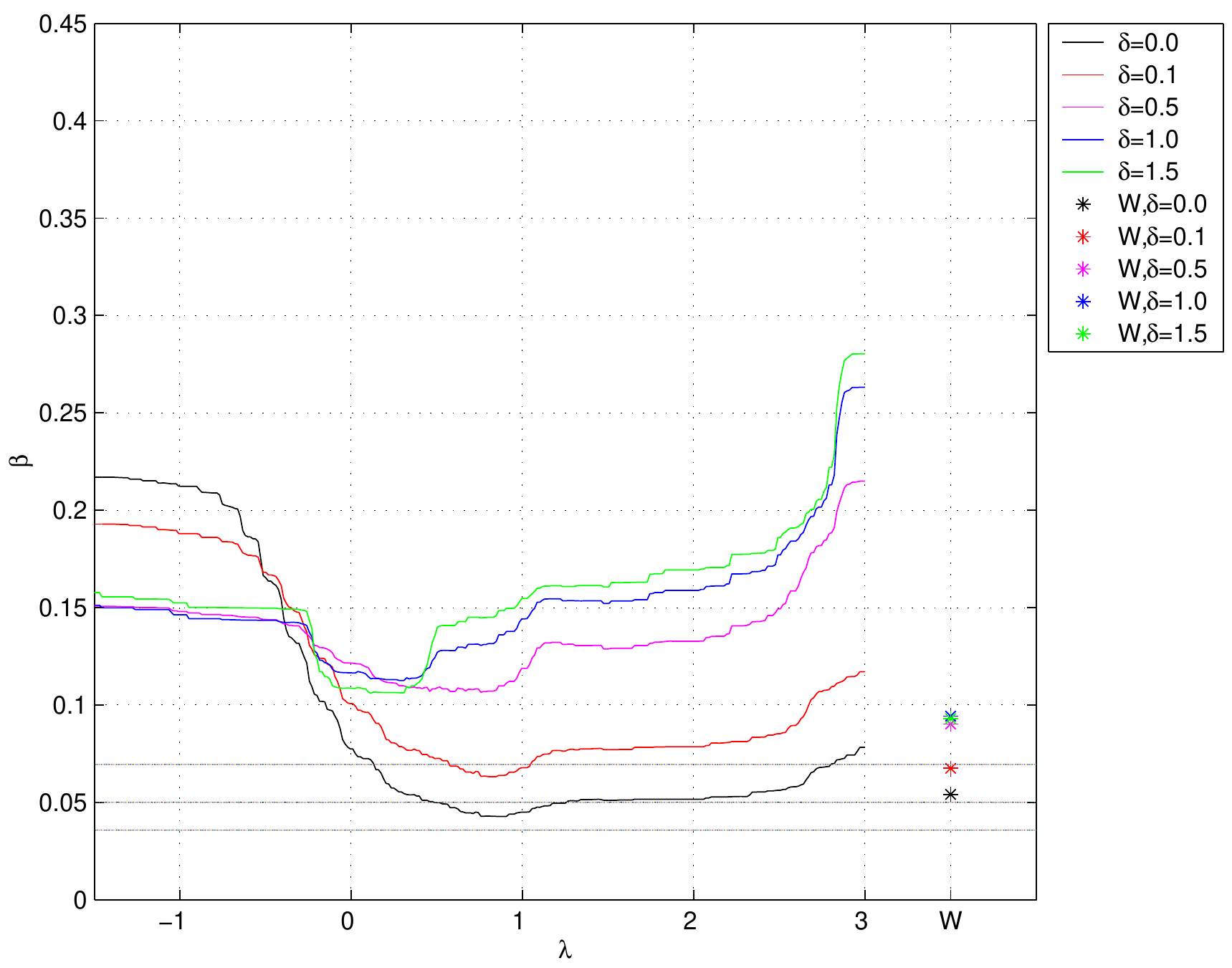}%
}
\\%
{\includegraphics[
height=2.4561in,
width=3.0701in
]%
{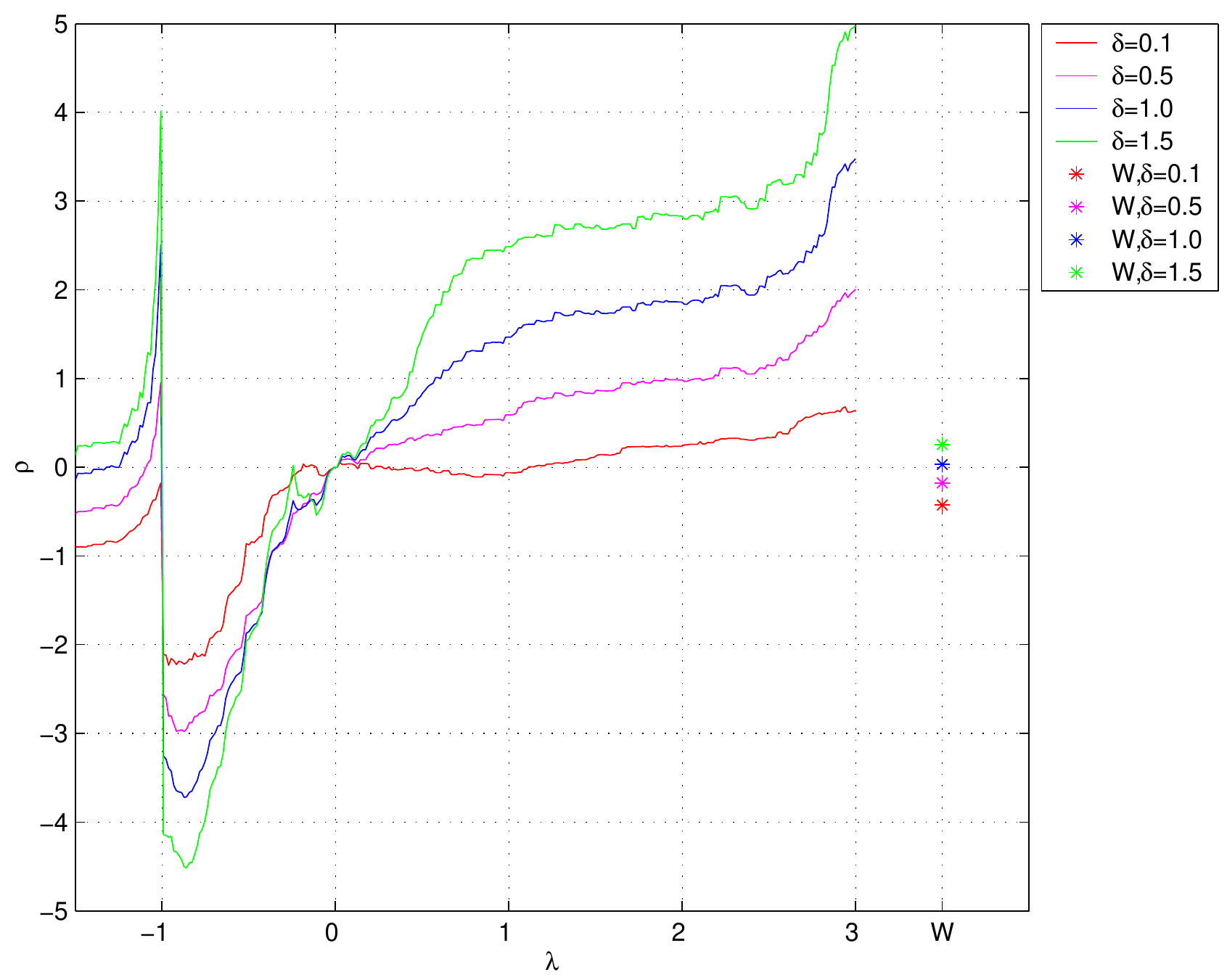}%
}
&
{\includegraphics[
height=2.4561in,
width=3.0701in
]%
{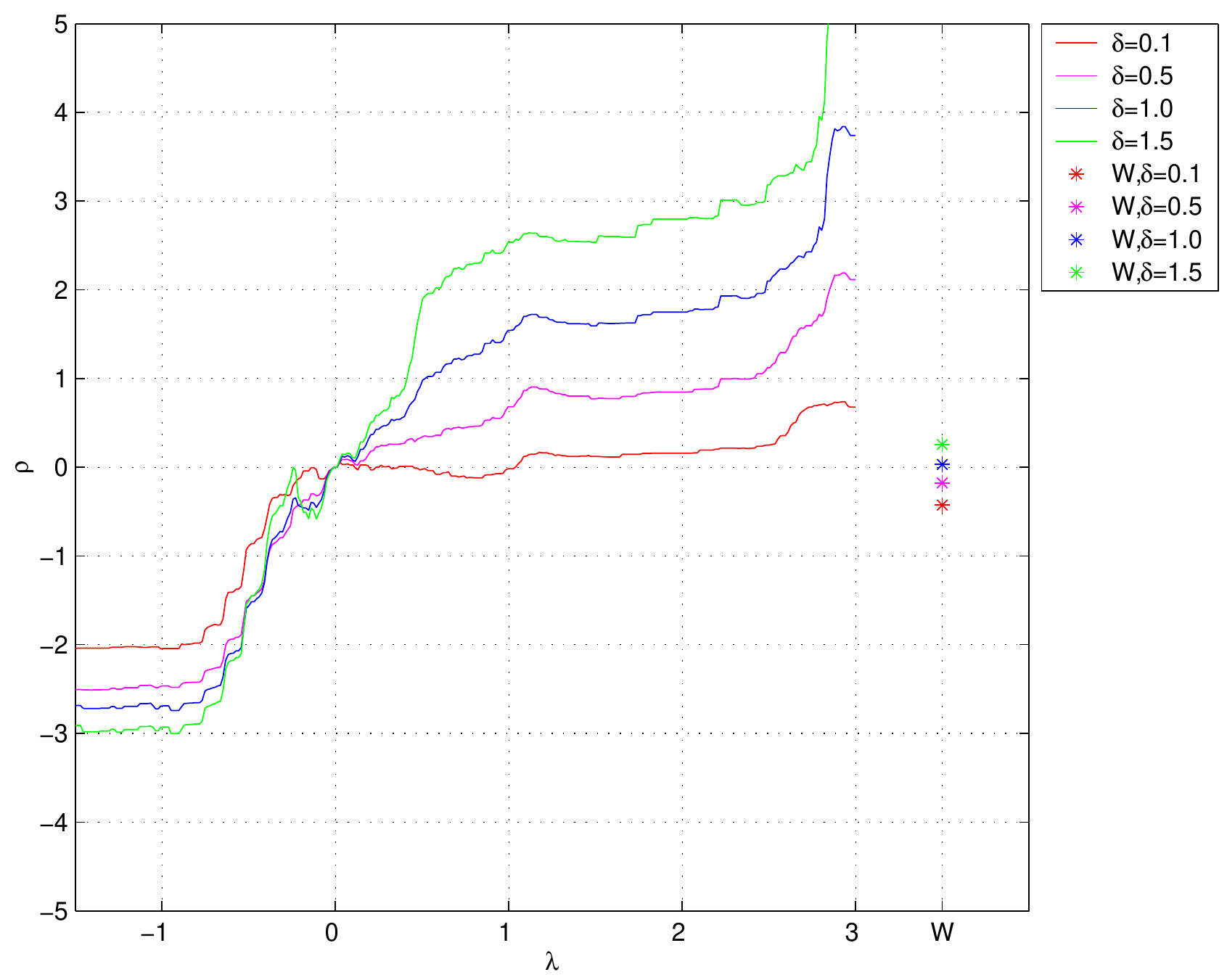}%
}
\\%
{\includegraphics[
height=2.4552in,
width=3.0701in
]%
{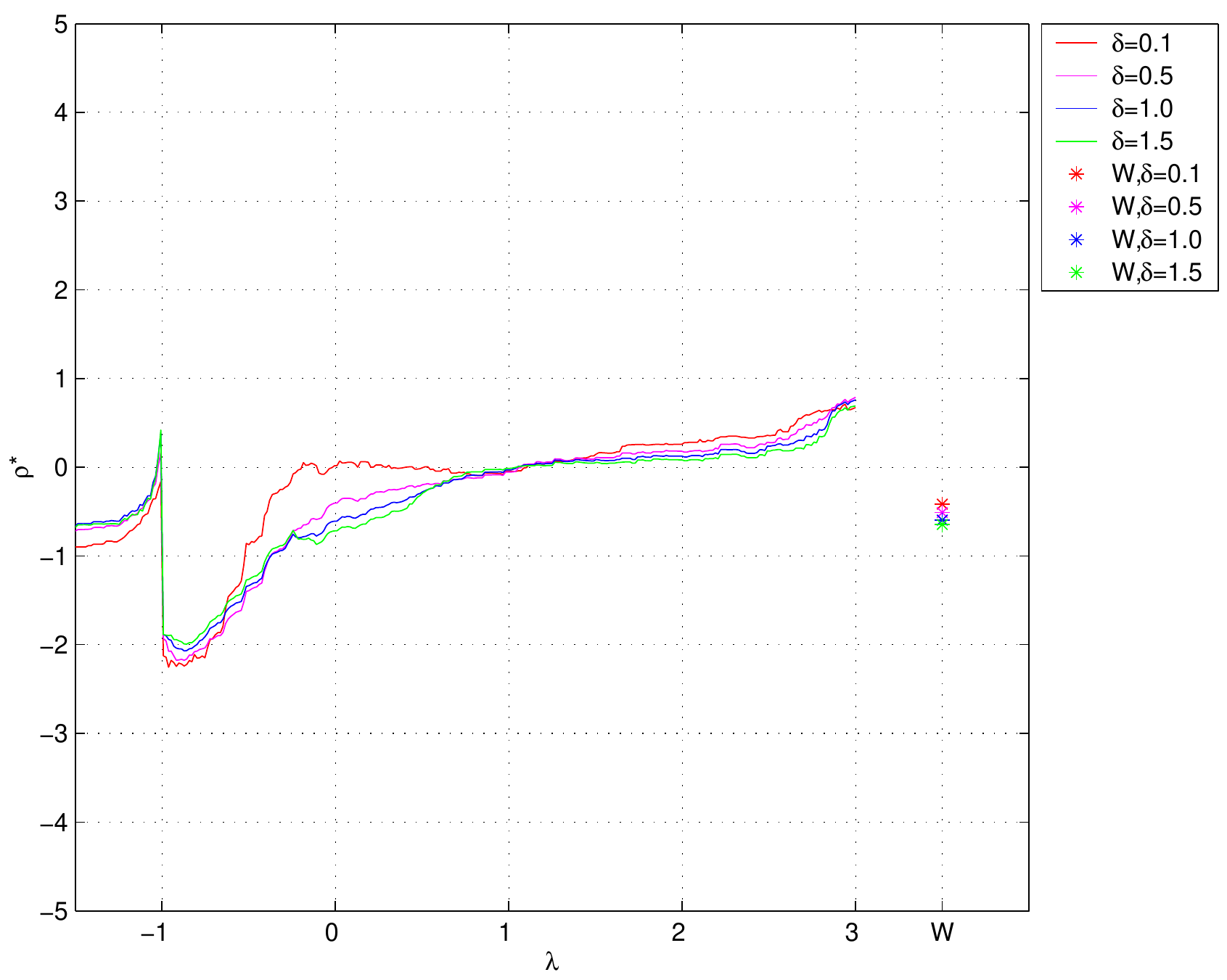}%
}
&
{\includegraphics[
height=2.4561in,
width=3.0701in
]%
{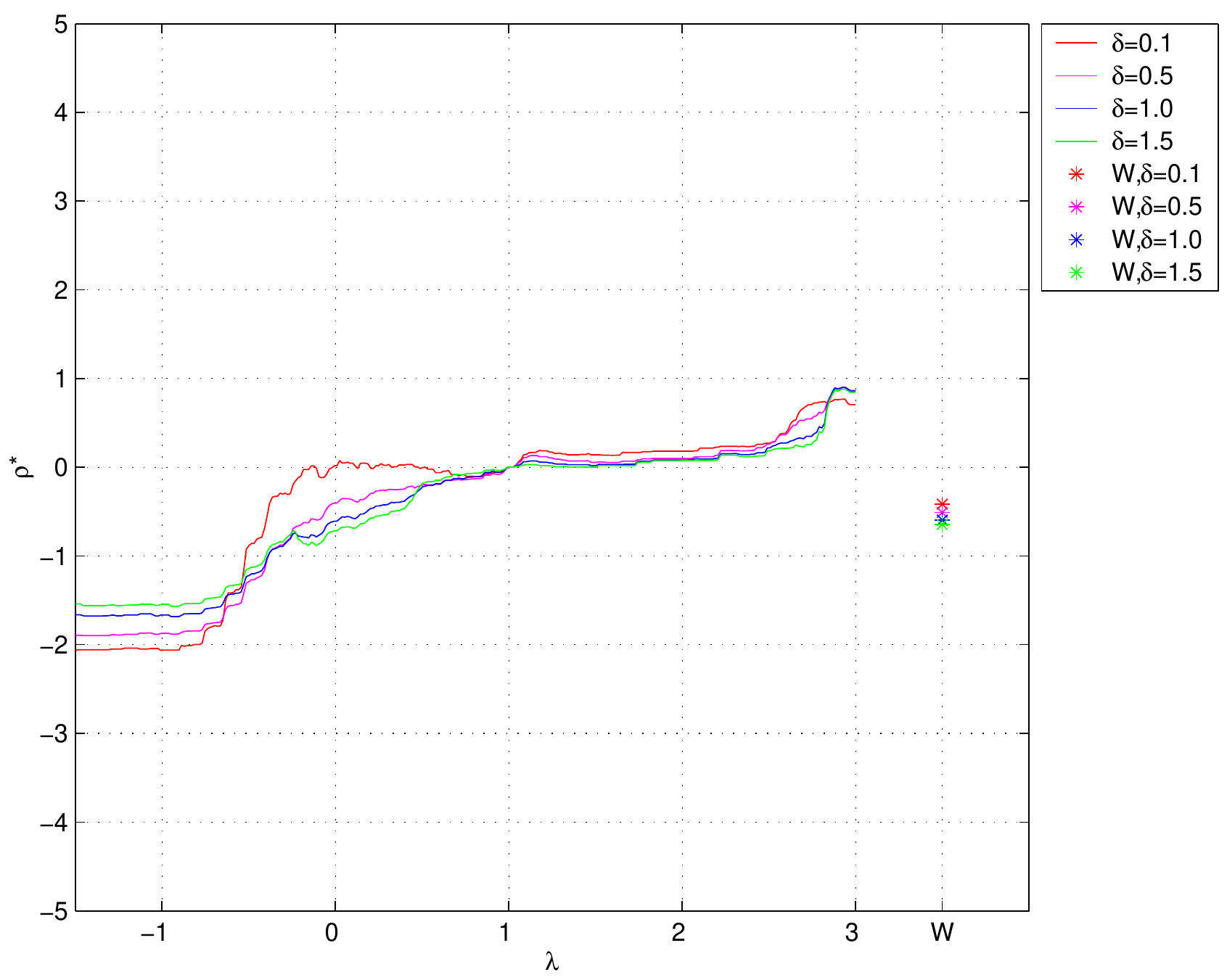}%
}
\end{tabular}
\caption{Power and relative local efficiencies for $T_{\lambda}$, $S_{\lambda}$ and $W$ in scenario G. \label{fig7}}%
\end{figure}%

\pagebreak

The plots are interpreted as follows:\medskip\newline\textbf{a)} In all the
scenarios a similar pattern is observed when plotting the exact power,
${\widehat{\beta}}_{T}$, for $\lambda\in\lbrack-1,3]$ since a U shaped curve
is obtained. This means that the exact power is higher in the corners of the
interval in comparison with the classical likelihood ratio test ($G^{2}=T_{0}%
$) as well as the classical Pearson test statistic ($X^{2}=S_{1}$), contained
in the middle.\medskip\newline\textbf{b)} If we pay attention on the local
efficiencies with respect to $G^{2}$ and $X^{2}$, $\widehat{\rho}_{T}$ and
$\widehat{\rho}_{T}^{\ast}$, to find positive values of them we need to
consider $\lambda\in\lbrack-1,0)$ or $\lambda\in(1,3]$ and thus it confirms
what was said in a). On the other hand, comparing the left hand
($T={T_{\lambda}}$) side of $\widehat{\rho}_{T}$ with the right side
($T={S_{\lambda}}$) and doing the same for $\widehat{\rho}_{T}^{\ast}$, a
slightly higher values of the local efficiencies of ${S_{\lambda}}$ are seen
in comparison with ${T_{\lambda}}$. For this reason we consider that
${\{S_{\lambda}\}}_{\lambda\in\lbrack-1,0)}$ have a better performance than
the classical test-statistics, $G^{2}$ and $X^{2}$ in scenarios B-E and
${\{S_{\lambda}\}}_{\lambda\in(1,3]}$ have a better performance than the
classical test-statistics, $G^{2}$ and $X^{2}$ in scenarios F-G. The Wilcoxon
test-statistic has in all the scenarios worse performance with respect to the
best classical asymptotic statistic, $G^{2}$ for scenarios B-E and $X^{2}$ for
scenarios F-G.\medskip\newline\textbf{c)} What is not so common in comparison
with usual models of categorical data is to find small size sample sizes with
so good performance in exact size as it happens in the case of the likelihood
ratio order. Moreover, the best test-statistic are not very common to be
selected as those with better performance than the classical ones.\newpage

\section{Concluding remark}

The likelihood ratio ordering is a useful technique for comparing treatments
in clinical trials, for this reason it is vitally important to provide
test-statistics to improve the classical ones. Having considered an asymptotic
distribution for two order restricted treatments, the weights needed to manage
the associated asymptotic chi-bar distribution are calculated in a simple way
and the useful matrix for that, $\boldsymbol{H}(\widehat{\boldsymbol{\theta}%
})$, has an easy interpretation in terms of log-linear modeling. The
simulation study highlights the good performance of the all the proposed tests
in relation to the exact size and the comparison is made in terms of the
power. For small and moderate sample sizes there are better choices than the
likelihood ratio test and the Wilcoxon test-statistics inside the family of
$\phi$-divergences. We think that this is a specific characteristic of the
likelihood ordering, and this is the reason of having obtained as the best
test-statistics a set of values of $\lambda\in\lbrack-1,0)\cup(1,3]$ not very
common in the literature of phi-divergence test-statistics. As exception,
notice that%
\begin{align}
S_{-1/2}  &  =S_{d_{\phi_{-1/2}}}(\boldsymbol{p}(\widetilde{\boldsymbol{\theta
}}),\boldsymbol{p}(\widehat{\boldsymbol{\theta}}))=8n\left(  1-%
{\displaystyle\sum\limits_{i=1}^{2}}
{\displaystyle\sum\limits_{j=1}^{J}}
p_{ij}^{\frac{1}{2}}(\widetilde{\boldsymbol{\theta}})p_{ij}^{\frac{1}{2}%
}(\widehat{\boldsymbol{\theta}})\right) \label{hel}\\
&  =4n%
{\displaystyle\sum\limits_{i=1}^{2}}
{\displaystyle\sum\limits_{j=1}^{J}}
\left(  p_{ij}^{\frac{1}{2}}(\widetilde{\boldsymbol{\theta}})-p_{ij}^{\frac
{1}{2}}(\widehat{\boldsymbol{\theta}})\right)  ^{2}\nonumber\\
&  =4n\mathrm{Hel}^{2}(\boldsymbol{p}(\widetilde{\boldsymbol{\theta}%
}),\boldsymbol{p}(\widehat{\boldsymbol{\theta}})),\nonumber
\end{align}
where%
\[
\mathrm{Hel}(\boldsymbol{p}(\widetilde{\boldsymbol{\theta}}),\boldsymbol{p}%
(\widehat{\boldsymbol{\theta}}))=\left(
{\displaystyle\sum\limits_{i=1}^{2}}
{\displaystyle\sum\limits_{j=1}^{J}}
\left(  p_{ij}^{\frac{1}{2}}(\widetilde{\boldsymbol{\theta}})-p_{ij}^{\frac
{1}{2}}(\widehat{\boldsymbol{\theta}})\right)  ^{2}\right)  ^{\frac{1}{2}},
\]
is the Hellinger distance between the probability vectors $\boldsymbol{p}%
(\widetilde{\boldsymbol{\theta}})$\ and $\boldsymbol{p}%
(\widehat{\boldsymbol{\theta}})$. Therefore, one of the test-statistic we are
proposing in this paper is a function of the well-known Hellinger distance,
which has been used in many different statistical problems. We think that the
reason why this happens is related to the robust properties of such a
test-statistic, since when dealing with the likelihood ratio ordering, under
the alternative hypothesis, on the left side of the contingency table empty
cells tend to appear. In particular, the theoretical probability in the first
cell for the second treatment, $\pi_{21}$, is the smallest one and this
circumstance does influence in the results obtained for skew sample sample
sizes in both treatments.

\begin{acknowledgement}
The authors acknowledge the referee. We modified and improved the manuscript
according to comments and questions pointed by the referee.
\end{acknowledgement}

%

\appendix

\section{Appendix}

Suppose we are interested in testing $H_{0}$: $\boldsymbol{R}_{12}%
\boldsymbol{\theta}_{12}=\boldsymbol{0}_{J-1}$ vs\ $H_{1}:\quad\boldsymbol{R}%
_{12}(S)\boldsymbol{\theta}_{12}=\boldsymbol{0}_{\mathrm{card}(S)}$ and
$\boldsymbol{R}_{12}\boldsymbol{\theta}_{12}\neq\boldsymbol{0}_{J-1}$. With
the complete notation, our interest is,%
\begin{equation}
H_{0}:\quad\boldsymbol{R\theta}=\boldsymbol{0}_{J-1}\quad\text{vs}\quad
H_{1}:\quad\boldsymbol{R}(S)\boldsymbol{\theta}=\boldsymbol{0}_{\mathrm{card}%
(S)}\quad\text{and}\quad\boldsymbol{R\theta}\neq\boldsymbol{0}_{J-1}.
\label{TB}%
\end{equation}
Under $H_{0}$, the parameter space is $\Theta_{0}=\left\{  \boldsymbol{\theta
}\in%
\mathbb{R}
^{^{2(J-1)}}:\boldsymbol{R\theta}=\boldsymbol{0}_{J-1}\right\}  $ and the MLE
of $\boldsymbol{\theta}$ in $\Theta_{0}$ is given by
$\widehat{\boldsymbol{\theta}}=\arg\max_{\boldsymbol{\theta\in}\Theta_{0}}%
\ell(\boldsymbol{N};\boldsymbol{\theta})$. Under the alternative hypothesis
the parameter space is $\Theta(S)-\Theta_{0}$, where $\Theta(S)=\left\{
\boldsymbol{\theta}\in%
\mathbb{R}
^{^{2(J-1)}}:\boldsymbol{R}(S)\boldsymbol{\theta}=\boldsymbol{0}%
_{J-1}\right\}  $, that is, under both hypotheses, $H_{0}$\ and $H_{1}$, the
parameter space is $\Theta(S)=\left\{  \boldsymbol{\theta}\in%
\mathbb{R}
^{^{2(J-1)}}:\boldsymbol{R}(S)\boldsymbol{\theta}=\boldsymbol{0}%
_{J-1}\right\}  $ and the MLE of $\boldsymbol{\theta}$ in $\Theta(S)$ is
$\widehat{\boldsymbol{\theta}}(S)=\arg\max_{\boldsymbol{\theta\in}\Theta
(S)}\ell(\boldsymbol{N};\boldsymbol{\theta})$. By following the same idea we
used for building test-statistics (\ref{5a})-(\ref{5b}) we shall consider two
family of test-statistics based on $\phi$-divergence measures,%
\begin{equation}
T_{\phi}(\overline{\boldsymbol{p}},\boldsymbol{p}(\widehat{\boldsymbol{\theta
}}(S)),\boldsymbol{p}(\widehat{\boldsymbol{\theta}}))=2n(d_{\phi}%
(\overline{\boldsymbol{p}},\boldsymbol{p}(\widehat{\boldsymbol{\theta}%
}))-d_{\phi}(\overline{\boldsymbol{p}},\boldsymbol{p}%
(\widehat{\boldsymbol{\theta}}(S)))) \label{5aB}%
\end{equation}
and%
\begin{equation}
S_{\phi}(\boldsymbol{p}(\widehat{\boldsymbol{\theta}}(S)),\boldsymbol{p}%
(\widehat{\boldsymbol{\theta}}))=2nd_{\phi}(\boldsymbol{p}%
(\widehat{\boldsymbol{\theta}}(S)),\boldsymbol{p}(\widehat{\boldsymbol{\theta
}})). \label{5bB}%
\end{equation}

\subsection{Proposition\label{Th1Contr}}

Under $H_{0}$,
\begin{equation}
S_{\phi}(\boldsymbol{p}(\widehat{\boldsymbol{\theta}}(S)),\boldsymbol{p}%
(\widehat{\boldsymbol{\theta}}))=T_{\phi}(\overline{\boldsymbol{p}%
},\boldsymbol{p}(\widehat{\boldsymbol{\theta}}(S)),\boldsymbol{p}%
(\widehat{\boldsymbol{\theta}}))+\mathrm{o}_{p}(1), \label{D}%
\end{equation}
the asymptotic distribution of (\ref{5aB})\ and (\ref{5bB}) is $\chi_{df}^{2}$
with $df=J-1-$\textrm{$card$}$(S)$.

\begin{proof}
The second order Taylor expansion of function $\mathrm{d}_{\phi}%
(\boldsymbol{\theta})=\mathrm{d}_{\phi}(\boldsymbol{p}(\boldsymbol{\theta
}),\boldsymbol{p}(\widehat{\boldsymbol{\theta}}))$ about
$\widehat{\boldsymbol{\theta}}$ is%
\begin{equation}
\mathrm{d}_{\phi}(\boldsymbol{\theta})=\mathrm{d}_{\phi}%
(\widehat{\boldsymbol{\theta}})+(\boldsymbol{\theta}%
-\widehat{\boldsymbol{\theta}})^{T}\left.  \frac{\partial}{\partial
\boldsymbol{\theta}}\mathrm{d}_{\phi}(\boldsymbol{\theta})\right\vert
_{\boldsymbol{\theta=}\widehat{\boldsymbol{\theta}}}+\frac{1}{2}%
(\boldsymbol{\theta}-\widehat{\boldsymbol{\theta}})^{T}\left.  \frac
{\partial^{2}}{\partial\boldsymbol{\theta}\partial\boldsymbol{\theta}^{T}%
}\mathrm{d}_{\phi}(\boldsymbol{\theta})\right\vert _{\boldsymbol{\theta
=}\widehat{\boldsymbol{\theta}}}(\boldsymbol{\theta}%
-\widehat{\boldsymbol{\theta}})+\mathrm{o}\left(  \left\Vert
\boldsymbol{\theta}-\widehat{\boldsymbol{\theta}}\right\Vert ^{2}\right)  ,
\label{eq16}%
\end{equation}
where%
\begin{align*}
\left.  \frac{\partial}{\partial\boldsymbol{\theta}}\mathrm{d}_{\phi
}(\boldsymbol{\theta})\right\vert _{\boldsymbol{\theta=}%
\widehat{\boldsymbol{\theta}}}  &  =\boldsymbol{0}_{J-1},\\
\left.  \frac{\partial^{2}}{\partial\boldsymbol{\theta}\partial
\boldsymbol{\theta}^{T}}\mathrm{d}_{\phi}(\boldsymbol{\theta})\right\vert
_{\boldsymbol{\theta=}\widehat{\boldsymbol{\theta}}}  &  =\phi^{\prime\prime
}\left(  1\right)  \mathcal{I}_{F}^{(n_{1},n_{2})}(\widehat{\boldsymbol{\theta
}}),
\end{align*}
and $\mathcal{I}_{F}^{(n_{1},n_{2})}(\boldsymbol{\theta})$ was defined at the
beginning of Section \ref{sec:Main results}. Let $\overline{\boldsymbol{\theta
}}$ be the parameter vector such that $\overline{\boldsymbol{p}}%
=\boldsymbol{p}(\overline{\boldsymbol{\theta}})$, where $\boldsymbol{p}%
(\overline{\boldsymbol{\theta}})=\boldsymbol{1}_{2J}\bar{u}+\boldsymbol{W}%
\overline{\boldsymbol{\theta}}$, with $\bar{u}=-\log(\boldsymbol{1}_{2J}%
^{T}\exp\{\boldsymbol{W}\overline{\boldsymbol{\theta}}\})$, is the saturated
log-linear model. In particular, for $\boldsymbol{\theta=}\overline
{\boldsymbol{\theta}}$ we have%
\[
\mathrm{d}_{\phi}(\boldsymbol{p}(\overline{\boldsymbol{\theta}}%
),\boldsymbol{p}(\widehat{\boldsymbol{\theta}}))=\frac{\phi^{\prime\prime
}\left(  1\right)  }{2}(\overline{\boldsymbol{\theta}}%
-\widehat{\boldsymbol{\theta}})^{T}\mathcal{I}_{F}^{(n_{1},n_{2}%
)}(\widehat{\boldsymbol{\theta}})(\overline{\boldsymbol{\theta}}%
-\widehat{\boldsymbol{\theta}})+\mathrm{o}\left(  \left\Vert \overline
{\boldsymbol{\theta}}-\widehat{\boldsymbol{\theta}}\right\Vert ^{2}\right)  .
\]
In a similar way it is obtained%
\[
\mathrm{d}_{\phi}(\boldsymbol{p}(\overline{\boldsymbol{\theta}}%
),\boldsymbol{p}(\widehat{\boldsymbol{\theta}}(S)))=\frac{\phi^{\prime\prime
}\left(  1\right)  }{2}(\overline{\boldsymbol{\theta}}%
-\widehat{\boldsymbol{\theta}}(S))^{T}\mathcal{I}_{F}^{(n_{1},n_{2}%
)}(\widehat{\boldsymbol{\theta}}(S))(\overline{\boldsymbol{\theta}%
}-\widehat{\boldsymbol{\theta}}(S))+\mathrm{o}\left(  \left\Vert
\overline{\boldsymbol{\theta}}-\widehat{\boldsymbol{\theta}}(S)\right\Vert
^{2}\right)  .
\]
Multiplying both sides of the equality by $\frac{2n}{\phi^{\prime\prime
}\left(  1\right)  }$ and taking the difference in both sides of the equality%
\begin{align*}
T_{\phi}(\overline{\boldsymbol{p}},\boldsymbol{p}(\widehat{\boldsymbol{\theta
}}(S)),\boldsymbol{p}(\widehat{\boldsymbol{\theta}}))  &  =\frac{2n}%
{\phi^{\prime\prime}(1)}\left(  \mathrm{d}_{\phi}(\boldsymbol{p}%
(\overline{\boldsymbol{\theta}}),\boldsymbol{p}(\widehat{\boldsymbol{\theta}%
}))-\mathrm{d}_{\phi}(\boldsymbol{p}(\overline{\boldsymbol{\theta}%
}),\boldsymbol{p}(\widehat{\boldsymbol{\theta}}(S)))\right) \\
&  =\sqrt{n}(\overline{\boldsymbol{\theta}}-\widehat{\boldsymbol{\theta}}%
)^{T}\mathcal{I}_{F}^{(n_{1},n_{2})}(\widehat{\boldsymbol{\theta}})\sqrt
{n}(\overline{\boldsymbol{\theta}}-\widehat{\boldsymbol{\theta}}%
)+\mathrm{o}\left(  \left\Vert \sqrt{n}\left(  \overline{\boldsymbol{\theta}%
}-\widehat{\boldsymbol{\theta}}\right)  \right\Vert ^{2}\right) \\
&  -\sqrt{n}(\overline{\boldsymbol{\theta}}-\widehat{\boldsymbol{\theta}%
}(S))^{T}\mathcal{I}_{F}^{(n_{1},n_{2})}(\widehat{\boldsymbol{\theta}%
}(S))\sqrt{n}(\overline{\boldsymbol{\theta}}-\widehat{\boldsymbol{\theta}%
}(S))+\mathrm{o}\left(  \left\Vert \sqrt{n}\left(  \overline
{\boldsymbol{\theta}}-\widehat{\boldsymbol{\theta}}(S)\right)  \right\Vert
^{2}\right)  .
\end{align*}
Now we are going to generalize the three types of estimators by
$\widehat{\boldsymbol{\theta}}(\bullet)$, understanding that for
$\bullet=\varnothing$, $\widehat{\boldsymbol{\theta}}(\varnothing
)=\overline{\boldsymbol{\theta}}$, $\boldsymbol{R}(\varnothing\mathbf{)=0}%
_{(J-1)\times(2J-1)}$, for $\bullet=E$, $\widehat{\boldsymbol{\theta}%
}(E)=\widehat{\boldsymbol{\theta}}$, $\boldsymbol{R}(E\mathbf{)=}%
\boldsymbol{R}$, and $\bullet=S$, $\widehat{\boldsymbol{\theta}}(S)$ and
$\boldsymbol{R}(S\mathbf{)}$ as originally defined. It is well-known that%
\begin{equation}
\sqrt{n}(\widehat{\boldsymbol{\theta}}(\bullet)-\boldsymbol{\theta}%
_{0})=\boldsymbol{\Gamma}(\boldsymbol{\theta}_{0},\bullet\mathbf{)}\frac
{1}{\sqrt{n}}\left.  \frac{\partial}{\partial\boldsymbol{\theta}}%
\ell(\boldsymbol{N};\boldsymbol{\theta})\right\vert _{\boldsymbol{\theta
}=\boldsymbol{\theta}_{0}}+\mathrm{o}_{p}(\boldsymbol{1}_{^{2(J-1)}}),
\label{eq14}%
\end{equation}
where $\boldsymbol{\theta}_{0}$ is the true and unknown value of the
parameter,%
\begin{equation}
\boldsymbol{\Gamma}(\boldsymbol{\theta}_{0},\bullet\mathbf{)}\mathbf{=}%
\mathcal{I}_{F}^{-1}(\boldsymbol{\theta}_{0})-\mathcal{I}_{F}^{-1}%
(\boldsymbol{\theta}_{0})\boldsymbol{R}^{T}(\bullet\mathbf{)}\left(
\boldsymbol{R}(\bullet\mathbf{)}\mathcal{I}_{F}^{-1}(\boldsymbol{\theta}%
_{0})\boldsymbol{R}^{T}(\bullet\mathbf{)}\right)  ^{-1}\boldsymbol{R}%
(\bullet\mathbf{)}\mathcal{I}_{F}^{-1}(\boldsymbol{\theta}_{0}), \label{PP}%
\end{equation}
is the variance covariance matrix of $\widehat{\boldsymbol{\theta}}(\bullet
)$,\ and $\frac{1}{\sqrt{n}}\left.  \frac{\partial}{\partial\boldsymbol{\theta
}}\ell(\boldsymbol{N};\boldsymbol{\theta})\right\vert _{\boldsymbol{\theta
}=\boldsymbol{\theta}_{0}}\underset{n_{1},n_{2}\rightarrow\infty
}{\overset{\mathcal{L}}{\longrightarrow}}\mathcal{N}(\boldsymbol{0}%
_{k},\mathcal{I}_{F}(\boldsymbol{\theta}_{0}))$ by the Central Limit Theorem.
We shall denote%
\[
\boldsymbol{\Gamma}(\boldsymbol{\theta}_{0}\mathbf{)=}\boldsymbol{\Gamma
}(\boldsymbol{\theta}_{0},E\mathbf{)=}\mathcal{I}_{F}^{-1}(\boldsymbol{\theta
}_{0})-\mathcal{I}_{F}^{-1}(\boldsymbol{\theta}_{0})\boldsymbol{R}^{T}\left(
\boldsymbol{R}\mathcal{I}_{F}^{-1}(\boldsymbol{\theta}_{0})\boldsymbol{R}%
^{T}\right)  ^{-1}\boldsymbol{R}\mathcal{I}_{F}^{-1}(\boldsymbol{\theta}%
_{0}).
\]
Taking the differences of both sides of the equality in (\ref{eq14}) with
cases $\bullet=\varnothing$\ and $\bullet=E$, we obtain%
\begin{equation}
\sqrt{n}(\overline{\boldsymbol{\theta}}-\widehat{\boldsymbol{\theta}})=\left(
\mathcal{I}_{F}^{-1}(\boldsymbol{\theta}_{0})-\boldsymbol{\Gamma
}(\boldsymbol{\theta}_{0}\mathbf{)}\right)  \frac{1}{\sqrt{n}}\left.
\frac{\partial}{\partial\boldsymbol{\theta}}\ell(\boldsymbol{N}%
;\boldsymbol{\theta})\right\vert _{\boldsymbol{\theta}=\boldsymbol{\theta}%
_{0}}+\mathrm{o}_{p}(\boldsymbol{1}_{^{2(J-1)}}), \label{A}%
\end{equation}
with cases $\bullet=\varnothing$\ and $\bullet=S$,%
\begin{equation}
\sqrt{n}(\overline{\boldsymbol{\theta}}-\widehat{\boldsymbol{\theta}%
}(S))=\left(  \mathcal{I}_{F}^{-1}(\boldsymbol{\theta}_{0})-\boldsymbol{\Gamma
}(\boldsymbol{\theta}_{0},S\mathbf{)}\right)  \frac{1}{\sqrt{n}}\left.
\frac{\partial}{\partial\boldsymbol{\theta}}\ell(\boldsymbol{N}%
;\boldsymbol{\theta})\right\vert _{\boldsymbol{\theta}=\boldsymbol{\theta}%
_{0}}+\mathrm{o}_{p}(\boldsymbol{1}_{^{2(J-1)}}), \label{B}%
\end{equation}
and taking into account $\mathcal{I}_{F}(\widehat{\boldsymbol{\theta}%
})\underset{n_{1},n_{2}\rightarrow\infty}{\overset{P}{\longrightarrow}%
}\mathcal{I}_{F}(\boldsymbol{\theta}_{0})$,%
\begin{align}
&  T_{\phi}(\overline{\boldsymbol{p}},\boldsymbol{p}%
(\widehat{\boldsymbol{\theta}}(S)),\boldsymbol{p}(\widehat{\boldsymbol{\theta
}}))\nonumber\\
&  =\frac{1}{\sqrt{n}}\left.  \frac{\partial}{\partial\boldsymbol{\theta}^{T}%
}\ell(\boldsymbol{N};\boldsymbol{\theta})\right\vert _{\boldsymbol{\theta
}=\boldsymbol{\theta}_{0}}\left(  \boldsymbol{\Gamma}(\boldsymbol{\theta}%
_{0},S\mathbf{)}-\boldsymbol{\Gamma}(\boldsymbol{\theta}_{0}\mathbf{)}\right)
^{T}\mathcal{I}_{F}(\boldsymbol{\theta}_{0})\left(  \boldsymbol{\Gamma
}(\boldsymbol{\theta}_{0},S\mathbf{)}-\boldsymbol{\Gamma}(\boldsymbol{\theta
}_{0}\mathbf{)}\right)  \frac{1}{\sqrt{n}}\left.  \frac{\partial}%
{\partial\boldsymbol{\theta}}\ell(\boldsymbol{N};\boldsymbol{\theta
})\right\vert _{\boldsymbol{\theta}=\boldsymbol{\theta}_{0}}+\mathrm{o}%
_{p}(1)\nonumber\\
&  =\boldsymbol{Y}^{T}\boldsymbol{Y}+\mathrm{o}_{p}(1), \label{C}%
\end{align}
where%
\[
\boldsymbol{Y}=\boldsymbol{A}(\boldsymbol{\theta}_{0}\mathbf{)}\left(
\boldsymbol{\Gamma}(\boldsymbol{\theta}_{0},S\mathbf{)}-\boldsymbol{\Gamma
}(\boldsymbol{\theta}_{0}\mathbf{)}\right)  \boldsymbol{A}(\boldsymbol{\theta
}_{0}\mathbf{)}^{T}\boldsymbol{Z}\text{,}%
\]
with $\boldsymbol{Z}\sim\mathcal{N}(\boldsymbol{0}_{J-1},\boldsymbol{I}%
_{J-1}\mathbf{)}$ and $\boldsymbol{A}(\boldsymbol{\theta}_{0}\mathbf{)}$\ is
the Cholesky's factorization matrix for a non singular matrix such a Fisher
information matrix, that is $\mathcal{I}_{F}(\boldsymbol{\theta}%
_{0})=\boldsymbol{A}(\boldsymbol{\theta}_{0}\mathbf{)}^{T}\boldsymbol{A}%
(\boldsymbol{\theta}_{0}\mathbf{)}$. In other words%
\[
\boldsymbol{Y}\sim\mathcal{N}(\boldsymbol{0}_{k},\boldsymbol{A}%
(\boldsymbol{\theta}_{0}\mathbf{)}\left(  \boldsymbol{\Gamma}%
(\boldsymbol{\theta}_{0},S\mathbf{)}-\boldsymbol{\Gamma}(\boldsymbol{\theta
}_{0}\mathbf{)}\right)  \boldsymbol{A}(\boldsymbol{\theta}_{0}\mathbf{)}%
^{T})\text{,}%
\]
where the variance covariance matrix is idempotent and symmetric. Following
Lemma 3 in Ferguson (1996, page 57), $\boldsymbol{A}(\boldsymbol{\theta}%
_{0}\mathbf{)}\left(  \boldsymbol{\Gamma}(\boldsymbol{\theta}_{0}%
,S\mathbf{)}-\boldsymbol{\Gamma}(\boldsymbol{\theta}_{0}\mathbf{)}\right)
\boldsymbol{A}(\boldsymbol{\theta}_{0}\mathbf{)}^{T}$ is idempotent and
symmetric, if only if $T_{\phi}(\overline{\boldsymbol{p}},\boldsymbol{p}%
(\widehat{\boldsymbol{\theta}}(S)),\boldsymbol{p}(\widehat{\boldsymbol{\theta
}}))$ is a chi-square random variable with degrees of freedom
\[
df=\mathrm{rank}(\boldsymbol{A}(\boldsymbol{\theta}_{0}\mathbf{)}\left(
\boldsymbol{\Gamma}(\boldsymbol{\theta}_{0},S\mathbf{)}-\boldsymbol{\Gamma
}(\boldsymbol{\theta}_{0}\mathbf{)}\right)  \boldsymbol{A}(\boldsymbol{\theta
}_{0}\mathbf{)}^{T})=\mathrm{trace}(\boldsymbol{A}(\boldsymbol{\theta}%
_{0}\mathbf{)}\left(  \boldsymbol{\Gamma}(\boldsymbol{\theta}_{0}%
,S\mathbf{)}-\boldsymbol{\Gamma}(\boldsymbol{\theta}_{0}\mathbf{)}\right)
\boldsymbol{A}(\boldsymbol{\theta}_{0}\mathbf{)}^{T}).
\]
Since%
\[
\left(  \boldsymbol{\Gamma}(\boldsymbol{\theta}_{0},S\mathbf{)}%
-\boldsymbol{\Gamma}(\boldsymbol{\theta}_{0}\mathbf{)}\right)  ^{T}%
\mathcal{I}_{F}(\boldsymbol{\theta}_{0})\left(  \boldsymbol{\Gamma
}(\boldsymbol{\theta}_{0},S\mathbf{)}-\boldsymbol{\Gamma}(\boldsymbol{\theta
}_{0}\mathbf{)}\right)  =\boldsymbol{\Gamma}(\boldsymbol{\theta}%
_{0},S\mathbf{)}-\boldsymbol{\Gamma}(\boldsymbol{\theta}_{0}\mathbf{),}%
\]
the condition is reached. The effective degrees of freedom are given by%
\begin{align*}
df  &  =\mathrm{trace}(\boldsymbol{\Gamma}(\boldsymbol{\theta}_{0}%
,S\mathbf{)}\boldsymbol{A}(\boldsymbol{\theta}_{0}\mathbf{)}^{T}%
\boldsymbol{A}(\boldsymbol{\theta}_{0}\mathbf{)})-\mathrm{trace}%
(\boldsymbol{\Gamma}(\boldsymbol{\theta}_{0}\mathbf{)}\boldsymbol{A}%
(\boldsymbol{\theta}_{0}\mathbf{)}^{T}\boldsymbol{A}(\boldsymbol{\theta}%
_{0}\mathbf{)})=\mathrm{trace}(\boldsymbol{\Gamma}(\boldsymbol{\theta}%
_{0},S\mathbf{)}\mathcal{I}_{F}(\boldsymbol{\theta}_{0}))-\mathrm{trace}%
(\boldsymbol{\Gamma}(\boldsymbol{\theta}_{0}\mathbf{)}\mathcal{I}%
_{F}(\boldsymbol{\theta}_{0}))\\
&  =\mathrm{trace}(-\left(  \boldsymbol{R}(S\mathbf{)}\mathcal{I}_{F}%
^{-1}(\boldsymbol{\theta}_{0})\boldsymbol{R}^{T}(S\mathbf{)}\right)
^{-1}\boldsymbol{R}(S\mathbf{)}\mathcal{I}_{F}^{-1}(\boldsymbol{\theta}%
_{0})\boldsymbol{R}^{T}(S\mathbf{)})\\
&  -\mathrm{trace}(-\left(  \boldsymbol{R}\mathcal{I}_{F}^{-1}%
(\boldsymbol{\theta}_{0})\boldsymbol{R}^{T}\right)  ^{-1}\boldsymbol{R}%
\mathcal{I}_{F}^{-1}(\boldsymbol{\theta}_{0})\boldsymbol{R}^{T})\\
&  =(J-1)-\mathrm{card}(S).
\end{align*}
Regarding the other test-statistic $S_{\phi}(\boldsymbol{p}%
(\widehat{\boldsymbol{\theta}}(S)),\boldsymbol{p}(\widehat{\boldsymbol{\theta
}}))$, observe that if we take (\ref{eq16}), in particular for
$\boldsymbol{\theta=}\widehat{\boldsymbol{\theta}}(S)$ it is obtained%
\[
\mathrm{d}_{\phi}(\widehat{\boldsymbol{\theta}}(S))=\frac{\phi^{\prime\prime
}\left(  1\right)  }{2}(\widehat{\boldsymbol{\theta}}%
(S)-\widehat{\boldsymbol{\theta}})^{T}\mathcal{I}_{F}%
(\widehat{\boldsymbol{\theta}})(\widehat{\boldsymbol{\theta}}%
(S)-\widehat{\boldsymbol{\theta}})+\mathrm{o}\left(  \left\Vert
\widehat{\boldsymbol{\theta}}(S)-\widehat{\boldsymbol{\theta}}\right\Vert
^{2}\right)  .
\]
In addition, (\ref{A})$-$(\ref{B}) is%
\[
\sqrt{n}(\widehat{\boldsymbol{\theta}}(S)-\widehat{\boldsymbol{\theta}%
})=\left(  \boldsymbol{\Gamma}(\boldsymbol{\theta}_{0},S\mathbf{)}%
-\boldsymbol{\Gamma}(\boldsymbol{\theta}_{0}\mathbf{)}\right)  \frac{1}%
{\sqrt{n}}\left.  \frac{\partial}{\partial\boldsymbol{\theta}}\ell
(\boldsymbol{N};\boldsymbol{\theta})\right\vert _{\boldsymbol{\theta
}=\boldsymbol{\theta}_{0}}+\mathrm{o}_{p}(\boldsymbol{1}_{^{2(J-1)}}),
\]
and taking into account $\mathcal{I}_{F}(\widehat{\boldsymbol{\theta}%
})\underset{n_{1},n_{2}\rightarrow\infty}{\overset{P}{\longrightarrow}%
}\mathcal{I}_{F}(\boldsymbol{\theta}_{0})$ and (\ref{C}), it follows
(\ref{D}), which means from Slutsky's Theorem that both test-statistics have
the same asymptotic distribution.
\end{proof}

\subsection{Lemma \label{LemContrA}}

Let $\boldsymbol{Y}$ be a $k$-dimensional random variable with normal
distribution $\mathcal{N}\left(  \boldsymbol{0}_{k},\boldsymbol{Q}\right)  $
with $\boldsymbol{Q}$ being a projection matrix, that is idempotent and
symmetric, and let $\boldsymbol{d}_{i}$ be the fixed $k$-dimensional vectors
such that for them either $\boldsymbol{Qd}_{i}=\boldsymbol{0}_{k}$ or
$\boldsymbol{Qd}_{i}=\boldsymbol{d}_{i}$, $i=1,...,k$, is true. Then $\left(
\boldsymbol{Y}^{T}\boldsymbol{Y}\left\vert \boldsymbol{d}_{i}^{T}%
\boldsymbol{Y}\geq0,i=1,...,k\right.  \right)  \sim\chi_{df}^{2}$, where
$df=\mathrm{rank}(\boldsymbol{Q})$.

\begin{proof}
This result can be found in several sources, for instance in Kud\^{o} (1963,
page 414), Barlow et al. (1972, page 128) and Shapiro (1985, page 139).
\end{proof}

\subsection{Proof of Theorem \ref{Th1}\label{ProofTh1ContrA}}

We shall perform the proof for $S_{\phi}(\boldsymbol{p}%
(\widetilde{\boldsymbol{\theta}}),\boldsymbol{p}(\widehat{\boldsymbol{\theta}%
}))$. It suppose that it is true $\boldsymbol{R\theta}\geq\boldsymbol{0}%
_{J-1}$ and we want to test $\boldsymbol{R\theta}=\boldsymbol{0}_{J-1}$
($H_{0}$). It is clear that if $H_{0}$ is not true is because there exists
some index $i\in E$ such that $\boldsymbol{R}(\{i\}\mathbf{)}%
\boldsymbol{\theta}>0$. Let us consider the family of all possible subsets in
$E$, denoted by $\mathcal{F}(E)$, then\ we shall specify more thoroughly
$\widetilde{\boldsymbol{\theta}}$ by $\widetilde{\boldsymbol{\theta}}(S)$ when
there exists $S\in\mathcal{F}(E)$ such that%
\[
\boldsymbol{R}(S)\widetilde{\boldsymbol{\theta}}=\boldsymbol{0}_{\mathrm{card}%
(S)}\qquad\text{and}\qquad\boldsymbol{R}(S^{C})\widetilde{\boldsymbol{\theta}%
}>\boldsymbol{0}_{(J-1)-\mathrm{card}(S)}.
\]
It is clear that for a sample $\widetilde{\boldsymbol{\theta}}%
=\widetilde{\boldsymbol{\theta}}(S)$ can be true only for a unique set of
indices $S\in\mathcal{F}(E)$, and thus by applying the Theorem of Total
Probability%
\[
\Pr\left(  S_{\phi}(\boldsymbol{p}(\widetilde{\boldsymbol{\theta}%
}),\boldsymbol{p}(\widehat{\boldsymbol{\theta}}))\leq x\right)  =\sum
_{S\in\mathcal{F}(E)}\Pr\left(  S_{\phi}(\boldsymbol{p}%
(\widetilde{\boldsymbol{\theta}}),\boldsymbol{p}(\widehat{\boldsymbol{\theta}%
}))\leq x,\widetilde{\boldsymbol{\theta}}=\widetilde{\boldsymbol{\theta}%
}(S)\right)  .
\]
From the Karush-Khun-Tucker necessary conditions (see for instance Theorem
4.2.13 in Bazaraa et al. (2006)) to solve the optimization problem $\max
\ell(\boldsymbol{N};\boldsymbol{\theta})$ s.t. $\boldsymbol{R\theta}%
\geq\boldsymbol{0}_{J-1}$, associated with $\widetilde{\boldsymbol{\theta}}$,%
\begin{subequations}
\begin{align}
\frac{\partial}{\partial\boldsymbol{\theta}}\ell(\boldsymbol{N}%
;\boldsymbol{\theta})+\sum_{i=1}^{J-1}\lambda_{i}\boldsymbol{R}^{T}%
(\{i\}\mathbf{)} &  =0\text{, }i=1,...,J-1,\label{KKT1}\\
\lambda_{i}\boldsymbol{R}(\{i\}\mathbf{)}\boldsymbol{\theta} &  =0\text{,
}i=1,...,J-1,\label{KKT2}\\
\lambda_{i} &  \leq0\text{, }i=1,...,J-1,\label{KKT3}%
\end{align}
the only conditions which characterize the MLE $\widetilde{\boldsymbol{\theta
}}=\widetilde{\boldsymbol{\theta}}(S)$ with a specific $S\in\mathcal{F}(E)$,
are the complementary slackness conditions $\boldsymbol{R}(\{i\}\mathbf{)}%
\boldsymbol{\theta}>0$, for $i\in S$ and $\lambda_{i}<0$, for $i\in S^{C}$,
since $\frac{\partial}{\partial\boldsymbol{\theta}}\ell(\boldsymbol{N}%
;\boldsymbol{\theta})+\lambda_{i}\boldsymbol{R}^{T}(\{i\}\mathbf{)}=0$,
$i=1,...,J-1$,$\ \boldsymbol{R}(\{i\}\mathbf{)}\boldsymbol{\theta}=0$, for
$i\in S^{C}$ and $\lambda_{i}=0$, for $i\in S$ are redundant conditions once
we know that the Karush-Khun-Tucker necessary conditions are true for all the
possible sets $S\in\mathcal{F}(E)$ which define $\widetilde{\boldsymbol{\theta
}}=\widetilde{\boldsymbol{\theta}}(S)$. For this reason we can consider%
\end{subequations}
\begin{align*}
&  \Pr\left(  S_{\phi}(\boldsymbol{p}(\widetilde{\boldsymbol{\theta}%
}),\boldsymbol{p}(\widehat{\boldsymbol{\theta}}))\leq
x,\widetilde{\boldsymbol{\theta}}=\widetilde{\boldsymbol{\theta}}(S)\right)
=\\
&  \Pr\left(  S_{\phi}(\boldsymbol{p}(\widetilde{\boldsymbol{\theta}%
}),\boldsymbol{p}(\widehat{\boldsymbol{\theta}}))\leq
x,\widetilde{\boldsymbol{\lambda}}(S)<\boldsymbol{0}_{\mathrm{card}%
(S)},\boldsymbol{R}(S^{C})\widetilde{\boldsymbol{\theta}}(S)>\boldsymbol{0}%
_{(J-1)-\mathrm{card}(S)}\right)  ,
\end{align*}
where $\widetilde{\boldsymbol{\lambda}}(S)$ is the vector of the vector of
Karush-Khun-Tucker multipliers associated with estimator
$\widetilde{\boldsymbol{\theta}}(S)$. Furthermore, under $H_{0}$,
$\boldsymbol{R}\widetilde{\boldsymbol{\theta}}(S)=\boldsymbol{R}%
\widetilde{\boldsymbol{\theta}}(S)-\boldsymbol{R\theta}_{0}$, because
$\boldsymbol{R\theta}_{0}=\boldsymbol{0}_{J-1}$, hence%
\[
\Pr\left(  S_{\phi}(\boldsymbol{p}(\widetilde{\boldsymbol{\theta}%
}),\boldsymbol{p}(\widehat{\boldsymbol{\theta}}))\leq x\right)  =\sum
_{S\in\mathcal{F}(E)}\Pr\left(  S_{\phi}(\boldsymbol{p}%
(\widetilde{\boldsymbol{\theta}}),\boldsymbol{p}(\widehat{\boldsymbol{\theta}%
}))\leq x,\widetilde{\boldsymbol{\lambda}}(S)<\boldsymbol{0}_{\mathrm{card}%
(S)},\boldsymbol{R}(S^{C})\widetilde{\boldsymbol{\theta}}(S)-\boldsymbol{R}%
(S^{C})\boldsymbol{\theta}_{0}>\boldsymbol{0}_{\mathrm{card}(S^{C})}\right)  ,
\]
where $\mathrm{card}(S^{C})=(J-1)-\mathrm{card}(S)$. On the other hand,
(\ref{KKT1}) and (\ref{KKT2}) are also true for $(\widehat{\boldsymbol{\theta
}}^{T}(S),\widehat{\boldsymbol{\lambda}}^{T}(S))^{T}$ according to the
Lagrange multipliers method. Hence, $\widetilde{\boldsymbol{\theta}%
}(S)=\widehat{\boldsymbol{\theta}}(S)$ and $\widetilde{\boldsymbol{\lambda}%
}(S)=\widehat{\boldsymbol{\lambda}}(S)$. It follows that:\newline$\bullet$
under $\widetilde{\boldsymbol{\theta}}=\widehat{\boldsymbol{\theta}}(S)$,
$S_{\phi}(\boldsymbol{p}(\widetilde{\boldsymbol{\theta}}),\boldsymbol{p}%
(\widehat{\boldsymbol{\theta}}))=S_{\phi}(\boldsymbol{p}%
(\widehat{\boldsymbol{\theta}}(S)),\boldsymbol{p}(\widehat{\boldsymbol{\theta
}}))$ and taking into account Proposition \ref{Th1Contr}
\begin{align*}
&  S_{\phi}(\boldsymbol{p}(\widetilde{\boldsymbol{\theta}}),\boldsymbol{p}%
(\widehat{\boldsymbol{\theta}}))=T_{\phi}(\overline{\boldsymbol{p}%
},\boldsymbol{p}(\widehat{\boldsymbol{\theta}}(S)),\boldsymbol{p}%
(\widehat{\boldsymbol{\theta}}))+\mathrm{o}_{p}(1)\\
&  =\left(  \boldsymbol{A}(\boldsymbol{\theta}_{0}\mathbf{)}\left(
\boldsymbol{\Gamma}(\boldsymbol{\theta}_{0},S\mathbf{)}-\boldsymbol{\Gamma
}(\boldsymbol{\theta}_{0}\mathbf{)}\right)  \boldsymbol{A}(\boldsymbol{\theta
}_{0}\mathbf{)}^{T}\boldsymbol{Z}\right)  ^{T}\left(  \boldsymbol{A}%
(\boldsymbol{\theta}_{0}\mathbf{)}\left(  \boldsymbol{\Gamma}%
(\boldsymbol{\theta}_{0},S\mathbf{)}-\boldsymbol{\Gamma}(\boldsymbol{\theta
}_{0}\mathbf{)}\right)  \boldsymbol{A}(\boldsymbol{\theta}_{0}\mathbf{)}%
^{T}\boldsymbol{Z}\right)  +\mathrm{o}_{p}(1),\\
&  =\boldsymbol{Z}^{T}\boldsymbol{A}(\boldsymbol{\theta}_{0}\mathbf{)}\left(
\boldsymbol{\Gamma}(\boldsymbol{\theta}_{0},S\mathbf{)}-\boldsymbol{\Gamma
}(\boldsymbol{\theta}_{0}\mathbf{)}\right)  \boldsymbol{A}(\boldsymbol{\theta
}_{0}\mathbf{)}^{T}\boldsymbol{Z}+\mathrm{o}_{p}(1).
\end{align*}
where $\boldsymbol{Z}\sim\mathcal{N}\left(  \boldsymbol{0}_{k},\boldsymbol{I}%
_{k}\right)  $.$\newline\bullet$ under $\widetilde{\boldsymbol{\lambda}%
}(S)=\widehat{\boldsymbol{\lambda}}(S)$ and from Sen et al. (2010, page 267
formula (8.6.28))%
\begin{align*}
\frac{1}{\sqrt{n}}\widetilde{\boldsymbol{\lambda}}(S) &  =\sqrt{n}%
\boldsymbol{Q}^{T}(\boldsymbol{\theta}_{0},S\mathbf{)}\frac{1}{\sqrt{n}%
}\left.  \frac{\partial}{\partial\boldsymbol{\theta}}\ell(\boldsymbol{N}%
;\boldsymbol{\theta})\right\vert _{\boldsymbol{\theta}=\boldsymbol{\theta}%
_{0}}+\mathrm{o}_{p}(\boldsymbol{1}_{\mathrm{card}(S)})\\
&  =\boldsymbol{Q}^{T}(\boldsymbol{\theta}_{0},S\mathbf{)}\boldsymbol{A}%
(\boldsymbol{\theta}_{0}\mathbf{)}^{T}\boldsymbol{Z}+\mathrm{o}_{p}%
(\boldsymbol{1}_{\mathrm{card}(S)}),
\end{align*}
where%
\[
\boldsymbol{Q}(\boldsymbol{\theta}_{0},S\mathbf{)}\mathbf{=}-\mathcal{I}%
_{F}^{-1}(\boldsymbol{\theta}_{0})\boldsymbol{R}^{T}(S\mathbf{)}%
\boldsymbol{L}(\boldsymbol{\theta}_{0},S\mathbf{)}\left(  \boldsymbol{R}%
(S\mathbf{)}\mathcal{I}_{F}^{-1}(\boldsymbol{\theta}_{0})\boldsymbol{R}%
^{T}(S\mathbf{)}\right)  ^{-1};
\]
$\bullet$ under $\widetilde{\boldsymbol{\theta}}=\widehat{\boldsymbol{\theta}%
}(S)$ and from (\ref{eq14})%
\begin{align*}
\sqrt{n}\left(  \boldsymbol{R}(S^{C})\widetilde{\boldsymbol{\theta}%
}(S)-\boldsymbol{R}(S^{C})\boldsymbol{\theta}_{0}\right)   &  =\sqrt
{n}\boldsymbol{R}(S^{C})\boldsymbol{\Gamma}(\boldsymbol{\theta}_{0}%
,S\mathbf{)}\frac{1}{\sqrt{n}}\left.  \frac{\partial}{\partial
\boldsymbol{\theta}}\ell(\boldsymbol{N};\boldsymbol{\theta})\right\vert
_{\boldsymbol{\theta}=\boldsymbol{\theta}_{0}}+\mathrm{o}_{p}(\boldsymbol{1}%
_{\mathrm{card}(S^{C})})\\
&  =\boldsymbol{R}(S^{C})\boldsymbol{\Gamma}(\boldsymbol{\theta}%
_{0},S\mathbf{)}\boldsymbol{A}(\boldsymbol{\theta}_{0}\mathbf{)}%
^{T}\boldsymbol{Z}+\mathrm{o}_{p}(\boldsymbol{1}_{\mathrm{card}(S^{C})}).
\end{align*}
That is,%
\begin{align*}
&  \lim_{n_{1},n_{2}\rightarrow\infty}\Pr\left(  S_{\phi}(\boldsymbol{p}%
(\widetilde{\boldsymbol{\theta}}),\boldsymbol{p}(\widehat{\boldsymbol{\theta}%
}))\leq x\right)  =\sum_{S\in\mathcal{F}(E)}\Pr\left(  \boldsymbol{Z}_{3}%
^{T}(S)\boldsymbol{Z}_{3}(S)\leq x,\boldsymbol{Z}_{1}(S)\geq\boldsymbol{0}%
_{\mathrm{card}(S)},\boldsymbol{Z}_{2}(S)\geq\boldsymbol{0}_{\mathrm{card}%
(S^{C})}\right)  \\
&  =\sum_{S\in\mathcal{F}(E)}\Pr\left(  \boldsymbol{Z}_{3}^{T}%
(S)\boldsymbol{Z}_{3}(S)\leq x|\boldsymbol{Z}_{1}(S)\geq\boldsymbol{0}%
_{\mathrm{card}(S)},\boldsymbol{Z}_{2}(S)\geq\boldsymbol{0}_{\mathrm{card}%
(S^{C})}\right)  \Pr\left(  \boldsymbol{Z}_{1}(S)\geq\boldsymbol{0}%
_{\mathrm{card}(S)},\boldsymbol{Z}_{2}(S)\geq\boldsymbol{0}_{\mathrm{card}%
(S^{C})}\right)  \\
&  =\sum_{S\in\mathcal{F}(E)}\Pr\left(  \boldsymbol{Z}_{3}^{T}%
(S)\boldsymbol{Z}_{3}(S)\leq x\left\vert \left(  \boldsymbol{Z}_{1}%
^{T}(S),\boldsymbol{Z}_{2}^{T}(S)\right)  ^{T}\geq\boldsymbol{0}_{J-1}\right.
\right)  \Pr\left(  \boldsymbol{Z}_{1}(S)\geq\boldsymbol{0}_{\mathrm{card}%
(S)},\boldsymbol{Z}_{2}(S)\geq\boldsymbol{0}_{\mathrm{card}(S^{C})}\right)  ,
\end{align*}
where%
\begin{align*}
\boldsymbol{Z}_{3}(S) &  =\boldsymbol{M}_{3}(\boldsymbol{\theta}%
_{0},S\mathbf{)}\boldsymbol{Z,}\qquad\boldsymbol{M}_{3}(\boldsymbol{\theta
}_{0},S\mathbf{)=}\boldsymbol{A}(\boldsymbol{\theta}_{0}\mathbf{)}\left(
\boldsymbol{\Gamma}(\boldsymbol{\theta}_{0},S\mathbf{)}-\boldsymbol{\Gamma
}(\boldsymbol{\theta}_{0}\mathbf{)}\right)  \boldsymbol{A}(\boldsymbol{\theta
}_{0}\mathbf{)}^{T},\\
\boldsymbol{Z}_{1}(S) &  =\boldsymbol{M}_{1}(\boldsymbol{\theta}%
_{0},S\mathbf{)}\boldsymbol{Z},\qquad\boldsymbol{M}_{1}(\boldsymbol{\theta
}_{0},S\mathbf{)=-}\boldsymbol{Q}^{T}(\boldsymbol{\theta}_{0},S\mathbf{)}%
\boldsymbol{A}(\boldsymbol{\theta}_{0}\mathbf{)}^{T},\\
\boldsymbol{Z}_{2}(S) &  =\boldsymbol{M}_{2}(\boldsymbol{\theta}%
_{0},S\mathbf{)}\boldsymbol{Z},\qquad\boldsymbol{M}_{2}(\boldsymbol{\theta
}_{0},S\mathbf{)=}\boldsymbol{R}(S^{C})\boldsymbol{\Gamma}(\boldsymbol{\theta
}_{0},S\mathbf{)}\boldsymbol{A}(\boldsymbol{\theta}_{0}\mathbf{)}^{T}.
\end{align*}
Taking into account that $\boldsymbol{M}_{3}(\boldsymbol{\theta}%
_{0},S\mathbf{)}\boldsymbol{M}_{2}^{T}(\boldsymbol{\theta}_{0},S\mathbf{)=}%
\boldsymbol{M}_{2}^{T}(\boldsymbol{\theta}_{0},S\mathbf{)}$ and
$\boldsymbol{M}_{3}(\boldsymbol{\theta}_{0},S\mathbf{)}\boldsymbol{M}_{1}%
^{T}(\boldsymbol{\theta}_{0},S\mathbf{)=}\boldsymbol{0}_{(J-1)\times
\mathrm{card}(S)}$, by applying the lemma given in Section \ref{LemContrA}%
\[
\Pr\left(  \boldsymbol{Z}_{3}^{T}(S)\boldsymbol{Z}_{3}(S)\leq x\left\vert
\left(  \boldsymbol{Z}_{1}^{T}(S),\boldsymbol{Z}_{2}^{T}(S)\right)  ^{T}%
\geq\boldsymbol{0}_{J-1}\right.  \right)  =\Pr\left(  \chi_{df}^{2}\leq
x\right)
\]
where%
\begin{align*}
df &  =\mathrm{rank}\left(  \boldsymbol{A}(\boldsymbol{\theta}_{0}%
\mathbf{)}\left(  \boldsymbol{\Gamma}(\boldsymbol{\theta}_{0},S\mathbf{)}%
-\boldsymbol{\Gamma}(\boldsymbol{\theta}_{0}\mathbf{)}\right)  \boldsymbol{A}%
(\boldsymbol{\theta}_{0}\mathbf{)}^{T}\right)  =\mathrm{trace}\left(
\boldsymbol{A}(\boldsymbol{\theta}_{0}\mathbf{)}\left(  \boldsymbol{\Gamma
}(\boldsymbol{\theta}_{0},S\mathbf{)}-\boldsymbol{\Gamma}(\boldsymbol{\theta
}_{0}\mathbf{)}\right)  \boldsymbol{A}(\boldsymbol{\theta}_{0}\mathbf{)}%
^{T}\right)  \\
&  =(J-1)-\mathrm{card}(S).
\end{align*}
Finally,%
\begin{align*}
&  \lim_{n_{1},n_{2}\rightarrow\infty}\Pr\left(  S_{\phi}(\boldsymbol{p}%
(\widetilde{\boldsymbol{\theta}}),\boldsymbol{p}(\widehat{\boldsymbol{\theta}%
}))\leq x\right)  \\
&  =\sum_{S\in\mathcal{F}(E)}\Pr\left(  \chi_{(J-1)-\mathrm{card}(S)}^{2}\leq
x\right)  \Pr\left(  \boldsymbol{Z}_{1}(S)\geq\boldsymbol{0}_{\mathrm{card}%
(S)},\boldsymbol{Z}_{2}(S)\geq\boldsymbol{0}_{\mathrm{card}(S^{C})}\right)  \\
&  =\sum_{j=0}^{J-1}\Pr\left(  \chi_{(J-1)-j}^{2}\leq x\right)  \sum
_{S\in\mathcal{F}(E),\mathrm{card}(S)=j}\Pr\left(  \boldsymbol{Z}_{1}%
(S)\geq\boldsymbol{0}_{\mathrm{card}(S)},\boldsymbol{Z}_{2}(S)\geq
\boldsymbol{0}_{\mathrm{card}(S^{C})}\right)  ,
\end{align*}
and since $\boldsymbol{Q}^{T}(\boldsymbol{\theta}_{0},S\mathbf{)}%
\mathcal{I}_{F}(\boldsymbol{\theta}_{0})\boldsymbol{\Gamma}(\boldsymbol{\theta
}_{0},S\mathbf{)=}\boldsymbol{0}_{\mathrm{card}(S)\times(J-1)}$, it holds
$\boldsymbol{M}_{1}(\boldsymbol{\theta}_{0},S\mathbf{)}\boldsymbol{M}_{2}%
^{T}(\boldsymbol{\theta}_{0},S\mathbf{)}=\boldsymbol{0}_{\mathrm{card}%
(S)\times\mathrm{card}(S^{C})}$ which means that $\boldsymbol{Z}_{1}(S)$ and
$\boldsymbol{Z}_{2}(S)$ are independent, that is%
\[
\lim_{n_{1},n_{2}\rightarrow\infty}\Pr\left(  S_{\phi}(\boldsymbol{p}%
(\widetilde{\boldsymbol{\theta}}),\boldsymbol{p}(\widehat{\boldsymbol{\theta}%
}))\leq x\right)  =\sum_{j=0}^{J-1}\Pr\left(  \chi_{(J-1)-j}^{2}\leq x\right)
w_{j}(\boldsymbol{\theta}_{0})
\]
where the expression of $w_{j}(\boldsymbol{\theta}_{0})$ is (\ref{eqw}). We
have also,
\[
\mathrm{Var}(\boldsymbol{Z}_{1}(S))=\boldsymbol{M}_{1}(\boldsymbol{\theta}%
_{0},S\mathbf{)}\boldsymbol{M}_{1}^{T}(\boldsymbol{\theta}_{0},S\mathbf{)}%
=\boldsymbol{Q}^{T}(\boldsymbol{\theta}_{0},S\mathbf{)}\mathcal{I}%
_{F}(\boldsymbol{\theta}_{0})\boldsymbol{Q}(\boldsymbol{\theta}_{0}%
,S\mathbf{)}=\left(  \boldsymbol{R}(S\mathbf{)}\mathcal{I}_{F}^{-1}%
(\boldsymbol{\theta}_{0})\boldsymbol{R}^{T}(S\mathbf{)}\right)  ^{-1}%
=\boldsymbol{H}^{-1}(S,S,\boldsymbol{\theta}_{0}),
\]%
\begin{align*}
\mathrm{Var}(\boldsymbol{Z}_{2}(S)) &  =\boldsymbol{M}_{2}(\boldsymbol{\theta
}_{0},S\mathbf{)}\boldsymbol{M}_{2}^{T}(\boldsymbol{\theta}_{0},S\mathbf{)}%
=\boldsymbol{R}(S^{C})\boldsymbol{\Gamma}(\boldsymbol{\theta}_{0}%
,S\mathbf{)}\mathcal{I}_{F}(\boldsymbol{\theta}_{0})\boldsymbol{\Gamma}%
^{T}(\boldsymbol{\theta}_{0},S\mathbf{)}\boldsymbol{R}^{T}(S^{C}%
)=\boldsymbol{R}(S^{C})\boldsymbol{\Gamma}(\boldsymbol{\theta}_{0}%
,S\mathbf{)}\boldsymbol{R}^{T}(S^{C})\\
&  =\boldsymbol{H}(S^{C},S^{C},\boldsymbol{\theta}_{0})-\boldsymbol{H}%
(S^{C},S,\boldsymbol{\theta}_{0})\boldsymbol{H}^{-1}(S,S,\boldsymbol{\theta
}_{0})\boldsymbol{H}^{T}(S^{C},S,\boldsymbol{\theta}_{0}).
\end{align*}
The proof of $T_{\phi}(\overline{\boldsymbol{p}},\boldsymbol{p}%
(\widetilde{\boldsymbol{\theta}}),\boldsymbol{p}(\widehat{\boldsymbol{\theta}%
}))$ is almost immediate from the proof for $S_{\phi}(\boldsymbol{p}%
(\widetilde{\boldsymbol{\theta}}),\boldsymbol{p}(\widehat{\boldsymbol{\theta}%
}))$ and taking into account that for some $S\in\mathcal{F}(E)$%
\[
T_{\phi}(\overline{\boldsymbol{p}},\boldsymbol{p}%
(\widetilde{\boldsymbol{\theta}}),\boldsymbol{p}(\widehat{\boldsymbol{\theta}%
}))=T_{\phi}(\overline{\boldsymbol{p}},\boldsymbol{p}%
(\widehat{\boldsymbol{\theta}}(S)),\boldsymbol{p}(\widehat{\boldsymbol{\theta
}}))+\mathrm{o}_{p}(1)=S_{\phi}(\boldsymbol{p}(\widetilde{\boldsymbol{\theta}%
}),\boldsymbol{p}(\widehat{\boldsymbol{\theta}})).
\]

\section{Fortran Code: example.f95}
\begin{verbatim}
!--------------------------------------------------------------------------------
! This program is only valid for 2 by 4 contingency tables
! (for other sizes some changes must be done:
! change the value of J and follow the formulas of the weights)
! To run it, the NAG library is required to have installed
! To change the sample go to line 18
! The FORTRAN program generates the outputs in 8 text files
!--------------------------------------------------------------------------------
MODULE ParGlob
INTEGER fail
INTEGER, PARAMETER :: I=2, J=4, nlam=9
DOUBLE PRECISION pr(I*J), W(I*J,I*J-1), RR((I-1)*(J-1),I*(J-1)), betatil(I*(J-1)), &
   pHat(I*J), zz((I-1)*(J-1)), tbt((I-1)*(J-1),(I-1)*(J-1)), bb((I-1)*(J-1),(I-1)*(J-1)), &
   we(0:(I-1)*(J-1)), k1((I-1),(I-1)), k2((J-1),(J-1)), hh((I-1)*(J-1),(I-1)*(J-1)), &
   hInv((I-1)*(J-1),(I-1)*(J-1)), ntt, nu(I), ppi(J), nn(I*J), ppit(I,J), un, sample(I*J),&
   odds(I-1,J-1), nt(I)
DOUBLE PRECISION, PARAMETER:: lamb(nlam)=(/-1.5d0,-1.d0,-0.5d0,0.d0,2.d0/3.d0,1.d0,1.5d0, &
   2.d0,3.d0/),del=0.0d0, pi=3.14159265358979323846264338327950d0, sample=(/11.d0,8.d0,    &
   8.d0,5.d0,6.d0,4.d0,10.d0,12.d0/)
END MODULE ParGlob
!--------------------------------------------------------------------------------

PROGRAM Example
USE ParGlob
IMPLICIT NONE

INTEGER n, m, ifail
DOUBLE PRECISION estT, estS, pval, table(I,J), contT(nlam), contS(nlam), iniTheta(I*J-1), &
   ro(3,2), marg(J), rank(J), wilc0, wilc, meanWilc, sdWilc, pValWilc, g01eaf

DO n=1,I
 DO m=1,J
  ppit(n,m)=(1.d0/3.d0)*((1.d0+n*(m-1.d0)*del)/(1.d0+n*del))
 ENDDO
ENDDO
DO n=1,I-1
 DO m=1,J-1
  odds(n,m)=ppit(n,m)*ppit(n+1,m+1)/(ppit(n+1,m)*ppit(n,m+1))
 ENDDO
ENDDO

marg=sample(1:J)+sample(J+1:2*J)
rank=0.d0
DO n=2,J
 rank(n)=rank(n-1)+marg(n-1)
ENDDO
rank=rank+(marg+1.d0)/2.d0
wilc0=SUM(rank*sample(1:J))
nt(1)=SUM(sample(1:J))
nt(2)=SUM(sample(J+1:2*J))
ntt=SUM(nt)
nu=nt/ntt
meanWilc=nt(1)*(nt(1)+nt(2)+1.d0)/2.d0
sdWilc=nt(1)*nt(2)*(nt(1)+nt(2)+1.d0)/12.d0
sdWilc=sdWilc-nt(1)*nt(2)*SUM(marg**3-marg)/(12.d0*(nt(1)+nt(2))*(nt(1)+nt(2)-1.d0))
sdWilc=SQRT(sdWilc)
wilc=(wilc0-meanWilc)/sdWilc
CALL DesignM()
CALL RestricM()

nn=sample
table=TRANSPOSE(RESHAPE(nn,(/J,I/)))
DO m=1,J
 ppi(m)=SUM(table(:,m))/ntt
ENDDO
iniTheta=0.d0
CALL emvH01(iniTheta)
IF (fail.NE.0) THEN
 iniTheta=0.1d0
 CALL emvH01(iniTheta)
 IF (fail.NE.0) THEN
  iniTheta=-0.1d0
  CALL emvH01(iniTheta)
 ENDIF
ENDIF

21 FORMAT (20F10.4)
22 FORMAT (20F15.10)

OPEN (10, FILE = "theta-Tilde.DAT", action="write",status="replace")
WRITE(10,*) "    **         Theta tilde         ** "
WRITE(10,*) "    --------------------------------- "
WRITE(10,21) (betatil(m), m=1,I*(J-1))
CLOSE(10)

OPEN (10, FILE = "P-Bar.DAT", action="write",status="replace")
WRITE(10,*) "    **    Probability Vector: P-Bar    ** "
WRITE(10,*) "    ------------------------------------- "
WRITE(10,21) (nn(n)/(SUM(nn)), n=1,I*J)
CLOSE(10)

OPEN (10, FILE = "P-theta-Tilde.DAT", action="write",status="replace")
WRITE(10,*) "    **    Probability Vector: P-theta-Tilde    ** "
WRITE(10,*) "    --------------------------------------------- "
WRITE(10,21) (pr(n), n=1,I*J)
CLOSE(10)

CALL ProbVector2(nu,ppi)
OPEN (10, FILE = "P-theta-Hat.DAT", action="write",status="replace")
WRITE(10,*) "    **    Probability Vector: P-theta-Hat    ** "
WRITE(10,*) "    ------------------------------------------- "
WRITE(10,21) (pHat(n), n=1,I*J)
CLOSE(10)

CALL Kmatrices()
CALL hMatrix()

ro(1,1)=hh(1,2)/SQRT(hh(1,1)*hh(2,2))
ro(2,1)=hh(1,3)/SQRT(hh(1,1)*hh(3,3))
ro(3,1)=hh(2,3)/SQRT(hh(2,2)*hh(3,3))
ro(1,2)=(ro(1,1)-ro(2,1))/SQRT((1.d0-ro(2,1)*ro(2,1))*(1.d0-ro(3,1)*ro(3,1)))
ro(2,2)=(ro(2,1)-ro(1,1)*ro(3,1))/SQRT((1.d0-ro(1,1)*ro(1,1))*(1.d0-ro(3,1)*ro(3,1)))
ro(3,2)=(ro(3,1)-ro(2,1)*ro(1,1))/SQRT((1.d0-ro(2,1)*ro(2,1))*(1.d0-ro(1,1)*ro(1,1)))
we(0)=(2.d0*pi-ACOS(ro(1,1))-ACOS(ro(2,1))-ACOS(ro(3,1)))/(4.d0*pi)
we(1)=(3.d0*pi-ACOS(ro(1,2))-ACOS(ro(2,2))-ACOS(ro(3,2)))/(4.d0*pi)
we(2)=0.5d0-we(0)
we(3)=0.5d0-we(1)

ifail=-1
pValWilc=g01eaf('L',wilc,ifail)

OPEN (10, FILE = "T-TESTS.DAT", action="write",status="replace")
WRITE(10,*) "    **     T-test Statistics     ** "
WRITE(10,*) "    -------------------------------- "
WRITE(10,21) (lamb(n), n=1,nlam)
WRITE(10,*) 'test-statistics'
WRITE(10,21) (estT(lamb(n)), n=1,nlam)
WRITE(10,*) 'p-values'
WRITE(10,22) (pval(estT(lamb(n))), n=1,nlam)
WRITE(10,*) "    **     Wilcoxon Statistics     ** "
WRITE(10,*) "    --------------------------------- "
WRITE(10,*) 'test-statistic'
WRITE(10,21) wilc0
WRITE(10,*) 'p-value'
WRITE(10,21) pValWilc
CLOSE(10)

OPEN (10, FILE = "S-TESTS.DAT", action="write",status="replace")
WRITE(10,*) "    **     S-test Statistics     ** "
WRITE(10,*) "    -------------------------------- "
WRITE(10,21) (lamb(n), n=1,nlam)
WRITE(10,*) 'test-statistics'
WRITE(10,21) (estS(lamb(n)), n=1,nlam)
WRITE(10,*) 'p-values'
WRITE(10,22) (pval(estS(lamb(n))), n=1,nlam)
WRITE(10,*) "    **     Wilcoxon Statistics     ** "
WRITE(10,*) "    --------------------------------- "
WRITE(10,*) 'test-statistic'
WRITE(10,21) wilc0
WRITE(10,*) 'p-value'
WRITE(10,21) pValWilc
CLOSE(10)

OPEN (10, FILE = "WEIGHTS.DAT", action="write",status="replace")
WRITE(10,*) "    **     Weights chi-bar     ** "
WRITE(10,*) "    ----------------------------- "
WRITE(10,*) " "
WRITE(10,22) (REAL(we(n)), n=0,(I-1)*(J-1))
WRITE(10,*) "    ---------------------------------------------------------- "
CLOSE(10)

END PROGRAM Example
!--------------------------------------------------------------------------------
! This soubrutine calculates the design matrix of a saturated log-linear model
! with canonical parametrization
!--------------------------------------------------------------------------------
SUBROUTINE DesignM()
USE ParGlob
IMPLICIT NONE

INTEGER h

DOUBLE PRECISION one_I(I), one_J(J), A(I,I-1), B(J,J-1), W12(I*J,(I-1)*(J-1)), &
                 W1(I*J,I-1), W2(I*J,J-1)

one_I=1.d0
one_J=1.d0
A=0.d0
DO h=1,I-1
 A(h,h)=1.d0
ENDDO
B=0.d0
DO h=1,J-1
 B(h,h)=1.d0
ENDDO

CALL Kronecker(I,I-1,A,J,1,one_J,W1)
CALL Kronecker(I,1,one_I,J,J-1,B,W2)
CALL Kronecker(I,I-1,A,J,J-1,B,W12)

W(:,1:I-1)=W1
W(:,I:I+J-2)=W2
W(:,I+J-1:I*J-1)=W12


END SUBROUTINE DesignM
!--------------------------------------------------------------------------------

!--------------------------------------------------------------------------------
! This soubrutines calculates the restriction matrix
!--------------------------------------------------------------------------------
SUBROUTINE RestricM()
USE ParGlob
IMPLICIT NONE

INTEGER h
DOUBLE PRECISION R2((I-1)*(J-1),J-1), R12((I-1)*(J-1),(I-1)*(J-1)), GI(I-1,I-1), &
                 GJ(J-1,J-1)

GI=0.d0
DO h=1,I-1
  GI(h,h)=1.d0
  IF (h.LT.I-1) THEN
    GI(h,h+1)=-1.d0
  ENDIF
ENDDO
GJ=0.d0
DO h=1,J-1
  GJ(h,h)=1.d0
  IF (h.LT.J-1) THEN
    GJ(h,h+1)=-1.d0
  ENDIF
ENDDO
R2 = 0.d0
CALL Kronecker(I-1,I-1,GI,J-1,J-1,GJ,R12)
RR(1:(I-1)*(J-1),1:J-1) = R2
RR(1:(I-1)*(J-1),J:I*(J-1)) = R12

END SUBROUTINE RestricM
!--------------------------------------------------------------------------------

!--------------------------------------------------------------------------------
! Given matrices A and B, this subroutines calculates C as the Kronecker product
! A's dimension n by m
! B's dimension p by q
!--------------------------------------------------------------------------------
SUBROUTINE Kronecker(n,m,A,p,q,B,C)
IMPLICIT NONE

INTEGER n, m, p, q
DOUBLE PRECISION A(n,m), B(p,q), C(n*p,m*q)
INTEGER i, j, k, d

DO i=1,n
 DO j=1,m
  DO k=1,p
   DO d=1,q
    C((i-1)*p+k,(j-1)*q+d) = A(i,j)*B(k,d)
   ENDDO
  ENDDO
 ENDDO
ENDDO

END SUBROUTINE Kronecker
!--------------------------------------------------------------------------------

!--------------------------------------------------------------------------------
! Given
! a) vector theta
! b) the design matrix X=(1,W)
! this subroutine calculates the probabilities of a log-linear model.
!--------------------------------------------------------------------------------
SUBROUTINE ProbVector(beta)
USE ParGlob
IMPLICIT NONE

INTEGER n
DOUBLE PRECISION beta(I*(J-1)), theta(I*J-1), u

theta(I:I*J-1)=beta
u=LOG(nt(I))-LOG(ntt)-LOG(1.d0+SUM(EXP(beta(1:J-1))))
DO n=1,I-1
 theta(n)=LOG(nt(n))-LOG(ntt)-u-LOG(1.d0+SUM(EXP(beta(1:J-1)+&
          beta(n*(J-1)+1:(n+1)*(J-1)))))
ENDDO

pr=EXP(MATMUL(W,theta))*EXP(u)

END SUBROUTINE ProbVector
!--------------------------------------------------------------------------------

!--------------------------------------------------------------------------------
! Subroutine to calculate p(theta-hat)
!--------------------------------------------------------------------------------
SUBROUTINE ProbVector2(nnu,pppi)
USE ParGlob
IMPLICIT NONE

INTEGER h, s
DOUBLE PRECISION nnu(I), pppi(J), aux(I,J)

DO h=1,I
 DO s=1,J
   IF (pppi(s).GT.0.d0) THEN
    aux(h,s)=nnu(h)*pppi(s)
   ELSE
    aux(h,s)=1.d-5
   ENDIF
 ENDDO
ENDDO
pHat=reshape(TRANSPOSE(aux),(/I*J/))


END SUBROUTINE ProbVector2
!--------------------------------------------------------------------------------

!--------------------------------------------------------------------------------
! Subroutine to calculate theta_tilde.
!--------------------------------------------------------------------------------

SUBROUTINE emvH01(x)
USE ParGlob
IMPLICIT NONE

INTEGER, PARAMETER:: n = I*J-1, nclin = (I-1)*(J-1), ncnln = 0, lda = nclin
INTEGER, PARAMETER:: ldcj = 1, ldr = n , liw= 3*n+nclin+2*ncnln, lw=530
INTEGER  iter, ifail, istate(n+nclin+ncnln), iwork(liw), iuser(1), nstate
DOUBLE PRECISION objf, A(nclin,n), user(1), work(lw), R(ldr,n), C(ncnln), CJAC(ldcj,n)
DOUBLE PRECISION clamda(n+nclin+ncnln), bl(n+nclin+ncnln), bu(n+nclin+ncnln), x(n), objgrd(n)
EXTERNAL confun, e04ucf, e04uef, objfun

A=0.d0
A(:,I:I*J-1)=RR
bl(1:n)=-1.d6
bl(n+1:n+nclin)=0.d0
bu=1.d6
ifail = -1
CALL e04uef ('INFINITE BOUND SIZE = 1.e5')
CALL e04uef ('ITERATION LIMIT = 250')
CALL e04uef ('PRINT LEVEL = 0')
CALL e04ucf(n, nclin, ncnln, lda, ldcj, ldr, A, bl, bu, confun, objfun, iter, istate, C,&
  CJAC,clamda,objf, objgrd, R, x, iwork, liw, work, lw, iuser, user, ifail)
betatil=x(I:I*J-1)
fail=ifail
END SUBROUTINE emvH01

SUBROUTINE objfun(mode, n, x, objf, objgrd, nstate, iuser, user)
USE ParGlob
IMPLICIT NONE
INTEGER  mode, n, iuser(1), nstate
DOUBLE PRECISION objf, objgrd(n), x(n), user(1)

CALL ProbVector(x(I:I*(J-1)))
IF (mode .EQ.0 .OR. mode .EQ.2) THEN
  objf =-SUM(nn*LOG(pr))
ENDIF
IF (mode .EQ.1 .OR. mode .EQ.2) THEN
  objgrd=MATMUL(TRANSPOSE(W),SUM(nn)*pr-nn)
ENDIF
END

SUBROUTINE confun (mode, ncnln, g, ldcj, needc, x, c, cjac, nstate, iuser, user)
INTEGER mode, ncnln, g, ldcj, needc(*), nstate, iuser(*)
DOUBLE PRECISION x(*), c(*), cjac(ldcj,*), user(*)

END

!--------------------------------------------------------------------------------
! Subroutine to calculate T-statistic.
!--------------------------------------------------------------------------------

FUNCTION estT(lan)
USE ParGlob
IMPLICIT NONE

DOUBLE PRECISION estT, lan, aux, n
INTEGER h

n=SUM(nn)
aux=0.d0
IF ((lan .GE. -1.d-9) .AND. (lan .LE. 1.d-9)) THEN    !lan=0
 DO h=1,I*J
  IF ((pr(h).GT.0.d0).AND.(pHat(h).GT.0.d0)) THEN
   aux=aux+nn(h)*LOG(pr(h)/pHat(h))
  ENDIF
 ENDDO
 estT=2.d0*aux
ELSE
 IF ((lan .GE. -1.d0-1.d-9) .AND. (lan .LE. -1.d0+1.d-9)) THEN    !lan=-1
  DO h=1,I*J
   IF ((pr(h).GT.0.d0).AND.(pHat(h).GT.0.d0).AND.(nn(h).GT.0.5d0)) THEN
    aux=aux+pHat(h)*LOG((n*pHat(h))/nn(h))
    aux=aux-pr(h)*LOG((n*pr(h))/nn(h))
   ENDIF
  ENDDO
  estT=2.d0*n*aux
 ELSE     !lan<>0, lan<>-1
  DO h=1,I*J
   IF ((pr(h).GT.0.d0).AND.(pHat(h).GT.0.d0).AND.(nn(h).GT.0.5d0)) THEN
    aux=aux+nn(h)*((nn(h)/(n*pHat(h)))**lan-(nn(h)/(n*pr(h)))**lan)
   ENDIF
  ENDDO
  estT=2.d0*aux/(lan*(1.d0+lan))
 ENDIF
ENDIF

END FUNCTION estT

!--------------------------------------------------------------------------------
! Subroutine to calculate S-statistic.
!--------------------------------------------------------------------------------

FUNCTION estS(lan)
USE ParGlob
IMPLICIT NONE

DOUBLE PRECISION estS, lan, aux, n
INTEGER h

n=SUM(nn)
aux=0.d0
IF ((lan .GE. -1.d-9) .AND. (lan .LE. 1.d-9)) THEN    !lan=0
 DO h=1,I*J
  IF ((pr(h).GT.0.d0).AND.(pHat(h).GT.0.d0)) THEN
   aux=aux+pr(h)*LOG(pr(h)/pHat(h))
  ENDIF
 ENDDO
 estS=2.d0*n*aux
ELSE
 IF ((lan .GE. -1.d0-1.d-9) .AND. (lan .LE. -1.d0+1.d-9)) THEN    !lan=-1
  DO h=1,I*J
   IF ((pr(h).GT.0.d0).AND.(pHat(h).GT.0.d0)) THEN
    aux=aux+pHat(h)*LOG(pHat(h)/pr(h))
   ENDIF
  ENDDO
  estS=2.d0*n*aux
 ELSE     !lan<>0, lan<>-1
  DO h=1,I*J
   IF ((pr(h).GT.0.d0).AND.(pHat(h).GT.0.d0)) THEN
    aux=aux+(pr(h)**(lan+1.d0))/(pHat(h)**lan)
   ENDIF
  ENDDO
  estS=2.d0*n*(aux-1.d0)/(lan*(1.d0+lan))
 ENDIF
ENDIF

END FUNCTION estS

!--------------------------------------------------------------------------------
! Subroutine to calculate matrix K.
!--------------------------------------------------------------------------------

SUBROUTINE KMatrices()
USE ParGlob
IMPLICIT NONE

INTEGER n

k1=0.d0
DO n=1,I-1
 k1(n,n)=(nu(n)+nu(n+1))/(nu(n)*nu(n+1))
 IF (n.GE.2) THEN
   k1(n,n-1)=-1.d0/nu(n)
 ENDIF
 IF (n.LE.I-2) THEN
   k1(n,n+1)=-1.d0/nu(n+1)
 ENDIF
ENDDO

k2=0.d0
DO n=1,J-1
 k2(n,n)=(ppi(n)+ppi(n+1))/(ppi(n)*ppi(n+1))
 IF (n.GE.2) THEN
   k2(n,n-1)=-1.d0/ppi(n)
 ENDIF
 IF (n.LE.J-2) THEN
   k2(n,n+1)=-1.d0/ppi(n+1)
 ENDIF
ENDDO

END SUBROUTINE KMatrices

!--------------------------------------------------------------------------------
! Subroutine to calculate matrix H.
!--------------------------------------------------------------------------------

SUBROUTINE HMatrix()
USE ParGlob
IMPLICIT NONE

CALL Kronecker(I-1,I-1,k1,J-1,J-1,k2,hh)

END SUBROUTINE HMatrix


!--------------------------------------------------------------------------------
! Soubrotine to calculate p-values in terms of a specific lambda: T(lam) o S(lam)
!--------------------------------------------------------------------------------

FUNCTION pval(est)
USE ParGlob
IMPLICIT NONE

INTEGER n, ifail
DOUBLE PRECISION pval, est, aux, g01ecf

IF (est.LE.0.d0) THEN
aux=1.d0
ELSE
 aux=0.d0
 DO n=1,(I-1)*(J-1)
  ifail=-1
  aux=aux+g01ecf('U',est,n*1.d0,ifail)*we((I-1)*(J-1)-n)
 ENDDO
 IF (est.LT.0) THEN
  aux=aux+we((I-1)*(J-1))
 ENDIF
ENDIF
pval=aux

END FUNCTION pval
\end{verbatim}

\section{Fortran code: simulation.f95}
\begin{verbatim}
!--------------------------------------------------------------------------------
! This program is only valid for 2 by 3 contingency tables
! (for other sizes some changes must be done:
! change the value of J and follow the formulas of the weights)
! To run it, the NAG library is required to have installed
! The FORTRAN program generates the outputs in several text files
!--------------------------------------------------------------------------------
MODULE ParGlob
INTEGER fail
INTEGER, PARAMETER :: I=2, J=3, nrr=25000, nlam=301
DOUBLE PRECISION pr(I*J), W(I*J,I*J-1), RR((I-1)*(J-1),I*(J-1)), betatil(I*(J-1)), &
  pHat(I*J), zz((I-1)*(J-1)), tbt((I-1)*(J-1),(I-1)*(J-1)), bb((I-1)*(J-1),(I-1)*(J-1)),&
  we(0:(I-1)*(J-1)), k1((I-1),(I-1)), k2((J-1),(J-1)), hh((I-1)*(J-1),(I-1)*(J-1)), &
  hInv((I-1)*(J-1),(I-1)*(J-1)), ntt, nu(I), ppi(J), nn(I*J), ppit(I,J), un,&
  sample(nrr,I*J), odds(I-1,J-1), lamb(nlam)
DOUBLE PRECISION, PARAMETER:: nt(I) = (/16.d0,20.d0/), starting=-1.5d0, ending=3.d0, &
  del=0.d0, pi=3.14159265358979323846264338327950d0
      !if nlam=1, the program only consideres the ending
END MODULE ParGlob
!--------------------------------------------------------------------------------

PROGRAM simulation
USE ParGlob
IMPLICIT NONE

INTEGER n, m, kk, rep, ifail
DOUBLE PRECISION estT, estS, pval, table(I,J), contT(nlam), contS(nlam), iniTheta(I*J-1),&
             marg(J), rank(J), wilc, meanWilc, sdWilc, pValWilc, g01eaf, contW

DO n=1,nlam-1
 lamb(n)=starting+(ending-starting)*(n*1.d0-1.d0)/(nlam*1.d0)
ENDDO
lamb(nlam)=ending
contT=0.d0
contS=0.d0
contW=0.d0
DO n=1,I
 DO m=1,J
  ppit(n,m)=(1.d0/3.d0)*((1.d0+n*(m-1.d0)*del)/(1.d0+n*del))
 ENDDO
ENDDO
DO n=1,I-1
 DO m=1,J-1
  odds(n,m)=ppit(n,m)*ppit(n+1,m+1)/(ppit(n+1,m)*ppit(n,m+1))
 ENDDO
ENDDO
ntt=SUM(nt)
nu=nt/ntt
CALL DesignM()

CALL RestricM()

CALL G05CBF(150)
CALL generaMult()
DO rep=1,nrr
 nn=sample(rep,:)
 DO n=1,I*J
  IF (nn(n).LE.0.d0) THEN
   nn(n)=1.d-5
  ENDIF
 ENDDO
 marg=nn(1:J)+nn(J+1:2*J)
 rank=0.d0
 DO kk=2,J
  rank(kk)=rank(kk-1)+marg(kk-1)
 ENDDO
 rank=rank+(marg+1.d0)/2.d0
 wilc=SUM(rank*nn(1:J))
 meanWilc=nt(1)*(nt(1)+nt(2)+1.d0)/2.d0
 sdWilc=nt(1)*nt(2)*(nt(1)+nt(2)+1.d0)/12.d0
 sdWilc=sdWilc-nt(1)*nt(2)*SUM(marg**3-marg)/(12.d0*(nt(1)+nt(2))*(nt(1)+nt(2)-1.d0))
 sdWilc=SQRT(sdWilc)
 wilc=(wilc-meanWilc)/sdWilc
 ifail=-1
 pValWilc=g01eaf('L',wilc,ifail)
 table=TRANSPOSE(RESHAPE(nn,(/J,I/)))
 DO m=1,J
  ppi(m)=SUM(table(:,m))/ntt
 ENDDO
 iniTheta=0.d0
 CALL emvH01(iniTheta)
 IF (fail.NE.0) THEN
  iniTheta=0.1d0
  CALL emvH01(iniTheta)
  IF (fail.NE.0) THEN
   iniTheta=-0.1d0
   CALL emvH01(iniTheta)
  ENDIF
 ENDIF

 21 FORMAT (20F10.4)
 22 FORMAT (20F15.10)

 CALL ProbVector2(nu,ppi)
 CALL Kmatrices()
 CALL hMatrix()
 we(2)=ACOS(hh(1,2)/SQRT(hh(1,1)*hh(2,2)))/(2.d0*pi)
 we(1)=0.5d0
 we(0)=0.5d0-we(2)

 IF (pValWilc.LE.0.05d0) THEN
  contW=contW+1.d0
 ENDIF
 DO n=1,nlam
  IF (pval(estT(lamb(n))).LE.0.05d0) THEN
   contT(n)=contT(n)+1.d0
  ENDIF
  IF (pval(estS(lamb(n))).LE.0.05d0) THEN
   contS(n)=contS(n)+1.d0
  ENDIF
 ENDDO
ENDDO
 OPEN (10, FILE = "SignLevT-2S.DAT", action="write",status="replace")
 WRITE(10,*) "    **     significance levels for T-test Statistics     ** "
 WRITE(10,*) "    ------------------------------------------------------- "
 DO n=1,nlam
  WRITE(10,21) REAL(lamb(n)),REAL(contT(n)/(nrr*1.d0))
 ENDDO
 CLOSE(10)

 OPEN (10, FILE = "SignLevS-2S.DAT", action="write",status="replace")
 WRITE(10,*) "    **     significance levels for S-test Statistics     ** "
 WRITE(10,*) "    ------------------------------------------------------- "
 DO n=1,nlam
  WRITE(10,21) REAL(lamb(n)),REAL(contS(n)/(nrr*1.d0))
 ENDDO
 CLOSE(10)

 OPEN (10, FILE = "Wilcoxon-2S.DAT", action="write",status="replace")
 WRITE(10,*) "    **     significance level for Wilcoxon Statistics     ** "
 WRITE(10,*) "    ------------------------------------------------------- "
 WRITE(10,*) REAL(contW/(nrr*1.d0))
 CLOSE(10)

END PROGRAM simulation
!--------------------------------------------------------------------------------
! This soubrutine calculates the design matrix of a saturated log-linear model
! with canonical parametrization
!--------------------------------------------------------------------------------
SUBROUTINE DesignM()
USE ParGlob
IMPLICIT NONE
INTEGER h
DOUBLE PRECISION one_I(I), one_J(J), A(I,I-1), B(J,J-1), W12(I*J,(I-1)*(J-1)), &
                 W1(I*J,I-1), W2(I*J,J-1)

ONE_I=1.d0
ONE_J=1.d0
A=0.d0
DO h=1,I-1
 A(h,h)=1.d0
ENDDO
B=0.d0
DO h=1,J-1
 B(h,h)=1.d0
ENDDO

CALL Kronecker(I,I-1,A,J,1,ONE_J,W1)
CALL Kronecker(I,1,ONE_I,J,J-1,B,W2)
CALL Kronecker(I,I-1,A,J,J-1,B,W12)

W(:,1:I-1)=W1
W(:,I:I+J-2)=W2
W(:,I+J-1:I*J-1)=W12


END SUBROUTINE DesignM
!-------------------------------------------------------------------------------


!--------------------------------------------------------------------------------
!--------------------------------------------------------------------------------
! This soubrutines calculates the restriction matrix
!--------------------------------------------------------------------------------
SUBROUTINE RestricM()
USE ParGlob
IMPLICIT NONE
INTEGER h
DOUBLE PRECISION R2((I-1)*(J-1),J-1), R12((I-1)*(J-1),(I-1)*(J-1)), GI(I-1,I-1), &
                 GJ(J-1,J-1)

GI=0.d0
DO h=1,I-1
  GI(h,h)=1.d0
  IF (h.LT.I-1) THEN
    GI(h,h+1)=-1.d0
  ENDIF
ENDDO
GJ=0.d0
DO h=1,J-1
  GJ(h,h)=1.d0
  IF (h.LT.J-1) THEN
    GJ(h,h+1)=-1.d0
  ENDIF
ENDDO
R2 = 0.d0
CALL Kronecker(I-1,I-1,GI,J-1,J-1,GJ,R12)
RR(1:(I-1)*(J-1),1:J-1) = R2
RR(1:(I-1)*(J-1),J:I*(J-1)) = R12

END SUBROUTINE RestricM
!--------------------------------------------------------------------------------
!--------------------------------------------------------------------------------
! Given matrices A and B, this subroutines calculates C as the Kronecker product
! A's dimension n by m
! B's dimension p by q
!--------------------------------------------------------------------------------
SUBROUTINE Kronecker(n,m,A,p,q,B,C)
IMPLICIT NONE

INTEGER n, m, p, q
DOUBLE PRECISION A(n,m), B(p,q), C(n*p,m*q)
INTEGER i, j, k, d

DO i=1,n
 DO j=1,m
  DO k=1,p
   DO d=1,q
    C((i-1)*p+k,(j-1)*q+d) = A(i,j)*B(k,d)
   ENDDO
  ENDDO
 ENDDO
ENDDO

END SUBROUTINE Kronecker
!--------------------------------------------------------------------------------
 !--------------------------------------------------------------------------------
! Given
! a) vector theta
! b) the design matrix X=(1,W)
! this subroutine calculates the probabilities of a log-linear model.
!--------------------------------------------------------------------------------
SUBROUTINE ProbVector(beta)
USE ParGlob
IMPLICIT NONE

INTEGER n
DOUBLE PRECISION beta(I*(J-1)), theta(I*J-1), u

theta(I:I*J-1)=beta
u=LOG(nt(I))-LOG(ntt)-LOG(1.d0+SUM(EXP(beta(1:J-1))))
DO n=1,I-1
 theta(n)=LOG(nt(n))-LOG(ntt)-u  &
         -LOG(1.d0+SUM(EXP(beta(1:J-1)+beta(n*(J-1)+1:(n+1)*(J-1)))))
ENDDO

pr=EXP(MATMUL(W,theta))*EXP(u)

END SUBROUTINE ProbVector
!--------------------------------------------------------------------------------
!--------------------------------------------------------------------------------
! Subroutine to calculate p(theta-hat)
!--------------------------------------------------------------------------------
SUBROUTINE ProbVector2(nnu,pppi)
USE ParGlob
IMPLICIT NONE

INTEGER h, s
DOUBLE PRECISION nnu(I), pppi(J), aux(I,J)

DO h=1,I
 DO s=1,J
   IF (pppi(s).GT.0.d0) THEN
    aux(h,s)=nnu(h)*pppi(s)
   ELSE
    aux(h,s)=1.d-5
   ENDIF
 ENDDO
ENDDO
!Nuestros vectores est\'{a}n en orden lexicogr\'{a}fico, por eso trasponemos
pHat=reshape(TRANSPOSE(aux),(/I*J/))


END SUBROUTINE ProbVector2
!--------------------------------------------------------------------------------
!--------------------------------------------------------------------------------
! Subroutine to calculate theta_tilde.
!--------------------------------------------------------------------------------

SUBROUTINE emvH01(x)
USE ParGlob
IMPLICIT NONE

INTEGER, PARAMETER:: n = I*J-1, nclin = (I-1)*(J-1), ncnln = 0, lda = nclin
INTEGER, PARAMETER:: ldcj = 1, ldr = n , liw= 3*n+nclin+2*ncnln, lw=530
INTEGER  iter, ifail, istate(n+nclin+ncnln), iwork(liw), iuser(1), nstate
DOUBLE PRECISION objf, A(nclin,n), user(1), work(lw), R(ldr,n), C(ncnln), CJAC(ldcj,n)
DOUBLE PRECISION clamda(n+nclin+ncnln), bl(n+nclin+ncnln), bu(n+nclin+ncnln), x(n), &
     objgrd(n)
EXTERNAL confun, e04ucf, e04uef, objfun

A=0.d0
A(:,I:I*J-1)=RR
bl(1:n)=-1.d6
bl(n+1:n+nclin)=0.d0
bu=1.d6
ifail = -1
CALL e04uef ('INFINITE BOUND SIZE = 1.e5')
CALL e04uef ('ITERATION LIMIT = 250')
CALL e04uef ('PRINT LEVEL = 0')
CALL e04ucf(n, nclin, ncnln, lda, ldcj, ldr, A, bl, bu, confun, objfun, iter, istate, C,&
 CJAC, clamda, objf, objgrd, R, x, iwork, liw, work, lw, iuser, user, ifail)
betatil=x(I:I*J-1)
fail=ifail
END SUBROUTINE emvH01

SUBROUTINE objfun(mode, n, x, objf, objgrd, nstate, iuser, user)
USE ParGlob
IMPLICIT NONE
INTEGER  mode, n, iuser(1), nstate
DOUBLE PRECISION objf, objgrd(n), x(n), user(1)

CALL ProbVector(x(I:I*(J-1)))
IF (mode .EQ.0 .OR. mode .EQ.2) THEN
  objf =-SUM(nn*LOG(pr))
ENDIF
IF (mode .EQ.1 .OR. mode .EQ.2) THEN
  objgrd=MATMUL(TRANSPOSE(W),SUM(nn)*pr-nn)
ENDIF
END

SUBROUTINE confun (mode, ncnln, g, ldcj, needc, x, c, cjac, nstate, iuser, user)
INTEGER mode, ncnln, g, ldcj, needc(*), nstate, iuser(*)
DOUBLE PRECISION x(*), c(*), cjac(ldcj,*), user(*)

END



!--------------------------------------------------------------------------------
! Subroutine to calculate T-statistic.
!--------------------------------------------------------------------------------

FUNCTION estT(lan)
USE ParGlob
IMPLICIT NONE

DOUBLE PRECISION estT, lan, aux, n
INTEGER h

n=SUM(nn)
aux=0.d0
IF ((lan .GE. -1.d-9) .AND. (lan .LE. 1.d-9)) THEN    !lan=0
 DO h=1,I*J
  IF ((pr(h).GT.0.d0).AND.(pHat(h).GT.0.d0).AND.(nn(h).GT.0.d0)) THEN
   aux=aux+nn(h)*LOG(pr(h)/pHat(h))
  ENDIF
 ENDDO
 estT=2.d0*aux
ELSE
 IF ((lan .GE. -1.d0-1.d-9) .AND. (lan .LE. -1.d0+1.d-9)) THEN    !lan=-1
  DO h=1,I*J
   IF ((pr(h).GT.0.d0).AND.(pHat(h).GT.0.d0).AND.(nn(h).GT.0.5d0)) THEN
    aux=aux+pHat(h)*LOG((n*pHat(h))/nn(h))
    aux=aux-pr(h)*LOG((n*pr(h))/nn(h))
   ENDIF
  ENDDO
  estT=2.d0*n*aux
 ELSE     !lan<>0, lan<>-1
  DO h=1,I*J
   IF ((pr(h).GT.0.d0).AND.(pHat(h).GT.0.d0).AND.(nn(h).GT.0.5d0)) THEN
    aux=aux+nn(h)*((nn(h)/(n*pHat(h)))**lan-(nn(h)/(n*pr(h)))**lan)
   ENDIF
  ENDDO
  estT=2.d0*aux/(lan*(1.d0+lan))
 ENDIF
ENDIF

END FUNCTION estT


!--------------------------------------------------------------------------------
! Subroutine to calculate S-statistic.
!--------------------------------------------------------------------------------

FUNCTION estS(lan)
USE ParGlob
IMPLICIT NONE

DOUBLE PRECISION estS, lan, aux, n
INTEGER h

n=SUM(nn)
aux=0.d0
IF ((lan .GE. -1.d-9) .AND. (lan .LE. 1.d-9)) THEN    !lan=0
 DO h=1,I*J
  IF ((pr(h).GT.0.d0).AND.(pHat(h).GT.0.d0)) THEN
   aux=aux+pr(h)*LOG(pr(h)/pHat(h))
  ENDIF
 ENDDO
 estS=2.d0*n*aux
ELSE
 IF ((lan .GE. -1.d0-1.d-9) .AND. (lan .LE. -1.d0+1.d-9)) THEN    !lan=-1
  DO h=1,I*J
   IF ((pr(h).GT.0.d0).AND.(pHat(h).GT.0.d0)) THEN
    aux=aux+pHat(h)*LOG(pHat(h)/pr(h))
   ENDIF
  ENDDO
  estS=2.d0*n*aux
 ELSE     !lan<>0, lan<>-1
  DO h=1,I*J
   IF ((pr(h).GT.0.d0).AND.(pHat(h).GT.0.d0)) THEN
    aux=aux+(pr(h)**(lan+1.d0))/(pHat(h)**lan)
   ENDIF
  ENDDO
  estS=2.d0*n*(aux-1.d0)/(lan*(1.d0+lan))
 ENDIF
ENDIF

END FUNCTION estS

!--------------------------------------------------------------------------------
! Subroutine to calculate matrix K.
!--------------------------------------------------------------------------------

SUBROUTINE KMatrices()
USE ParGlob
IMPLICIT NONE

INTEGER n

k1=0.d0
DO n=1,I-1
 k1(n,n)=(nu(n)+nu(n+1))/(nu(n)*nu(n+1))
 IF (n.GE.2) THEN
   k1(n,n-1)=-1.d0/nu(n)
 ENDIF
 IF (n.LE.I-2) THEN
   k1(n,n+1)=-1.d0/nu(n+1)
 ENDIF
ENDDO

k2=0.d0
DO n=1,J-1
 k2(n,n)=(ppi(n)+ppi(n+1))/(ppi(n)*ppi(n+1))
 IF (n.GE.2) THEN
   k2(n,n-1)=-1.d0/ppi(n)
 ENDIF
 IF (n.LE.J-2) THEN
   k2(n,n+1)=-1.d0/ppi(n+1)
 ENDIF
ENDDO

END SUBROUTINE KMatrices

!--------------------------------------------------------------------------------
! Subroutine to calculate matrix H.
!--------------------------------------------------------------------------------
SUBROUTINE HMatrix()
USE ParGlob
IMPLICIT NONE
CALL Kronecker(I-1,I-1,k1,J-1,J-1,k2,hh)
END SUBROUTINE HMatrix
!--------------------------------------------------------------------------------
! Soubrotine to calculate p-values in terms of a specific lambda: T(lam) o S(lam)
!--------------------------------------------------------------------------------

FUNCTION pval(est)
USE ParGlob
IMPLICIT NONE

INTEGER n, ifail
DOUBLE PRECISION pval, est, aux, g01ecf

IF (est.LE.0.d0) THEN
aux=1.d0
ELSE
 aux=0.d0
 DO n=1,(I-1)*(J-1)
  ifail=-1
  aux=aux+g01ecf('U',est,n*1.d0,ifail)*we((I-1)*(J-1)-n)
 ENDDO
 IF (est.LT.0) THEN
  aux=aux+we((I-1)*(J-1))
 ENDIF
ENDIF
pval=aux


END FUNCTION pval

!--------------------------------------------------------------------------------
! Soubrotine to generate Multinomial samples with the parameters specified as
! global parameters (first lines of this program)
!--------------------------------------------------------------------------------

SUBROUTINE generaMult()
USE ParGlob
IMPLICIT NONE

INTEGER n, m, h, s
DOUBLE PRECISION c(I,0:J)
REAL G05CAF

c=0.d0
sample=0.d0
DO n=1,I
 DO h=1,J
  c(n,h)=c(n,h-1)+ppit(n,h)
 ENDDO
ENDDO
DO s=1,nrr
 DO n=1,I
  DO m=1,INT(nt(n))
   un=G05CAF(un)
   h=1
   DOWHILE (.NOT.((un.GE.c(n,h-1)).AND.(un.LT.c(n,h))))
    h=h+1
   ENDDO
   sample(s,(n-1)*J+h)=sample(s,(n-1)*J+h)+1.d0
  ENDDO
 ENDDO
ENDDO

END SUBROUTINE
\end{verbatim}

\end{document}